\documentclass{lmcs}
\pdfoutput=1

\usepackage{lastpage}
\lmcsdoi{15}{1}{27}
\lmcsheading{}{\pageref{LastPage}}{}{}%
{Oct.~10,~2017}{Mar.~14,~2019}{}

\usepackage{booktabs}
\usepackage{adjustbox}


\usepackage{amsfonts}
\usepackage{amsmath}
\usepackage{amssymb}
\usepackage{amsthm}
\usepackage{bussproofs}
\usepackage{graphicx}
\usepackage{mathrsfs}
\usepackage{semantic}
\usepackage{stmaryrd}
\usepackage{tikz-cd}
\usepackage{mathrsfs}
\usepackage{wasysym}
\usepackage{url}
\usepackage{array}
\usepackage{graphicx}

\newcommand{\logicfont}[1]{\textbf{#1}}
\newcommand{\catfont}[1]{\mathrm{#1}}
\newcommand{\algfont}[1]{\mathbb{#1}}
\newcommand{\framefont}[1]{\mathcal{#1}}
\newcommand{\functorfont}[1]{\mathcal{#1}}

\newcommand{\Prop}{\mathrm{Prop}}
\newcommand{\Valuation}{\mathcal{V}}
\newcommand{\pow}[1]{\ensuremath{\mathcal{P}(#1)}}
\newcommand{\filter}[1]{[#1)} 
\newcommand{\ideal}[1]{(#1]} 

\newcommand{\dbrace}[1]{\{\!\{ #1 \}\!\}}

\newcommand{\Atom}[1]{\mbox{\rm #1}}

\newcommand{\gcomp}{\mathop{@}\nolimits}
\newcommand{\gcompE}[2]{#1 \mathop{@_{\mathcal{E}}}\nolimits#2}  

\newcommand*{\myalign}[2]{\multicolumn{1}{#1}{#2}}
\newcommand*{\vtiny}[1]{{\text{\scalebox{.7}{#1}}}}

\newcommand{\gimp}{\ensuremath{\mathop{- \!\!\!\!\! - \!\!\! \blacktriangleright}}}
\newcommand{\limp}{\ensuremath{\mathop{\blacktriangleright \!\!\!\! - \!\!\!\!\! -}}}

\newcommand{\wand}{\mathbin{-\mkern-6mu*}} 
\newcommand{\dnaw}{\mathbin{*\mkern-6mu-}} 
\newcommand\mtop{\top^{*}}
\newcommand{\mor}{%
  \mathbin{\ooalign{$\vee$\cr\hss\raisebox{0.9ex}{\scriptsize $*$}\hss}}} 
\newcommand\lslash{\mathbin{\ooalign{$\slash$\cr
  \hidewidth\raise1ex\hbox{\scriptsize${*}\mkern4mu$}\cr}}} 
\newcommand\rslash{\mathbin{\ooalign{$\backslash$\cr
  \hidewidth\raise1ex\hbox{\scriptsize${*}\mkern-1mu$}\cr}}} 
\newcommand\mbot{\mathop{\ooalign{$\bot$\cr
  \hidewidth\raise0.95ex\hbox{\scriptsize${*}\mkern-2mu$}\cr}}} 
\newcommand\mneg{ \mathbin{\ooalign{$\neg$\cr\hss\raisebox{0.4ex}{\scriptsize $*$}\hss}}}

\newcommand{\septraction}{%
  \mathrel{\mbox{$\hspace*{-0.03em}\mathord{-}\hspace*{-0.66em}
  \mathord{-}\hspace*{-0.155em}\mathord{\circledast}$\hspace*{0.05em}}}}

\newcommand{\Cslash}{\mathbin{\ooalign{$\backslash$\cr
  \hidewidth\raise0.8ex\hbox{\scriptsize${\blacktriangledown}\mkern-0.3mu$}\cr}}}

\newcommand{\Rres}{\ensuremath{\mathop{- \!\!\!\!\! - \!\!\! \bullet}}}
\newcommand{\Lres}{\ensuremath{\mathop{\bullet \!\!\!\! - \!\!\!\!\! -}}}
\newcommand{\Interp}[1]{\llbracket#1\rrbracket}

\newcommand\Parimp{\mathbin{\ooalign{$\setminus$\cr
  \hidewidth\raise1ex\hbox{\scriptsize${*}\mkern-2mu$}\cr}}}
\newcommand\Parbot{\mathbin{\ooalign{$\bot$\cr
  \hidewidth\raise1ex\hbox{\scriptsize${*}\mkern-2mu$}\cr}}}

\newcommand{\PrimePredicate}[2]{\ensuremath{P(#1) = \begin{cases} 1 & \text{if } #2 \\
				0 & \text{otherwise} \end{cases}}}

\newcommand{\realization}[1]{\lfloor#1\rfloor}

\newcommand{\rseq}{\mathbin{-\mkern-2mu\triangleright}} 
\newcommand{\lseq}{\mathbin{\triangleright\mkern-2mu-}} 
\newcommand{\mseq}{\mathbin{\ooalign{$\hidewidth;\hidewidth$\cr$\phantom{\land}$}}}

\begin{document}

\title{Stone-Type Dualities for Separation Logics}

\author[S. Docherty]{Simon Docherty\rsuper{a}}
\address{\lsuper{a}University College London, London, UK}
\email{simon.docherty@ucl.ac.uk}

\author[D. Pym]{David Pym\rsuper{{a, b}}}
\address{\lsuper{b}The Alan Turing Institute, London, UK}
\email{d.pym@ucl.ac.uk}

\begin{abstract}
Stone-type duality theorems, which relate algebraic and relational/topological models, are
important tools in logic because --- in addition to elegant abstraction --- they strengthen
soundness and completeness to a categorical equivalence, yielding a framework through
which both algebraic and topological methods can be brought to bear on a logic. We give a
systematic treatment of Stone-type duality for the structures that interpret bunched logics,
starting with the weakest systems, recovering the familiar \logicfont{BI} and Boolean BI (\logicfont{BBI}), and extending
to both classical and intuitionistic Separation Logic. 
We demonstrate the uniformity and modularity of this analysis by additionally capturing the
bunched logics obtained by extending \logicfont{BI} and \logicfont{BBI} with modalities and multiplicative connectives corresponding
to disjunction, negation and falsum. This includes the logic of separating modalities (\logicfont{LSM}), De Morgan BI (\logicfont{DMBI}), Classical BI (\logicfont{CBI}), and the sub-classical family of logics extending Bi-intuitionistic (B)BI (\logicfont{Bi}(\logicfont{B})\logicfont{BI}). We additionally obtain as corollaries
soundness and completeness theorems for the specific Kripke-style models of these logics
as presented in the literature: for \logicfont{DMBI}, the sub-classical logics extending \logicfont{BiBI} and a new bunched logic, Concurrent Kleene BI (connecting our work to Concurrent Separation Logic), this is the first time soundness and completeness theorems have been proved. We thus obtain a comprehensive semantic account of the multiplicative variants of all standard propositional connectives in the bunched logic setting. This approach synthesises a variety of techniques
from modal, substructural and categorical logic and contextualizes the `resource semantics'
interpretation underpinning Separation Logic amongst them. 
This enables the application of algebraic and topological methods to both Separation
Logic and the systems of bunched logics it is built upon. Conversely, the new notion of
\emph{indexed frame} (generalizing the standard memory model of Separation Logic) and its
associated completeness proof can easily be adapted to other non-classical predicate logics.
\end{abstract}

\maketitle

\section{Introduction}
\subsection{Background}
Bunched logics, beginning with O'Hearn and Pym's \logicfont{BI}~\cite{OP99}, have
proved to be exceptionally useful tools in modelling and reasoning about
computational and information-theoretic phenomena such as resources, the
structure of complex systems, and access control~\cite{LGL,LGL-AC,IJCAR}. Perhaps the most striking example is Separation
Logic~\cite{Reynolds:LICS,YO02} (via Pointer Logic~\cite{IO00}), a specific
theory of first-order (Boolean) BI with primitives for mutable data structures.

\begin{figure}
    \hrule
    \vspace{1mm}
    \centering
    \includegraphics[scale=0.25]{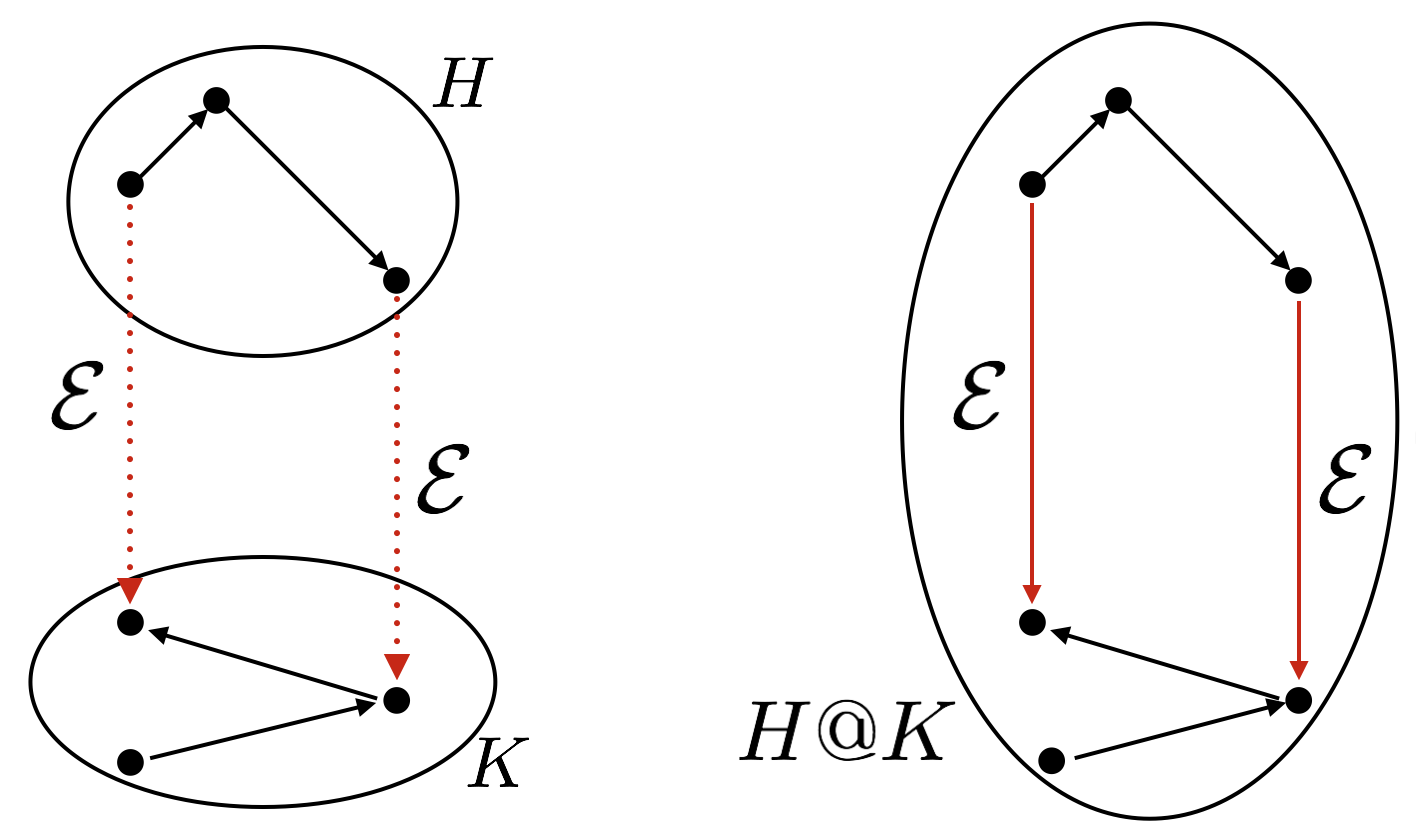}
    \caption{{A layered graph $H \gcomp_\mathcal{E} K$}}
    \vspace{1mm}
    \hrule%
    \label{fig:layered-graph}
\end{figure}

The weakest bunched systems are the so-called layered graph logics~\cite{LGL,IJCAR}.
These logics have a multiplicative conjunction that is
neither associative nor commutative, together with its associated
implications, and additives that may be classical or intuitionistic.
These systems can be used to describe the decomposition of directed graphs into layers (see Fig.~\ref{fig:layered-graph}),
with applications such as complex
systems modelling (e.g.,~\cite{LGL,IJCAR}) and issues in security concerning the
relationship of policies and the systems to which they are intended to
apply (e.g.,~\cite{LGL-AC,IJCAR}). Strengthening the multiplicative conjunction to
be associative and commutative and adding a multiplicative unit yields \logicfont{BI}, for intuitionistic additives,
and Boolean BI (\logicfont{BBI}), for classical additives. A number of extensions of these logics can also be defined which include multiplicative variants on all of the standard propositional connectives~\cite{displayed,CBI,subclassical}. Further extensions
include additive and multiplicative modalities and, with the addition of
parametrization of modalities on actions, Hennessy--Milner-style process 
logics~\cite{CP09,AP2016}. Yet further extensions include additive and
multiplicative epistemic modalities~\cite{LSM,epistemic}, with applications in
security modelling.

All of the applications of bunched logics to reasoning about
computational and informa\-tion-theoretic phenomena essentially rely on
the interpretation of the truth-functional models of these systems
known as \emph{resource semantics}. Truth-functional models of bunched
logics are, essentially, constructed from pre- or partially ordered
partial monoids~\cite{GMP05} which, in resource semantics, are interpreted as
describing how resource-elements can be combined (monoid composition)
and compared (order). The program logic known as \emph{Separation Logic}~\cite{IO00,Reynolds:LICS,YO02}
is a specific theory of first-order bunched logic based on the partial monoid of elements of
the heap (with the order being simply equality). Separation Logic has
found industrial-strength application to static analysis through
Facebook's Infer tool (\url{fbinfer.com}).

Stone's representation theorem for Boolean algebras~\cite{Stone}
establishes that every Boolean algebra is isomorphic to a field of sets. Specifically,
every Boolean algebra $\algfont{A}$ is isomorphic to the algebra of
clopen subsets of its associated \emph{Stone space}~\cite{StoneSpaces}
$S(\algfont{A})$. This result generalizes to a
family of Stone-type duality theorems which establish equivalences
between certain categories of topological spaces and categories of
partially ordered sets.  From the point of view of logic, Stone-type
dualities strengthen the semantic equivalence of truth-functional (such
as \logicfont{BI}'s resource semantics or Kripke's semantics for intuitionistic
logic) and algebraic (such as BI algebras or Heyting
algebras) models to a dual equivalence of categories. This is useful for a number of reasons: on the
one hand, it provides a theoretically convenient abstract characterization of
semantic interpretations and, on the other, it provides a systematic
approach to soundness and completeness theorems, via the close
relationship between the algebraic structures and Hilbert-type proof
systems. Beyond this, Stone-type dualities set up a framework through which techniques
from both algebra and topology can be brought to bear on a logic.

\subsection{Contributions}

In this paper, we give a systematic account of resource semantics via a family of Stone-type
duality theorems that encompass the range of
systems from the layered graph logics, via \logicfont{BI} and \logicfont{BBI}, to
Separation Logic. Our analysis is extended to bunched logics featuring multiplicative variants of all of the standard propositional connectives~\cite{displayed,CBI,subclassical} and --- through straightforward combination with the analogous results from the modal logic literature --- is additionally given for the modal and
epistemic systems extending \logicfont{(B)BI}~\cite{modalbi,LSM,epistemic}. As corollaries we retrieve the soundness and completeness of the standard Kripke models in the literature as well as several new ones.


Soundness and completeness theorems for bunched logics and their extensions tend to be proved through labelled tableaux countermodel procedures~\cite{LSM,epistemic,GMP05,Larchey-Wendling16} that must be specified on a logic-by-logic basis, or by lengthy translations into auxilliary modal logics axiomatized by Sahlqvist~\cite{Sahlqvist1975} formulae~\cite{CBI,subclassical,context}. A notable exception to this (and precursor of the completeness result for \logicfont{(B)BI} given in the present work) is Galmiche \& Larchey-Wendling's completeness result for what they call relational \logicfont{BBI}~\cite{Galmiche2006}. We predict our framework will increase the ease with which completeness theorems can be proved for future bunched logics, as the family of duality theorems are modular in the sense that all that need be verified in an extension of a bunched logic is that the additional axioms and structure that define the extension are also witnessed by the duality. In particular, we give an exhaustive treatment of the structure required for every multiplicative propositional connective in the bunched logic setting as well as the correspondence theory for a selection of axioms that includes all those that define known systems. The only remaining ingredient for a comprehensive treatment of bunched logic semantics is the identification of a Sahlqvist-like class of bunched logic formulas from which Kripke completeness can automatically be obtained. However, the results of the present work lay a clear foundation for such a class to be identified, as the duality theoretic techniques Sambin \& Vaccaro use in their proof of Sahlqvist's theorem~\cite{Sambin1989} should be adaptable given our duality theorems.

Of particular interest in the present paper are bunched logics with intuitionistic additives. In proving completeness for bunched logics via translations into Sahlqvist-axiomatized modal logics, some connectives must be converted into their `diamond-like' De Morgan dual and back using Boolean negation to apply Sahlqvist's theorem. For example, magic wand $\wand$ must be encoded as septraction --- that is, $\varphi \septraction \psi := \neg (\neg \varphi \wand \neg \psi)$
--- in a logically equivalent Sahlqvist-axiomatized modal logic. This is not possible on weaker-than-Boolean bases, which necessarily lack Boolean negation. Our direct proof method side steps this issue and as a result we are able to give the first completeness proofs for the intuitionistic variants of bunched logics that were not amenable to the proof method used for their Boolean counterparts. We further demonstrate the viability of this proof technique for bunched logics by specifying a new logic \logicfont{CKBI}, derived from the interpretation of Concurrent Separation Logic in concurrent Kleene algebra, and prove it sound and complete via duality. More generally, the notion of \emph{indexed frame} (generalizing the standard model of Separation Logic) and its associated completeness proof can easily be adapted to other non-classical predicate logics.


All of the structures given in existing algebraic and relational approaches to
Separation Logic --- including~\cite{Brotherston2014,CalcagnoOY07,DANG2011221,Dockins2009,DongolGS15} --- are instances of the structures utilized in the present work. Thus these approaches are all proved
sound with respect to the standard semantics on store-heap pairs by the results of this paper.
In particular, we strengthen Biering et al.'s~\cite{Biering2005} interpretation of Separation Logic in BI hyperdoctrines to a dual equivalence of categories. To do so we synthesise a variety of related work from modal~\cite{tarski, correspondence,goldblatt}, relevant~\cite{Entailment2,Urquhart1996}, substructural~\cite{Dunn2008} and categorical logic~\cite{Coumans2010}, and thus mathematically some of what follows may be familiar, even if the application to Separation Logic is not. Much of the theory these areas enjoy is produced by way of algebraic and topological arguments. We hope that by recontextualizing the resource semantics of bunched logics in this way similar theory can be given for both Separation Logic and its underlying systems. As a preliminary example (which can be found in the first author's PhD thesis~\cite{DochertyThesis}), we have used the framework given here to prove a bunched logic variant of the Goldblatt-Thomason theorem, which characterises the classes of model that are definable by bunched logic formulae.

\subsection{Structure}
The paper proceeds as described below. In Section~\ref{sec:preliminaries}, we introduce the core bunched logics: the weakest systems \logicfont{LGL} and \logicfont{ILGL}, the resource logics \logicfont{BI} and \logicfont{BBI}, and finally the program logic Separation Logic (in both classical and intuitionistic form). In Section~\ref{sec:layeredgraphlogics}, we define the algebraic and topological structures suitable for interpreting (\logicfont{I})\logicfont{LGL} and give representation and duality theorems relating them. In Section~\ref{sec:bunchedimplications} these results are extended to the logics of bunched implications, (\logicfont{B})\logicfont{BI}. In Section~\ref{sec:separationlogic}, we consider categorical structure appropriate for giving algebraic and truth-functional semantics for first-order (B)BI (\logicfont{FO}(\textbf{B})\textbf{BI}). We recall how \logicfont{FO}(\textbf{B})\textbf{BI} can be interpreted on (B)BI hyperdoctrines and define new structures called \emph{indexed (B)BI frames}. Crucially, we show that the standard models of Separation Logic are instantiations of an indexed frames. We then extend (\logicfont{B})\logicfont{BI} duality to give a dual equivalence of categories between the category of (B)BI hyperdoctrines and the category of indexed (B)BI frames.
In Section~\ref{sec:extension} we extend the theory of the previous sections to bunched logics featuring the full range of multiplicative propositional connectives, as well as normal and separating modal extensions of (\logicfont{B})\logicfont{BI}. In doing so we give new completeness theorems for De Morgan BI and the intuitionistic subclassical bunched logics, and specify a new propositional logic connecting our work to Concurrent Separation Logic, \logicfont{CKBI}, which also fits into our framework. In Section~\ref{sec:conclusions}, we consider possibilities for further work as a result of the duality theorems.

The present work is an extended and expanded version of a conference
paper presented at MFPS XXXIII~\cite{MFPS}. New to this version are the corresponding results for the intuitionistic variants of the logics considered there, as well as the extension to bunched logics with additional modalities and multiplicatives. 
Proofs have been included wherever length permits. A further expansion of this work can be found in the first author's PhD thesis~\cite{DochertyThesis}.

\section{Preliminaries}\label{sec:preliminaries}

In this section we give an introduction to the core bunched logics we consider in this paper: layered graph logics, logics of bunched implications and Separation Logic. In particular, we focus on the intended models of these logics, which will be generalised to classes of frames amenable to our duality theoretic framework in what follows.

\subsection{Layered Graph Logics}\label{subsec:graphlogics}

We begin by presenting the layered graph logics \logicfont{LGL}~\cite{LGL} and \logicfont{ILGL}~\cite{IJCAR}.
First, we give a formal, graph-theoretic definition of layered graph that, we claim,
captures the concept as used in modelling complex systems~\cite{LGL,LGL-AC,IJCAR}.
Informally, two layers in a directed graph are connected by a specified set
of edges, each element of which starts in the upper layer and ends in the lower layer.

Given a directed graph, $\mathcal{G}$, we refer to its \emph{vertex set} and its \emph{edge set} by
$V(\mathcal{G})$ and $E(\mathcal{G})$ respectively, while its set of subgraphs is denoted
$S\!g(\mathcal{G})$, with $H \subseteq \mathcal{G} \text{ iff } H \in S\!g(\mathcal{G})$. For a
\emph{distinguished edge set} $\mathcal{E} \subseteq E(\mathcal{G})$, the \emph{reachability relation}
$\leadsto_{\mathcal{E}}$ on $S\!g(\mathcal{G})$ is defined $H \leadsto_{\mathcal{E}} K$ iff
a vertex of $K$ can be reached from a vertex of $H$ by an $\mathcal{E}$-edge. This generates a partial composition $@_{\mathcal{E}}$ on subgraphs, with $\gcompE{H}{K} \downarrow$ (where $\downarrow$ denotes definedness) iff
$V(H) \cap V(K) = \emptyset, H \leadsto_{\mathcal{E}} K \text{ and } K \not \leadsto_{\mathcal{E}} H$. Output is given by the graph union of the two subgraphs
and the $\mathcal{E}$-edges between them. We say $G$ is  a \emph{layered graph} (with
 respect to $\mathcal{E}$) if there exist $H$, $K$ such that $\gcompE{H}{K}\downarrow$ and
$G = \gcompE{H}{K}$ (see Fig.~\ref{fig:layered-graph}). Layering is evidently neither commutative
nor (because of the definedness condition) associative.

\begin{figure}
\centering
\hrule
\vspace{-3mm}
\setlength\tabcolsep{7pt}
\setlength\extrarowheight{15pt}
\begin{tabular}{lclclcl}
0. & $\cfrac{}{\neg \neg \varphi \vdash \varphi}$   & 1. & $\cfrac{}{\varphi\vdash\varphi}$ & 2. & $\cfrac{}{\varphi\vdash\top}$  \\
3.  & $\cfrac{}{\bot\vdash\varphi}$  & 4. & $\cfrac{\eta \vdash \varphi \quad \eta \vdash \psi}{\eta \vdash \varphi \wedge \psi}$  & 5. & $\cfrac{\varphi \vdash \psi_1 \wedge \psi_2}{\varphi \vdash \psi_i}$ \\
6.  & $\cfrac{\varphi \vdash \psi}{\eta \wedge\varphi \vdash \psi }$  & 7. & $\cfrac{\eta \vdash \psi \quad \varphi \vdash \psi}{\eta \lor \varphi \vdash \psi}$
& 8.  & $\cfrac{\varphi\vdash\psi_i} {\varphi\vdash\psi_1\lor\psi_2}$ \\
9. & $\cfrac{\eta \vdash\varphi\to\psi \quad \eta\vdash\varphi}{\eta\vdash\psi}$ &
10. &  $\cfrac{\eta\wedge\varphi\vdash\psi}{\eta\vdash\varphi\to\psi}$ & 11. & $\cfrac{\xi\vdash\varphi\quad\eta\vdash\psi}{\xi*\eta\vdash\varphi*\psi}$  \\
12. & $\cfrac{\eta*\varphi\vdash\psi}{\eta\vdash\varphi\wand\psi}$ & 13. & $\cfrac{\xi\vdash\varphi\wand\psi\quad\eta\vdash\varphi}{\xi*\eta\vdash\psi}$ &  14. & $\cfrac{\eta*\varphi\vdash\psi}{\varphi\vdash\eta\dnaw\psi}$  \\
& & 15. & $\cfrac{\xi\vdash\varphi\dnaw\psi\quad\eta\vdash\varphi}{\eta*\xi\vdash\psi}$
\end{tabular}
\caption[Hilbert rules for layered graph logics.]{Hilbert rules for layered graph logics. $i=1$ or $2$ for 5.\ and 8.}
\vspace{1mm}
\hrule%
  \label{fig:hilbert-lgl}
\end{figure}
\begin{figure}
\centering
\hrule
\vspace{1mm}
\setlength\tabcolsep{3pt}
\setlength\extrarowheight{2pt}
\begin{tabular}{c c c c l r c c c c r r c}
$G$ & $\vDash$ & $\Atom{p}$ & iff & \myalign{l}{$G \in \Valuation(\Atom{p})$} & &
$G$ & $\vDash$ & $\top$ & & &
 $G  \not\vDash \bot$  \\
$G$ & $\vDash$ & $\varphi \land \psi$ & iff & \myalign{l}{$G \vDash \varphi$ and $G \vDash \psi$} &&
$G$ & $\vDash$ & $\varphi \lor \psi$ & iff & \multicolumn{2}{l}{$G \vDash \varphi$ or} $G \vDash \psi$
\\
$G$ & $\vDash$ & $\varphi \rightarrow \psi$ & iff & \myalign{l}{for all $H \succcurlyeq G$, $H \vDash\varphi$ implies $H$  $\vDash$ $\psi$} \\
$G$ & $\vDash$ & $\varphi * \psi$ & iff & \multicolumn{8}{l}{there exists $H, K$ such that $\gcompE{H}{K}\downarrow$, $\gcompE{H}{K} \preccurlyeq G$, $H \vDash \varphi$ and  $K \vDash \psi$} \\
$G$ & $\vDash$ & $\varphi \wand \psi$ & iff & \multicolumn{8}{l}{for all $H, K$ such that $\gcompE{H}{K}\downarrow$ and $G \preccurlyeq H$,  $H \vDash \varphi$ implies $\gcompE{H}{K} \vDash \psi$} \\
$G$ & $\vDash$ & $\varphi \dnaw \psi$ & iff & \multicolumn{8}{l}{for all $H, K$ such that $\gcompE{K}{H}\downarrow$ and $G \preccurlyeq H$,  $H \vDash \varphi$ implies $\gcompE{K}{H} \vDash \psi$}
\end{tabular}
\caption{Satisfaction on layered graphs for \logicfont{(I)LGL}. \logicfont{LGL} is the case where $\preccurlyeq$ is $=$.}
\vspace{1mm}
\hrule%
\label{fig:sat-LGL}
\end{figure}

Let $\Prop$ be a set of atomic propositions, ranged over by $\mathrm{p}$.
The set of all formulae of $\logicfont{LGL}$ and $\logicfont{ILGL}$ is generated by the following grammar:
\[
\varphi::=\mathrm{p}
\mid \top
\mid \bot
\mid \varphi\wedge\varphi
\mid\varphi\lor\varphi
\mid\varphi\to\varphi
\mid\varphi*\varphi
\mid\varphi\wand\varphi
\mid \varphi\dnaw\varphi .
\]
The connectives above are the standard logical connectives,
together with a (non-commutative and non-associative) multiplicative conjunction, $*$, and its associated implications $\wand$ and $\dnaw$, in the spirit of the Lambek calculus~\cite{Lambek1961, Lambek1993}.
We define $\neg \varphi$ as $\varphi \to \bot$. Hilbert-type systems for the logics are given in Fig.~\ref{fig:hilbert-lgl}: \logicfont{ILGL} is specified by rules 1--15, whilst \logicfont{LGL} is specified by rules 0--15. We note that our notation differs here compared from that found in the work of Collinson et al.~\cite{LGL,LGL-AC} and our previous work~\cite{IJCAR,Docherty2018a}.
There, the multiplicative conjunction is given by $\blacktriangleright$, with the associated implications written $\gimp$ and $\limp$. While this notation has the benefit of directly presenting the non-commutativity of the conjunction, we use the $*$ notation associated with the logic of bunched implications uniformly across all logics we consider. This allows us to uniformly present rules, algebras, and constructions that instantiate the same structure independently of the logic under consideration.

\logicfont{LGL} and \logicfont{ILGL} are interpreted on  \emph{scaffolds}: structures $\mathcal{X} = (\mathcal{G}, \mathcal{E}, X, \preccurlyeq)$ where $\mathcal{G}$ is a directed graph, $\mathcal{E}$ is a distinguished edge set, $X \subseteq S\!g(\mathcal{G})$ is such that --- if $\gcompE{H}{K}\downarrow$ --- $H, K \in X$ iff $\gcompE{H}{K} \in X$ and $\preccurlyeq$ is a preorder on $X$. We note that for any $\mathcal{G}, \mathcal{E}$ and $X$ there are always two ``canonical'' orders one can consider: the subgraph and the supergraph relations.

To model the logic soundly in the intuitionistic case, valuations $\Valuation: \Prop \rightarrow \pow{X}$ (where $\pow{X}$ is the power set of $X$) must be \emph{persistent}: for all $G, H \in X$, if $G \in \Valuation(\Atom{p})$ and $G \preccurlyeq H$, then $H \in \Valuation(\Atom{p})$. This has a spatial interpretation when we consider the order to be the subgraph relation: if \Atom{p} designates that a resource is located in the region $G$, it should also hold of the region $H$ containing $G$.

Given a scaffold $\framefont{X}$ and a persistent valuation
$\Valuation$ the satisfaction relation $\vDash_{\Valuation}$ for \logicfont{ILGL} is inductively defined in Fig.~\ref{fig:sat-LGL}; the particular case where $\preccurlyeq$ is equality yields \logicfont{LGL}. As is necessary for a sound interpretation of intuitionistic logic, persistence extends to all formulas with respect to the semantics of Fig.~\ref{fig:sat-LGL}: for all $\varphi$, $G$, $H$, if $G \vDash_{\Valuation} \varphi$ and $G \preccurlyeq H$ then $H \vDash_{\Valuation} \varphi$. We define validity of $\varphi$ in a model defined on $\mathcal{X}$ to mean for all $G \in X, G \vDash_{\Valuation} \varphi$. Validity of $\varphi$ is defined to be validity in all models. Finally, the entailment relation $\varphi \vDash \psi$ holds iff for all models and all $G$, $G \vDash_{\Valuation} \varphi$ implies $G \vDash_{\Valuation} \psi$.

\subsection{BI and Boolean BI}\label{subsec:bunchedimplications}

Next we present the resource logics \logicfont{BI} and \logicfont{BBI}~\cite{OP99}. Let $\Prop$ be a set of atomic propositions, ranged over by $\Atom{p}$. The set of all formulae of $\logicfont{(B)BI}$ is generated by the following grammar:
\[
\varphi::=\mathrm{p}
\mid \top
\mid \bot
\mid \mtop
\mid \varphi\wedge\varphi
\mid\varphi\lor\varphi
\mid\varphi\to\varphi
\mid\varphi*\varphi
\mid\varphi\wand\varphi.
\]
Once again we have the standard logical connectives, this time joined by a commutative, associative multiplicative conjunction $*$ together with its associated implication $\wand$, and unit $\mtop$. By extending rules 1--15 of Fig.~\ref{fig:hilbert-lgl} with the rules of Fig.~\ref{fig:hilbert-bbi} we obtain a system for \logicfont{BI}; extending rules 0--15 of Fig.~\ref{fig:hilbert-lgl} with the rules of Fig.~\ref{fig:hilbert-bbi} yields a system for \logicfont{BBI}. These rules can be straightforwardly read as enforcing commutativity and associativity of the multiplicative conjunction $*$, the adjointness between $*$ and its associated implication $\wand$, and the fact that $\mtop$ is a unit for $*$.

\begin{figure}
\hrule
    \vspace{1mm}
\setlength\tabcolsep{7pt}
\setlength\extrarowheight{15pt}
\begin{tabular}{lclclc}
16. & $\cfrac{}{(\varphi * \psi) * \xi \vdash \varphi * (\psi * \xi)}$ &
17. & $\cfrac{}{\varphi * \psi \vdash \psi * \varphi}$ &
18. & $\cfrac{}{\varphi * \mtop \dashv \vdash \varphi}$
\end{tabular}

\caption{Rules for the \logicfont{(B)BI} Hilbert Systems, $\mathrm{(B)BI}_{\mathrm{H}}$.}
\vspace{1mm}
\hrule%
\label{fig:hilbert-bbi}
\end{figure}

\begin{figure}
\centering
\hrule
\vspace{1mm}
 \setlength\tabcolsep{3pt}
\setlength\extrarowheight{4pt}
\begin{tabular}{c c c c l r c c c c r r c}
$r$ & $\vDash_{\Valuation}$ & $\Atom{p}$ & iff & \myalign{l}{$r \in \Valuation(\Atom{p})$} & &
$r$ & $\vDash_{\Valuation}$ & $\top$ & & &
 $r  \not\vDash_{\Valuation} \bot$  \\
$r$ & $\vDash_{\Valuation}$ & $\varphi \land \psi$ & iff & \myalign{l}{$r \vDash_{\Valuation} \varphi$ and $r\vDash_{\Valuation} \psi$} &&
$r$ & $\vDash_{\Valuation}$ & $\varphi \lor \psi$ & iff & \multicolumn{2}{l}{$r \vDash_{\Valuation} \varphi$ or} $r \vDash_{\Valuation} \psi$
\\
$r$ & $\vDash_{\Valuation}$ & $\varphi \rightarrow \psi$ & iff & \myalign{l}{for all $r' \succcurlyeq r$, $r' \vDash_{\Valuation}\varphi$ implies $r'$  $\vDash_{\Valuation}$ $\psi$;} && $r$ & $\vDash_{\Valuation}$ & $\mtop$ & iff & \myalign{l}{$r \in E$} \\
$r$ & $\vDash_{\Valuation}$ & $\varphi * \psi$ & iff & \multicolumn{8}{l}{there exists $r', r''$ such that  $r \in r' \circ r''$, $r' \vDash_{\Valuation} \varphi$ and  $r'' \vDash_{\Valuation} \psi$} \\
$r$ & $\vDash_{\Valuation}$ & $\varphi \wand \psi$ & iff & \multicolumn{8}{l}{for all $r', r''$ such that $r'' \in r \circ r'$, $r' \vDash_{\Valuation} \varphi$ implies $r'' \vDash_{\Valuation} \psi$} \\
\end{tabular}
\caption{Satisfaction on UDMFs for \logicfont{(B)BI}. \logicfont{BBI} is the case where $\preccurlyeq$ is =.}
\vspace{1mm}
\hrule%
\label{fig:sat-BBI}
\end{figure}

Models of these logics arise from the analysis of the abstract notion of resource first outlined in the original work on \logicfont{BI}~\cite{OP99}. There, the principal properties were determined to be composability and comparison, formalized by a commutative and associative composition $\circ$ with unit $e$, and a pre- or partial order $\preccurlyeq$, respectively. The interplay between composition and comparison in this analysis is formalized by a \emph{bifunctorality} condition: if $a \preccurlyeq a'$ and $b \preccurlyeq b'$ then $a \circ a' \preccurlyeq b \circ b'$. Structures $\mathcal{M} = (M, \preccurlyeq, \circ, e)$ satisfying these conditions are known as \emph{resource monoids}.

\logicfont{(B)BI} is interpreted on ordered monoidal structures that generalize the notion of resource monoid in various ways. For example, $\circ$ can be generalized to be non-deterministic and/or partial, the unit $e$ can be generalized to a set of units $E$, and the bifunctorality condition can be generalized to a number of different coherence conditions between $\circ$ and $\preccurlyeq$. The majority of these choices define the same notion of validity, although there are some sharp boundaries (an analysis of some of these issues can be found in~\cite{Larchey-Wendling2014}), and all of them can be grouped under the name \emph{resource semantics}. As an example we give a class of models that corresponds closely to the standard intuitionistic model of Separation Logic. This class will be further generalized in Section~\ref{sec:bunchedimplications} for the purposes of the duality theorem.

A \emph{monoidal frame} is a structure $\mathcal{X} = (X, \preccurlyeq, \circ, E)$ s.t. $\preccurlyeq$ is a preorder, $\circ: X^2 \rightarrow \mathcal{P}(X)$ is commutative operation satisfying \emph{non-deterministic associativity}
\[ \forall x, y, z, t, s: s \in x \circ y \text{ and } t\in s \circ z \text{ implies } \exists s'(s' \in y \circ z \text{ and } t \in x \circ s),  \] and $E$ is a set satisfying, for all $x, y \in X$, the conditions
(Unit Existence) $\exists e \in E(x \in x \circ e)$; (Coherence) $x \in y \circ e \land e \in E \rightarrow y \preccurlyeq x$; and (Closure) $e \in E \land e \preccurlyeq e' \rightarrow e' \in E$.
It is \emph{downwards-closed} if, whenever $x \in y \circ z$, $y' \preccurlyeq y$ and $z' \preccurlyeq z$, there exists $x' \preccurlyeq x$ such that $x' \in y' \circ z'$. Conversely it is \emph{upwards-closed} if, whenever $x \in y \circ z$ and $x \preccurlyeq x'$, there exists $y' \succcurlyeq y$ and $z' \succcurlyeq z$ such that $x' \in y' \circ z'$.

The notions of upwards and downwards closure are properties that the majority of models of Separation Logic satisfy. Cao et al.~\cite{CCA17} show that when a model satisfies both upwards and downwards closure, the semantic clauses of $*$ and $\wand$ can be given identically in the classical and intuitionistic cases.

Given an upwards and downwards closed monoidal frame (UDMF) $\mathcal{X}$ and a persistent valuation $\mathcal{V}: \Prop \rightarrow \mathcal{P}(\mathrm{X})$, the satisfaction relation $\vDash_{\Valuation}$ is inductively defined in Fig.~\ref{fig:sat-BBI}. Once again, a semantics for \logicfont{BBI} is obtained as the particular case where the order $\preccurlyeq$ is equality $=$. As with $\logicfont{ILGL}$, persistence extends to all formulas of $\logicfont{BI}$, with upwards and downwards closure ensuring that this is the case for formulas of the form $\varphi * \psi$ and $\varphi \wand \psi$. Validity and entailment are defined in much the same way as the case for \logicfont{(I)LGL}.

\subsection{Separation Logic}\label{subsec:separationlogic}

Separation Logic~\cite{Primer}, introduced by Ishtiaq and O'Hearn~\cite{IO00}, and Reynolds~\cite{Reynolds:LICS}, is an extension of
Hoare's program logic which addresses reasoning about programs that access and mutate shared data structures. The usual
presentation of Separation Logic is based on Hoare triples --- for reasoning about the state of imperative programs ---
of the form $ \{ \,\varphi \, \} \, C \, \{ \,\psi \, \}$,
where $C$ is a program command, $\varphi$ is pre-condition for $C$, and $\psi$ is a post-condition for $C$. The formulas $\varphi$ and $\psi$ are given by the following grammar:
\[
\varphi::= E = E'
\mid E \mapsto F
\mid \top
\mid \bot
\mid \mathrm{Emp}
\mid \varphi\wedge\varphi
\mid\varphi\lor\varphi
\mid\varphi\to\varphi
\mid\varphi*\varphi
\mid\varphi\wand\varphi
\mid \exists v. \varphi
\mid \forall v. \varphi.
\]
Reynolds' programming language is a simple language of commands with a Lisp-like set-up for creating and accessing cons
cells: $C \, ::= \, x := E \mid x := E.i \mid E.i := E' \mid x := cons(E_1,E_2) \mid \dots \,$. Here the expressions $E$
of the language are built up using booleans, variables, etc., $cons$ cells, and atomic expressions.
Separation Logic thus facilitates verification procedures for programs that alter the heap.

A key feature of Separation Logic is the local reasoning provided by the Frame Rule,
\[\frac{\{ \,\varphi \, \} \, C \, \{ \,\psi \, \} }{\{ \,\varphi \ast \chi\, \} \, C \, \{ \,\psi \ast \chi\, \} },\]
where $\chi$ does not include any free variables modified by the program $C$. Static analysis procedures based on the
Frame Rule form the basis of Facebook's Infer tool (\url{fbinfer.com}) that is deployed in its code production. The decomposition of
the analysis that is facilitated by the Frame Rule is critical to the practical deployability of Infer.

Separation Logic can usefully and safely be seen (see~\cite{YO02} for the details) as a presentation of (B)BI Pointer
Logic~\cite{IO00}. The semantics of (B)BI Pointer Logic, a theory of (first-order) \logicfont{(B)BI}, is an instance of \logicfont{(B)BI}'s resource
semantics in which the monoid of resources is constructed from the program's heap. In detail, this model
has two components, the store and the heap. The store is a partial function mapping from variables
to values $a \in \mbox{\rm Val}$ (e.g., integers) and the heap is a partial function from natural numbers to values.
In logic, the store is often called the valuation, and the heap is a possible world.
In programming languages, the store is sometimes called the environment. Within this set-up, the atomic
formulae of (B)BI Pointer Logic include equality between expressions, $E = E'$, and, crucially, the points-to predicate,
$E \mapsto F$.

We use the following additional notation: $dom(h)$ denotes the domain of definition of a heap $h$ and
$dom(s)$ is the domain of  a store $s$; $h \# h'$ denotes that $dom(h) \cap dom(h') = \emptyset$;
$h \cdot h'$ denotes the union of functions with disjoint domains, which is undefined if the domains
overlap; $h \sqsubseteq h'$ denotes that the graph of $h$ is a subgraph of $h'$; $[]$ denotes the \emph{empty heap} that is nowhere defined; $[ f \mid v \mapsto a ]$ is the partial function that is equal to $f$ except that $v$ maps to $a$;
expressions $E$ are built up from variables and constants, and so determine denotations
$\dbrace{E}s \in {\rm Val}$.

With this basic data we can define the satisfaction relations for BI and BBI pointer logic. BI pointer logic is given by extending the \emph{intuitionistic} semantic clause for $\mapsto$,
\[ s, h \vDash E \!\mapsto\! F{\ } \mbox{\rm iff } \dbrace{E}s \in dom(h) \mbox{\rm and } h(\dbrace{E}s) = \dbrace{F}s, \] 
with those in Fig.~\ref{fig:sat-PL}; similarly, BBI pointer logic is given by extending the \emph{classical} semantic clause for $\mapsto$,
\[ s , h \vDash E \!\mapsto\! F{\  } \mbox{\rm  iff }  \dbrace{E}s = {\rm dom}(h) \mbox{\rm and } h(\dbrace{E}s) = \dbrace{F}s, \] 
with those in Fig.~\ref{fig:sat-PL}, where $\sqsubseteq$ is replaced for $=$ in the clauses for $\rightarrow$ and $\mathrm{Emp}$. The judgement, $s, h \vDash \varphi$, says that the assertion $\varphi$ holds for a given store and heap, assuming that the free variables of $\varphi$ are contained in the domain of $s$. The classical interpretation of $\mapsto$ requires $E$ to be the only active address in the current
heap, whereas the intuitionistic interpretation is the weaker judgement that $E$ is \emph{at least} one of the active addresses in the current heap.

\begin{figure}
\centering
\hrule
\vspace{1mm}
 \setlength\tabcolsep{2pt}
\setlength\extrarowheight{4pt}
\begin{tabular}{c c l c l r c c c c r r c}
$s, h$ & $\vDash$ & $E = E'$ & iff & \myalign{l}{$\dbrace{E}s = \dbrace{E'}s$} & &
$s,h$ & $\vDash$ & $\top$ & &&
 $s, h  \not\vDash \bot$  \\
$s, h$ & $\vDash$ & $\varphi \land \psi$ & iff & \myalign{l}{$s, h \vDash \varphi$ and $s, h \vDash \psi$} &&
$s, h$ & $\vDash$ & $\varphi \lor \psi$ & iff & \multicolumn{2}{l}{$s, h \vDash \varphi$ or} $s, h \vDash \psi$
\\
$s, h$ & $\vDash$ & $\varphi \rightarrow \psi$ & iff & \myalign{l}{for all $h' \sqsupseteq h$, $h'$ $\vDash$ $\varphi$ implies $h'$  $\vDash$ $\psi$} && $s, h$ & $\vDash$ & $\Atom{Emp}$ & iff & \myalign{l}{$h \sqsupseteq []$} \\
$s, h$ & $\vDash$ & $\varphi * \psi$ & iff & \multicolumn{8}{l}{there exists $h', h''$ s.t. $h\# h'$, $h = h' \cdot h''$, $s, h' \vDash \varphi$ and  $s, h'' \vDash \psi$} \\
$s, h$ & $\vDash$ & $\varphi \wand\psi$ & iff & \multicolumn{8}{l}{for all $h'$ such that $h \# h',  s, h' \vDash \varphi$ implies $s, h \cdot h' \vDash \psi$} \\
$s, h$ & $\vDash$ & $\exists v. \varphi$ & iff & \multicolumn{8}{l}{there exists $a \in \mathrm{Val}$,  $[s \mid v \mapsto a], h \vDash \varphi$} \\
$s, h$ & $\vDash$ & $\forall v. \varphi$ & iff & \multicolumn{8}{l}{for all $a \in \mathrm{Val}$,  $[s \mid v \mapsto a], h \vDash \varphi$} \\
\end{tabular}
\caption{Satisfaction for (B)BI Pointer Logic. The BBI variant replaces $\sqsupseteq$ with =.}
\vspace{1mm}
\hrule%
\label{fig:sat-PL}
\end{figure}

The technical reason for this difference is the requirement of persistence with respect to the heap extension ordering $\sqsubseteq$ for soundness of the intuitionistic semantics.  
This leads to some quirks in the intuitionistic model: for example, the multiplicative unit $\mathrm{Emp}$ collapses to $\top$ as every heap extends the empty heap. Some aspects remain the same: as heaps are upwards and downwards closed (in the sense defined in Section~\ref{subsec:bunchedimplications}) the semantic clauses for $*$ and $\wand$ can be given identically. Most importantly, in both cases descriptions of larger heaps can be built up using $\ast$, and this coheres with the local reasoning provided by the Frame Rule.

When Separation Logic is given as BI pointer logic it is known as \emph{Intuitionistic Separation Logic}, whereas the formulation as BBI pointer logic is known as \emph{Classical Separation Logic}. Although both were defined in the paper introducing Separation Logic as a theory of bunched logic~\cite{IO00}, the classical variant has taken precedence in both theoretical work and practical implementations. A key exception to this is the higher-order Concurrent Separation Logic framework IRIS, which is based on Intuitionistic Separation Logic~\cite{IRIS}. IRIS requires an underlying intuitionistic logic as it utilizes the `later' modality $\triangleright$~\cite{later}, which collapses to triviality ($\triangleright \varphi \leftrightarrow \top$ for all $\varphi$) in the presence of the law of excluded middle.

\section{Layered Graph Logics}\label{sec:propintduality}\label{sec:layeredgraphlogics}

\subsection{Algebra and Frames for (I)LGL}
We begin our analysis with the weakest systems, \logicfont{LGL} and \logicfont{ILGL}. Each of the logics we consider can be obtained by extending the basic structures associated with these logics and so we are able to systematically extend the theory in each case by accounting for just the extensions to the structure. First, we consider lattice-based algebras suitable for interpreting (\logicfont{I})\logicfont{LGL}.

\begin{defi}[(I)LGL Algebra]\hfill
\begin{enumerate}
\item An \emph{ILGL algebra} is an algebra $\mathbb{A} = (A, \land, \lor, \rightarrow, \top, \bot, *, \wand, \dnaw)$
such that $(A, \land, \lor, \rightarrow, \top, \bot)$ is a Heyting algebra and $*, \wand, \dnaw$ are binary operations on $A$ satisfying, for all $a, b, c \in A$, \[a * b \leq c \text{ iff } a \leq b \wand c \text{ iff } b \leq a \dnaw c.\]
\item A \emph{LGL algebra} is an ILGL algebra $\mathbb{A} = (A, \land, \lor, \rightarrow, \top, \bot, *, \wand, \dnaw)$ for which $(A, \land, \lor, \rightarrow, \top, \bot)$ is a Boolean algebra. \qed%
\end{enumerate}
\end{defi}

\noindent
Residuation of $*$, $\wand$ and $\dnaw$ with respect to the underlying lattice order entails a number of useful properties that are utilised in what follows.

\begin{prop}[cf.~\cite{Jipsen}]\label{prop-alg-prop}
Let $\mathbb{A}$ be an (I)LGL algebra. Then, for all $a, b, a', b' \in A$ and $X, Y \subseteq A$,  we have the following:
\begin{enumerate}
\item If $a \leq a'$ and $b \leq b'$ then $a * b \leq a' * b'$;
\item If $\bigvee X$ and $\bigvee Y$ exist then $\bigvee_{x \in X, y \in Y} x * y$ exists and
	$(\bigvee X) * (\bigvee Y) =  \bigvee_{x \in X, y \in Y} x * y$;
\item If $a = \bot$ or $b = \bot$ then $a * b = \bot$;
\item If $\bigvee X$ exists then for any $z \in A$, $\bigwedge_{x \in X} (x \wand z)$ and $\bigwedge_{x \in X} (x 	\dnaw z)$ exist with \[\bigwedge_{x \in X} (x \wand z) = (\bigvee X) \wand z \text{ and } \bigwedge_{x \in X} (x 			\dnaw z) = (\bigvee X) \dnaw z;\]
\item If $\bigwedge X$ exists then for any $z \in A$, $\bigwedge_{x \in X} (z \wand x)$ and
	$\bigwedge_{x \in X} (z \dnaw x)$ exist with \[ \bigwedge_{x \in X} (z \wand x) = z \wand (\bigwedge X)
	\text{ and } \bigwedge_{x \in X} (z \dnaw x) = z \dnaw (\bigwedge X); \;\mbox{\rm and}\]
\item $a \wand \top = a \dnaw \top = \bot \wand a =  \bot \dnaw a = \top$. \qed%
\end{enumerate}
\end{prop}

\noindent
Interpretations of \logicfont{(I)LGL} on an (I)LGL algebra work as follows: let $\Valuation: \Prop \rightarrow A$ be an assignment of propositional variables to elements of the algebra; this is uniquely extended to an interpretation $\llbracket - \rrbracket$ of every \logicfont{ILGL} formula by induction, with $\Interp{\Atom{p}} = \Valuation(\Atom{p}), \Interp{\top} = \top$ and $\Interp{\bot} = \bot$ as base cases:
\[
\begin{array}{llllllll}
\llbracket \varphi \land \psi \rrbracket &= \llbracket \varphi \rrbracket \land \llbracket \psi \rrbracket  			&\llbracket \varphi \lor \psi \rrbracket &= \llbracket \varphi \rrbracket \lor \llbracket \psi \rrbracket &
\llbracket \varphi \rightarrow \psi \rrbracket & = \llbracket \varphi \rrbracket \rightarrow \llbracket \psi \rrbracket \\
  			\Interp{\varphi * \psi} & = \Interp{\varphi} * \Interp{\psi}
 			&\Interp{\varphi \wand \psi} & = \Interp{\varphi} \wand \Interp{\psi}
 &\Interp{\varphi \dnaw \psi} & = \Interp{\varphi} \dnaw \Interp{\psi}.\\
\end{array}
\]

An interpretation $\Interp{-}$ on a(n) (I)LGL algebra $\mathbb{A}$ \emph{satisfies} $\varphi$ iff $\Interp{\varphi} = \top$; $\varphi$ is valid on (I)LGL algebras iff $\varphi$ is satisfied by every interpretation $\Interp{-}$ on every ILGL algebra $\mathbb{A}$. By constructing Lindenbaum-Tarski algebras from the Hilbert systems of Fig.~\ref{fig:hilbert-lgl} we obtain the following soundness and completeness theorem for (I)LGL algebras.

\begin{thm}[Algebraic Soundness \& Completeness]\label{thm:algilglcomplete}
For all \logicfont{(I)LGL} formulas $\varphi$, $\psi$, $\varphi \vdash \psi$ is provable in $\mathrm{(I)LGL}_{\mathrm{H}}$ iff, for all (I)LGL algebras $\mathbb{A}$ and all interpretations $\Interp{-}$ on $\mathbb{A}$, $\llbracket \varphi \rrbracket \leq \Interp{\psi}$. \qed%
\end{thm}

Next we generalize the intended model of the logics to a class of relational structures.

\begin{defi}[(I)LGL Frame]\label{defn:ilglframe}
An \emph{ILGL frame} is a triple $\mathcal{X} = (X, \preccurlyeq, \circ)$ where $X$ is a set, $\preccurlyeq$ a preorder on $X$ and $\circ: X^2 \rightarrow \pow{X}$ a binary operation. An \emph{LGL frame} is an ILGL frame for which the order $\preccurlyeq$ is equality $=$. \qed%
\end{defi}

Here $\circ$ is a generalization of Section~\ref{subsec:graphlogics}'s $\mathop{@}\nolimits$. We can therefore reconfigure the semantics given in Fig.~\ref{fig:sat-LGL} to give a satisfaction relation on (I)LGL frames by straightforward substitutions of $\circ$ (see Fig.~\ref{fig:sat-LGL-frame}). Partiality is encoded by the fact $\circ$ has codomain $\pow{X}$ and thus $x \circ y = \emptyset$ holds when $\circ$ is undefined for $x, y$. However, $\circ$ is strictly more general as it does not necessarily satisfy the partial determinism that holds of scaffolds: $G =\gcompE{H}{K}$ and $G' = \gcompE{H}{K}$ implies $G = G'$. As $\circ$ is nondeterministic, it can equivalently be seen as a ternary relation, though we maintain the partial function notation to emphasise its interpretation as a composition operator.

\begin{figure}
\centering
\hrule
\vspace{1mm}
\setlength\tabcolsep{3pt}
\setlength\extrarowheight{2pt}
\begin{tabular}{c c c c l r c c c c r r c}
$x$ & $\vDash_{\mathcal{V}}$ & $\mathrm{p}$ & iff & \myalign{l}{$x \in \mathcal{V}(\mathrm{p})$} \\
$x$ & $\vDash_{\mathcal{V}}$ & $\top$ & & always & \\
$x$ &  $\vDash_{\mathcal{V}}$ & $\bot$ & & never & \\
$x$ & $\vDash_{\mathcal{V}}$ & $\varphi \land \psi$ & iff & \myalign{l}{$x \vDash_{\mathcal{V}} \varphi$ and $x \vDash_{\mathcal{V}} \psi$} \\
$x$ & $\vDash_{\mathcal{V}}$ & $\varphi \lor \psi$ & iff & $x \vDash_{\mathcal{V}} \varphi$ or $x \vDash_{\mathcal{V}} \psi$ \\
$x$ & $\vDash_{\mathcal{V}}$ & $\varphi \rightarrow \psi$ & iff & \myalign{l}{for all $y \succcurlyeq x$, $y \vDash_{\mathcal{V}}\varphi$ implies $y$  $\vDash_{\mathcal{V}}$ $\psi$} \\
$x$ & $\vDash_{\mathcal{V}}$ & $\varphi * \psi$ & iff & \multicolumn{8}{l}{there exists $x', y, z$ s.t. $x \succcurlyeq x' \in y \circ z$, $y \vDash_{\mathcal{V}} \varphi$ and  $z \vDash_{\mathcal{V}} \psi$} \\
$x$ & $\vDash_{\mathcal{V}}$ & $\varphi \wand \psi$ & iff & \multicolumn{8}{l}{for all $x', y, z$ s.t. $x' \succcurlyeq x$ and $z \in x' \circ y$: $y \vDash_{\mathcal{V}} \varphi$ implies $z \vDash_{\mathcal{V}} \psi$} \\
$x$ & $\vDash_{\mathcal{V}}$ & $\varphi \dnaw \psi$ & iff & \multicolumn{8}{l}{for all $x', y, z$ s.t. $x' \succcurlyeq x$ and $z \in y \circ x'$:  $y \vDash_{\mathcal{V}} \varphi$ implies $z \vDash_{\mathcal{V}} \psi$}
\end{tabular}
\caption[Satisfaction for (I)LGL.]{Satisfaction for (I)LGL\@. LGL is given by the case where $\succcurlyeq$ is $=$.}
\vspace{1mm}
\hrule%
\label{fig:sat-LGL-frame}
\end{figure}

The structure of what follows will be reflected across each bunched logic we consider. First, we equip the frames associated with the logic with an appropriate notion of morphism to obtain a category. Next we set up the dual functors $Com$ and $Pr$ for transforming algebras into frames (and vice versa) and homomorphisms into frame morphisms (and vice versa), and prove a representation theorem falls out of this relationship. Finally we add appropriate topological structure to the frames to obtain a category of topological spaces that is dually equivalent to the category of algebras. 

(I)LGL frames can be equipped a notion of morphism to obtain categories $\mathrm{ILGLFr}$ and $\mathrm{LGLFr}$. As in modal logic, morphisms have the following connection with the logics: if there exists a surjective (I)LGL morphism $g: \mathcal{X} \rightarrow \mathcal{X}'$ then any \logicfont{(I)LGL} formula valid in $\mathcal{X}$ is also valid in $\mathcal{X}'$.

\begin{defi}[ILGL Morphism]\label{defn:ilglmorphism}
Given ILGL frames $\mathcal{X}$ and $\mathcal{X}'$, an \emph{ILGL morphism} is a map $g: X \rightarrow X'$ satisfying
\begin{enumerate}
\item $x \preccurlyeq y$ implies $g(x) \preccurlyeq' g(y)$,
\item $g(x) \preccurlyeq' y'$ implies there exists $y \in X$ s.t. $x \preccurlyeq y$ and $g(y) = y'$,
\item $x \in y \circ z$ implies $g(x) \in g(y) \circ' g(z)$,
\item $w' \preccurlyeq' g(x)$ and $w' \in y' \circ' z'$ implies there exists $w, y, z \in X$ s.t. $w \preccurlyeq x$, $w \in y \circ z$, $y' \preccurlyeq' g(y)$ and $z' \preccurlyeq' g(z)$,
\item $g(x) \preccurlyeq' w'$ and $z' \in w' \circ' y'$ implies there exists $w, y, z \in X$ s.t. $x \preccurlyeq w$, $z \in w \circ y$, $y' \preccurlyeq' g(y)$ and $g(z) \preccurlyeq' z'$, and
\item $g(x) \preccurlyeq' w'$ and $z' \in y' \circ' w'$ implies there exists $w, y, z \in X$ s.t. $x \preccurlyeq w$, $z \in y \circ w$, $y' \preccurlyeq' g(y)$ and $g(z) \preccurlyeq' z'$. \qed%
\end{enumerate}
\end{defi}

\noindent
For \logicfont{LGL}, replacing $\preccurlyeq$ with $=$ collapses the above definition to a simpler and more familiar notion of morphism.

\begin{defi}[LGL Morphism (cf.~\cite{Brotherston2014hybrid})]\label{def:layeredpmorphism} Given LGL frames $\mathcal{X}$ and $\mathcal{X}'$, a \emph{LGL morphism} is a map $g: \mathcal{X} \rightarrow \mathcal{X}'$ satisfying
\begin{enumerate}
\item $x \in y \circ z$ implies $g(x) \in g(y) \circ' g(z)$,
\item $g(x) \in y' \circ' z'$ implies there exists $y, z \in X$ s.t. $x \in y \circ z, g(y) = y' \text{ and } g(z) = z'$,
\item $z' \in g(x) \circ' y'$ implies there exists $y, z \in X$ s.t. $z \in x \circ y, g(y) = y' \text{ and } g(z) = z'$, and
\item $z' \in y' \circ' g(x)$ implies there exists $y, z \in X$ s.t. $z \in y \circ x, g(y) = y' \text{ and } g(z) = z'$. \qed%
\end{enumerate}
\end{defi}

\noindent
A variant of this definition is used in the context of BBI by Brotherston \& Villard~\cite{Brotherston2014hybrid} to demonstrate that the logic is not sufficiently expressive to axiomatise a number of properties common to models of separation logic. Urquhart~\cite{Urquhart1996} defines similar maps in order to define dualities for relevant logic and Bimb\'{o} \& Dunn~\cite{Dunn2008} generalise Urquhart's definition further to give morphisms that respect residuals on the dual algebras of frames for gaggles; it is this definition that we modify for our morphisms --- bunched logics can be seen to be variants of gaggles with extra operators and/or axioms.

\subsection{Duality for (I)LGL}
We now give representation and duality theorems for ILGL algebras. As a corollary we obtain the equivalence of the relational semantics to the algebraic semantics, as well as its completeness with respect to the Hilbert system of Fig.~\ref{fig:hilbert-lgl}. We first define two transformations that underpin functors between the category of (I)LGL algebras and the category of (I)LGL frames.

\begin{defi}[ILGL Complex Algebra]\label{defn:ilglcomplex}
Given an ILGL frame $\mathcal{X}$, the \emph{complex algebra} of $\mathcal{X}$ is given by $Com^{\text{\scalebox{.7}{ILGL}}}(\mathcal{X}) = (\mathcal{P}_{\succcurlyeq}(X), \cap, \cup, \Rightarrow, X, \emptyset, \bullet_{\mathcal{X}}, \Rres_{\mathcal{X}}, \Lres_{\mathcal{X}})$ where 
\begin{align*}
	\mathcal{P}_{\succcurlyeq}(X) &= \{ A \subseteq X \mid \text{ if } a \in A \text{ and } a \preccurlyeq b \text{ then } b \in A \} \\
	A \Rightarrow B &= \{ x \mid \text{if } x \preccurlyeq x' \text{ and } x' \in A \text{ then } x' \in B \} \\
	A \bullet_{\mathcal{X}} B &= \{ x \mid \text{there exists } w, y, z \text{ s.t } w \preccurlyeq x, w \in y \circ z, y \in A \text{ and } z \in B \} \\
	 A \Rres\nolimits_{\mathcal{X}} B &= \{ x \mid \text{for all } w, y, z, \text{ if } x \preccurlyeq w, z \in w \circ y \text{ and } y \in A \text{ then } z \in B \}  \\
	 A \Lres\nolimits_{\mathcal{X}} B &= \{ x \mid \text{for all } w, y, z, \text{ if } x \preccurlyeq w, z \in y \circ w \text{ and } y \in A \text{ then } z \in B \}. \tag*{\qed}
\end{align*}
\end{defi}

\noindent
Each operator here maps upwards-closed sets to upwards-closed sets so this is well defined. By substituting $\preccurlyeq$ for $=$, the above definition collapses to one suitable for \logicfont{LGL}.

\begin{defi}[LGL Complex Algebra]
Given a LGL frame $\framefont{X}$, the \emph{complex algebra} of $\framefont{X}$ is given by
$Com^{\text{\scalebox{.7}{LGL}}}(\framefont{X}) = (\pow{X}, \cap, \cup, \setminus, X, \emptyset, \bullet_{\mathcal{X}}, \Rres_{\mathcal{X}}, \Lres_{\mathcal{X}})$,
where $\bullet_{\mathcal{X}}, \Rres_{\mathcal{X}}$, and $\Lres_{\mathcal{X}}$ are defined as follows: \[ \begin{array}{cl}
A \Rightarrow B &= \overline{A} \cup B \\
A \bullet_{\mathcal{X}} B  &=  \{ x \mid \text{there exists } y \in A, z \in B \text{ s.t. } x \in y \circ z \} \\
A \Rres_{\mathcal{X}} B  &=  \{ x \mid \text{for all } y, z \in X, z \in x \circ y \text{ and } y \in A \text{ implies } z \in B \} \\
A \Lres_{\mathcal{X}} B  &=  \{ x \mid \text{for all } y, z \in X, z \in y \circ x \text{ and } y \in A \text{ implies }
z \in B \}.
\end{array}
\]
where $\overline{A}$ denotes set compliment. \qed%
\end{defi}

The residuated structure in each case is then easy to verify.

\begin{lem}\label{lem:ilglcomplex}
Given an (I)LGL frame $\mathcal{X}$, $Com^{\text{\scalebox{.7}{(I)LGL}}}(\mathcal{X})$ is an (I)LGL algebra. \qed%
\end{lem}

Any valuation $\Valuation$ on an (I)LGL frame $\mathcal{X}$ generates an interpretation $\Interp{-}_{\Valuation}$ on the complex algebra $Com^{\text{\scalebox{.7}{(I)LGL}}}(\mathcal{X})$. A straightforward inductive argument then shows satisfiability coincides on these two models.

\begin{prop}\label{prop:ilglsat}
For any relational (I)LGL model $(\mathcal{X}, \Valuation)$, $x \vDash_{\mathcal{V}} \varphi$ iff $x \in \Interp{\varphi}_{\Valuation}$. \qed%
\end{prop}

Conversely we can create ILGL frames from ILGL algebras. First we will need a sequence of definitions. A \emph{filter} on a bounded distributive lattice $\mathbb{A}$ is a non-empty set $F \subseteq A$ such that, for all $x, y \in A$, (i) $x \in F$ and $x \leq y$ implies $y \in F$, and (ii) $x, y \in F$ implies $x \land y \in F$. It is a \emph{proper} filter if it additionally satisfies (iii) $\bot \not\in F$. It is \emph{prime} if in addition it satisfies (iv) $x \lor y \in F$ implies $x \in F$ or $y \in F$. If $\mathbb{A}$ is a Boolean algebra and $F$ a proper filter, (iv) is equivalent to (v) $x \in F$ or $\neg x \in F$ and (vi) $F$ is a maximal proper filter with respect to $\subseteq$. The order-dual notion of a filter is an \emph{ideal}, with proper and prime ideals defined as one would expect. Importantly, if $I$ is a prime ideal then the complement $\overline{I}$ is a prime filter.

Given a non-empty subset $X$ of an algebra $\mathbb{A}$, we define
\[
    \filter{X} := \{ a \in \mathbb{A} \mid \exists b_1, \ldots, b_n \in X: a \geq b_1 \land \cdots \land b_n \}
\]
to be the \emph{filter generated by $X$}. It is easily seen that this is the least filter (with respect to set theoretic inclusion) containing the set $X$. Dually, we obtain the \emph{ideal generated by X} as $\ideal{X} := \{ a \in \mathbb{A} \mid \exists b_1, \ldots, b_n \in X: b_1 \lor \cdots \lor b_n \geq a \}$, and this is the least ideal containing $X$. For a singleton $\{ a\}$ we write $\filter{a}$ and $\ideal{a}$ for $\filter{\{a\}}$ and $\ideal{\{a\}}$.

We will frequently need to prove the existence of prime filters/ideals satisfying certain properties in order to prove our duality theoretic framework captures the structures associated with bunched logics. To this end we generalize a concept of Galmiche \& Larchey-Wendling~\cite{Galmiche2006} that gives a systematic method for showing such prime filters/ideals exist. First some terminology: a \emph{$\subseteq$-chain} is a sequence of sets ${(X_{\alpha})}_{\alpha < \lambda}$ such that $\alpha \leq \alpha'$ implies $X_{\alpha} \subseteq X_{\alpha'}$. A basic fact about proper filters (ideals) is that the union of a $\subseteq$-chain of proper filters (ideals) is itself a proper filter (ideal). We lift the terminology to $n$-tuples of sets by determining ${(X^1_{\alpha}, \ldots, X^n_{\alpha})}_{\alpha < \lambda}$ to be a $\subseteq$-chain if each ${(X^i_{\alpha})}_{\alpha < \lambda}$ is a $\subseteq$-chain.

\begin{defi}[Prime Predicate]\label{defn:primepredicate}
A \emph{prime predicate} is a map $P: \mathbb{F}_{\mathbb{A}}^n \times \mathbb{I}_{\mathbb{A}}^m \rightarrow \{ 0, 1\}$, where $n, m \geq 0$ and $n + m \geq 1$, such that
\begin{enumerate}[a)] 
\item Given a $\subseteq$-chain ${(F^0_{\alpha}, \ldots, F^n_{\alpha}, I^0_{\alpha}, \ldots, I^m_{\alpha})}_{\alpha < \lambda}$ of proper filters/ideals, \[\min \{ P(F^0_{\alpha}, \ldots, I^m_{\alpha}) \mid \alpha < \lambda \} \leq P(\bigcup_\alpha F^0_\alpha, \ldots, \bigcup_{\alpha} I^m_{\alpha})\]
\item $P(\ldots, H_0 \cap H_1, \ldots) \leq \max\{P(\ldots, H_0, \ldots), P(\ldots, H_1, \ldots)\}$. \qed%
\end{enumerate}
\end{defi}

\noindent
A prime predicate is a property that can hold of an $(n+m)$-tuple of proper filters and ideals that is evaluated to true or false. The two simple conditions it must satisfy to be a prime predicate are that the truth of the property persists from a chain of tuples of proper filters and ideals to their component-wise union and that if an n-tuple of proper filters and ideals for which one component is an intersection $H_0 \cap H_1$ is evaluated true, at least one of the tuples of prime filters and ideals obtained by replacing that intersection with either $H_0$ or $H_1$ is also evaluated true. Our definition differs from that of Galmiche \& Larchey-Wendling in a key way: our notion of prime predicate takes an argument of tuples of proper filters and ideals rather than a single filter, as we frequently need to prove the simultaneous existence of prime filters and ideals satisfying a particular condition. This will always follow from the following important lemma, which is easily proved using Zorn's lemma.

\begin{lem}[Prime Extension Lemma]\label{lem:primeextension}
If $P$ is an $(n+m)$-ary prime predicate and $F_0, \ldots, F_n, I_0, \ldots, I_m$ an $(n+m)$-tuple of proper filters and ideals such that
\[
    P(F_0, \ldots, F_n, I_0, \ldots, I_m) = 1
\]
then there exists a $(n+m)$-tuple of prime filters and ideals $F^{pr}_0, \ldots, F^{pr}_n, I^{pr}_0, \ldots I^{pr}_m$ such that
\begin{equation*}
    P(F^{pr}_0, \ldots, F^{pr}_n, I^{pr}_0, \ldots I^{pr}_m)=1 \tag*{\qed}
\end{equation*}
\end{lem}

We now return to the task at hand.

\begin{defi}[Prime Filter (I)LGL Frame]\label{defn:ilglprime}
Given an (I)LGL algebra $\mathbb{A}$, the \emph{prime filter frame} of $\mathbb{A}$ is given by $Pr^{\text{\scalebox{.7}{(I)LGL}}}(\mathbb{A}) = (Pr(A), \subseteq, \circ_{\mathbb{A}})$ where 
\begin{equation*}
F \circ_{\mathbb{A}} F' = \{ F'' \in Pr(\mathbb{A}) \mid \forall a \in F, \forall b \in F': a * b \in F'' \}. \tag*{\qed} 
\end{equation*}
\end{defi}

Of course, because of the structure of prime filters on Boolean algebras, the order $\subseteq$ collapses to $=$ for a prime filter LGL frame, as we would expect.

\begin{lem}\label{lem:ilglprime}
Given an (I)LGL algebra $\mathbb{A}$, $Pr^\vtiny{(I)LGL}(\mathbb{A})$ is an (I)LGL frame. \qed
\end{lem}

In analogy with Stone's representation theorem for Boolean algebras, we can give a representation theorem for (I)LGL algebras using these constructions. In particular, for ILGL algebras this extends the representation theorem for Heyting algebras~\cite{DH01}, whereas for LGL algebras this extends Stone's theorem. These results are closely related to various representation theorems for algebras with operators (e.g.,~\cite{tarski},~\cite{goldblatt}). The key difference  is the use of a single operation $\circ$ for the operator $*$ \emph{and} its non-operator adjoints $\wand$ and $\dnaw$. The derived structure required to take care of these adjoints was not investigated in the frameworks of J\'onsson-Tarski or Goldblatt but has been in the context of gaggle theory~\cite{Dunn2008,Dunn1990}. There the result for LGL algebras can be found as a particular case of that for Boolean gaggles (\cite{Dunn2008}, Theorem 1.4.16).

\begin{thm}[Representation Theorem for (I)LGL Algebras]\label{thm:ilglrep}
Every (I)LGL algebra is isomorphic to a subalgebra of a complex algebra. Specifically, given an (I)LGL algebra $\mathbb{A}$, the map $\theta_{\mathbb{A}}: \mathbb{A} \rightarrow Com^\vtiny{(I)LGL}(Pr^\vtiny{(I)LGL}(\mathbb{A}))$ defined $\theta_{\mathbb{A}}(a) = \{ F \in Pr^\vtiny{(I)LGL}(\mathbb{A}) \mid a \in F \}$ is an embedding. 
\end{thm}

\begin{proof} We prove the theorem for ILGL algebras; the case for LGL algebras can be obtained by substituting $\succcurlyeq$ for $=$ throughout. That $\theta_{\mathbb{A}}$ is an embedding and a homomorphism on the Heyting algebra operations is simply the representation theorem for Heyting algebras. It thus remains to show that $\theta_{\mathbb{A}}$ respects $*, \wand$ and $\dnaw$. We focus on the case for $\wand$; the others are similar.

We first note the corner cases: for all $a, b \in A$ we trivially have that $\theta_{\mathbb{A}}(a \wand \top) = \theta_{\mathbb{A}}(a) \Rres_{Pr(\mathbb{A})} \theta_{\mathbb{A}}(\top)$ and $\theta_{\mathbb{A}}(\bot \wand b) = \theta_{\mathbb{A}}(\bot) \Rres_{Pr(\mathbb{A})} \theta_{\mathbb{A}}(b)$ by Proposition~\ref{prop-alg-prop} property 6. Hence it is sufficient to consider $a \wand b$ where $a \neq \bot$ and $b \neq \top$.
 For the inclusion
$\theta_{\mathbb{A}}(a\wand b) \subseteq \theta_{\mathbb{A}}(a) \Rres_{Pr(\mathbb{A})} \theta_{\mathbb{A}}(b)$, assume $a \wand b \in F$ with $F_0, F_1, F_2$ such that $F \subseteq F_0$, $F_2 \in F_0 \circ_{\mathbb{A}} F_1$ and $a \in F_1$. Then $(a \wand b) * a \in F_2$ and so $b \in F_2$ since residuation entails $(a \wand b) * a \leq b$ and $F_2$ is upwards closed. Hence $F \in \theta_{\mathbb{A}}(a) \Rres_{Pr(\mathbb{A})} \theta_{\mathbb{A}}(b)$. 

For the reverse inclusion, consider $F$ such that $a \wand b \not\in F$. We show, for proper filter $G$ and proper ideal $I$, that (abusing notation for $\circ_{\mathbb{A}}$)
\[ P(G, I) =\begin{cases} 1 & \text{if } \overline{I} \in F \circ_{\mathbb{A}} G, a \in G \text{ and } b \in I \\
				0 & \text{otherwise} \end{cases} \] is a prime predicate.
We concentrate on the non-trivial verifications: suppose $P(G \cap G', I) =1$. Clearly, $a \in G, G'$ so suppose for contradiction that there exists $c, c' \in F$, $d \in G$ and $d' \in G'$ such that $c * d, c' * d' \in I$. We have that $c'' := c \land c' \in F$ and $c'' * d, c'' * d' \in I$. This entails $c'' * (d \lor d') = (c'' * d) \lor (c'' * d') \in I$. Since $d \lor d' \in G \cap G'$ we have $c'' * (d \lor d') \not\in I$ by assumption: a contradiction. Hence either $\overline{I} \in F \circ_{\mathbb{A}} G$ or $\overline{I} \in F \circ_{\mathbb{A}} G'$. If $P(G, I \cap I') = 1$ we clearly have $b \in I, I'$, so suppose for contradiction there exist $c, c' \in F$, $d, d' \in G$ such that $c * d \in I$ and $c' * d' \in I'$. $c'' = c \land c' \in F$ and $d'' = d \land d' \in G$ so we have $c'' * d'' \in \overline{I} \cup \overline{I'}$. This means $c'' * d'' \in \overline{I}$ or $c'' * d'' \in \overline{I'}$, but $c * d, c' * d' \geq c'' * d'' \in I \cap I'$, a contradiction. Thus either $\overline{I} \in F \circ_{\mathbb{A}} G$ or $\overline{I'} \in F \circ_{\mathbb{A}} G$.

Hence $P$ is a prime predicate. By our assumption on $a$ and $b$, $\filter{a}$ and $\ideal{b}$ are a proper filter and a proper ideal respectively, and $P(\filter{a}, \ideal{b}) =1$: if $x \in F$ and $y \geq a$ then $x * y \not\leq b$, otherwise by residuation and monotonicity of $*$ we would have $x \leq a \wand b \in F$, a contradiction. Hence by Lemma~\ref{lem:primeextension} there exist prime $G$ and $I$ with $P(G, I) =1$. Letting $G' = \overline{I}$, we have the prime filters  we require.
\end{proof}

That $\theta_{\mathbb{A}}$ is an embedding immediately gives us an analogous result to Prop.~\ref{prop:ilglsat} and thus a soundness and completeness theorem for the ``relational'' semantics of the logics.

\begin{cor}
For all (I)LGL algebras $\mathbb{A}$: given an interpretation $\Interp{-}$, the valuation $\Valuation_{\Interp{-}}(p) = \theta_{\mathbb{A}}(\Interp{p})$ on $Pr^\vtiny{(I)LGL}(\mathbb{A})$ is such that $\Interp{\varphi} \in F$  iff $F \vDash_{\Valuation_{\Interp{-}}} \varphi$. \qed
\end{cor}

\begin{cor}[Relational Soundness and Completeness]
For all formulas $\varphi, \psi$ of \logicfont{(I)LGL}: $\varphi \vdash \psi$ is provable in $\mathrm{(I)LGL}_{\mathrm{H}}$ iff $\varphi \vDash \psi$ in the relational semantics. \qed%
\end{cor}

The assignment of complex algebras and prime filter frames can be made functorial by defining $Pr^\vtiny{(I)LGL}(f) = f^{-1}$ and $Com^\vtiny{(I)LGL}(g) = g^{-1}$: that these define the appropriate morphisms is a straightforward but tedious verification. 

\begin{lem}\label{lem:ILGLmorphisms}
The functors $Pr^{\vtiny{(I)LGL}}$ and $Com^{\vtiny{(I)LGL}}$ are well defined. \qed
\end{lem}

However we are not yet in a position to give a dual equivalence of categories: in particular, a dual adjunction does not hold between these functors: there exist frames $\mathcal{X}$ such that \emph{no} morphism exists between $\mathcal{X}$ and $Pr^{\vtiny{(I)LGL}}Com^{\vtiny{(I)LGL}}(\mathcal{X})$ (cf. Venema's comments~\cite[pg 352]{Venema2007}). We \emph{can} obtain this by introducing topology, however. We first define ILGL spaces. This definition (necessarily) extends that of the topological duals of Heyting algebra given by Esakia duality~\cite{Esakia}. The coherence conditions on the composition $\circ$ are inspired by those found on the topological duals of gaggles~\cite{Dunn2008}. 




\begin{defi}[ILGL Space]
An \emph{ILGL space} is a structure $\mathcal{X} = (X, \mathcal{O}, \preccurlyeq, \circ)$ such that:
\begin{enumerate}
\item $(X, \mathcal{O}, \preccurlyeq)$ is an Esakia space~\cite{Esakia};
\item $(X, \preccurlyeq, \circ)$ is an ILGL frame;
\item The upwards-closed clopen sets of $(X, \mathcal{O}, \preccurlyeq)$ are closed under $\bullet_{\mathcal{X}}, \Rres_{\mathcal{X}}, \Lres_{\mathcal{X}}$;
\item If $x \not\in y \circ z$ then there exist upwards-closed clopen sets $C_1, C_2$ such that $y \in C_1, z \in C_2$ and $x \not\in C_1 \bullet_{\mathcal{X}} C_2$.
\end{enumerate}
A morphism of ILGL spaces is a continuous ILGL morphism, yielding a category $\mathrm{ILGLSp}$. \qed%
\end{defi}

Once again, substituting $\preccurlyeq$ for $=$ in the definition of ILGL space obtains the topological duals for LGL algebras.

\begin{defi}[LGL Space]
An \emph{LGL space} is a structure $\mathcal{X} = (X, \mathcal{O}, \circ)$ such that
\begin{enumerate}
\item $(X, \mathcal{O})$ is an Stone space~\cite{Stone};
\item $(X, \circ)$ is an LGL frame;
\item The clopen sets of $(X, \mathcal{O})$ are closed under $\bullet_{\mathcal{X}}, \Rres_{\mathcal{X}}, \Lres_{\mathcal{X}}$;
\item If $x \not\in y \circ z$ then there exist clopen sets $C_1, C_2$ such that $y \in C_1, z \in C_2$ and $x \not\in C_1 \bullet_{\mathcal{X}} C_2$.
\end{enumerate}
A morphism of LGL spaces is a continuous LGL morphism, yielding a category $\mathrm{LGLSp}$. \qed%
\end{defi} 

We now adapt $Pr^\vtiny{(I)LGL}$ and $Com^\vtiny{(I)LGL}$ for these categories and define the natural isomorphisms obtaining the duality.  First, it is straightforward to augment the prime filter frame with the topological structure required to make it an (I)LGL space. We define a subbase by $S = \{ \theta_{\mathbb{A}}(a) \mid a \in A \} \cup \{ \overline{\theta_{\mathbb{A}}(a)} \mid a \in A \}$ where $\overline{\theta_{\mathbb{A}}(a)}$ denotes the set complement. This generates a topology $\mathcal{O}_{\mathbb{A}}$ and we can (abusing notation) define $Pr^\vtiny{ILGL}: \mathrm{ILGLAlg} \rightarrow \mathrm{ILGLSp}$ by $Pr^\vtiny{ILGL}(\mathbb{A}) = (Pr(A), \mathcal{O}_{\mathbb{A}}, \subseteq, \circ_{\mathbb{A}})$ and $Pr^\vtiny{ILGL}(f) = f^{-1}$. In the case of \logicfont{LGL}, the sets $\overline{\theta_{\mathbb{A}}(a)}$ are redundant as for every prime filter $F$ on a Boolean algebra $\mathbb{A}$, $a \in A$ implies $a \in F$ or $\neg a \in F$. Hence $\mathcal{B} = \{ \theta_{\mathbb{A}}(a) \mid a \in A \}$ defines a base for a topology $\mathcal{O}_{\mathbb{A}}$ and we can set $Pr^\vtiny{LGL}(\mathbb{A}) = (Pr(A), \mathcal{O}_{\mathbb{A}}, \circ_{\mathbb{A}})$ and $Pr^\vtiny{LGL}(f) = f^{-1}$. 

Conversely, given an ILGL space $\mathcal{X}$ we now take the set of upwards-closed clopen sets $\mathcal{CL}_{\succcurlyeq}(\mathcal{X})$ as the carrier of an ILGL algebra, together with the operations of the ILGL complex algebra; ILGL space property (3) ensures this is well defined. Hence we take $Clop^\vtiny{ILGL}: \mathrm{ILGLSp} \rightarrow \mathrm{ILGLAlg}$ to be $Clop^\vtiny{ILGL}(\mathcal{X}) = (\mathcal{CL}_{\succcurlyeq}(\mathcal{X}), \cap, \cup, \Rightarrow, X, \emptyset, \bullet_{\mathcal{X}}, \Rres_{\mathcal{X}}, \Lres_{\mathcal{X}})$ and $Clop^\vtiny{ILGL}(g) = g^{-1}$.  For \logicfont{LGL} we simply take the clopen sets of the underlying topological space $\mathcal{CL}(\mathcal{X})$ together with the operations of the LGL complex algebra. Hence $Clop^\vtiny{LGL}(\mathcal{X}) = (\mathcal{CL}(\mathcal{X}), \cap, \cup, \setminus, X, \emptyset, \bullet_{\mathcal{X}}, \Rres_{\mathcal{X}}, \Lres_{\mathcal{X}})$ and $Clop^\vtiny{LGL}(g) = g^{-1}$.

The well-definedness of these functors can be seen by straightforwardly combining Stone (Esakia) duality for Boolean (Heyting) algebras with the preceding results relating to (I)LGL structures. For natural transformations, one --- $\theta$ --- is already defined in the representation theorem. The other is given by  $\eta_{\mathcal{X}}(x) = \{ A \in Clop^\vtiny{(I)LGL}(\mathcal{X}) \mid x \in A \}$. That these are natural isomorphisms is an easy extension of the case for Stone (Esakia) duality, with the additional verification that each $\eta_{\mathcal{X}}$ is an isomorphism on the (I)LGL frame structure following from (I)LGL space property (4). We thus obtain the duality theorem. For LGL algebras this is also obtainable as a specific case of Bimb\'o \& Dunn's duality theorem for Boolean gaggles (\cite{Dunn2008}, Theorem 9.2.22).

\begin{thm}[Duality Theorem for \logicfont{(I)LGL}]
$\theta: Id_{\mathrm{(I)LGLAlg}} \rightarrow Clop^{\vtiny{(I)LGL}}_{\succcurlyeq}$ and $\eta: Id_{\mathrm{(I)LGLSp}} \rightarrow Pr^{\vtiny{(I)LGL}}Clop^{\vtiny{(I)LGL}}_{\succcurlyeq}$ form a dual equivalence of categories between $\mathrm{(I)LGLAlg}$ and $\mathrm{(I)LGLSp}$. \qed
\end{thm}

\section{Logics of Bunched Implications}\label{sec:bunchedimplications}
\subsection{Algebra and Frames for (B)BI}
We now extend the results of the previous section to the logics \logicfont{BI} and \logicfont{BBI} by systematically extending the structures defined for \logicfont{LGL} and \logicfont{ILGL}. We begin once again with the lattice-based algebras that interpret the logics.

\begin{defi}[(B)BI Algebra]\label{def:bialgebra}\hfill
\begin{enumerate}
\item A \emph{BI algebra} is an algebra $\mathbb{A} = (A, \land, \lor, \rightarrow, \top, \bot, *, \wand, \mtop)$
such that $(A, \land, \lor, \rightarrow, \top, \bot, *, \wand, \wand)$ is an ILGL algebra and $(A, *, \mtop)$ a commutative monoid: that is, $*$ is commutative and associative with $\mtop$ a unit for $*$.
\item A \emph{BBI algebra} is a BI algebra $\mathbb{A} = (A, \land, \lor, \rightarrow, \top, \bot, *, \wand, \mtop)$ for which $(A, \land, \lor, \rightarrow, \top, \bot)$ is a Boolean algebra. \qed%
\end{enumerate}
\end{defi}

\noindent
Thus BI algebras are the subclass of ILGL algebras with a commutative and associative $*$ that has unit $\mtop$, and likewise for BBI algebras with respect to LGL algebras. In particular, commutativity of $*$, together with residuation, causes $\wand = \dnaw$.  Interpretations on BI and BBI algebras are given in much the same way as ILGL and LGL algebras, with the addition of $\Interp{\mtop} = \mtop$. A soundness and completeness theorem for algebraic interpretations is proved in the same fashion as Theorem~\ref{thm:algilglcomplete}.

\begin{thm}[Algebraic Soundness \& Completeness]\label{thm:bialgcomplete}
For \logicfont{(B)BI} formulas $\varphi$, $\psi$, $\varphi \vdash \psi$ is provable in $\mathrm{(B)BI}_{\mathrm{H}}$ iff for all (B)BI algebras $\mathbb{A}$ and all interpretations $\Interp{-}$ on $\mathbb{A}$, $\Interp{\varphi} \leq \Interp{\psi}$. \qed%
\end{thm}

Next we define a class of relational models appropriate for interpreting \logicfont{(B)BI}. The outermost universal quantification in each frame condition is left implicit for readability.

\begin{defi}[BI Frame]\label{defn:biframe}
A \emph{BI Frame} is a tuple $\mathcal{X} = (X, \succcurlyeq, \circ, E)$ where $(X, \succcurlyeq, \circ)$ is an ILGL frame, $E \subseteq X$ and the following conditions are satisfied:
\[\begin{array}{rlrl}
\text{(Commutativity)} & z \in x \circ y \rightarrow z \in y \circ x & \text{(Closure)} & e \in E \land e' \succcurlyeq e \rightarrow e' \in E \\
\text{(Unit Existence)} & \exists e \in E(x \in x \circ e) & \text{(Coherence)} & e \in E \land x \in y \circ e \rightarrow x \succcurlyeq y \\
\multicolumn{4}{l}{\text{(Associativity) } t' \succcurlyeq t \in x \circ y \land w \in t' \circ z \rightarrow \exists s, s', w'(s' \succcurlyeq s \in y \circ z \land w \succcurlyeq w' \in x \circ s')} \qed
\end{array}\]
\end{defi}

The notion of BI frame and its associated semantics can be defined in many ways, and indeed those given in Section~\ref{subsec:bunchedimplications} can suffice for a duality theorem if one defines the natural transformations and morphisms slightly differently. These choices are driven by the requirements that the model validates the associativity axiom and that the semantic clauses for $*$ and $\wand$ satisfy persistence. This requires careful interplay between the definition of (Associativity), the conditions relating $\circ$ and $\preccurlyeq$, and the semantic clauses for $*$ and $\wand$. One such solution is that given in the original papers on \logicfont{BI}~\cite{OP99}, in which $\circ$ is a deterministic,  associative function that is bifunctorial with respsect to the order. This is complicated somewhat when $\circ$ is non-deterministic and partial, however. A general analysis of the choices available for such a $\circ$ when the definition of the associativity-like condition it satisfies is kept fixed can be found in the work of Cao et al.~\cite{CCA17}.

Here, we use a definition that enables a uniform extension of the structures and theorems relating to \logicfont{ILGL}. In particular, we use a more general than usual formulation of (Associativity) that allows us to directly extend the results of Section~\ref{sec:layeredgraphlogics}. This solution places no coherence conditions on $\circ$ and $\preccurlyeq$ and thus requires what Cao et al.~\cite{CCA17} call the ``strong semantics'' for $*$ and $\wand$ to maintain persistence: precisely the clauses given for (I)LGL frames (Fig.~\ref{fig:sat-LGL-frame}) together with the semantic clause for $\mtop$ given in Fig.~\ref{fig:sat-bi}.
\begin{figure}
\centering
\hrule
\vspace{1mm}
\setlength\tabcolsep{3pt}
\setlength\extrarowheight{2pt}
\begin{tabular}{c c c c l r c c c c r r c}
$x$ & $\vDash_{\mathcal{V}}$ & $\mtop$ & iff & \myalign{l}{$x \in E$} & &
\end{tabular}
\caption{Satisfaction for (B)BI.}
\vspace{1mm}
\hrule%
\label{fig:sat-bi}
\end{figure}
This can be seen as a strict generalization of the UDMF models given in Section~\ref{subsec:bunchedimplications}, as witnessed by the fact that the definition of \emph{non-deterministic associativity} given there implies the frame condition (Associativity) when $\circ$ is additionally upwards and downwards closed.

\begin{prop}
Every upwards and downwards closed monoidal frame is a BI frame. \qed%
\end{prop}

Further, the respective semantic clauses for $*$ and $\wand$ are equivalent when the underlying model is upwards and downwards closed~\cite{CCA17}. The converse does not hold: not every BI frame is upwards or downwards closed. However, every BI frame generates a upwards and downwards closed monoidal frame with an equivalent satisfaction relation. Given a BI frame $\mathcal{X} = (X, \preccurlyeq, \circ, E)$, define its upwards and downwards closure by $\mathcal{X}^{\Uparrow \Downarrow} = (X, \preccurlyeq, \circ^{\Uparrow \Downarrow}, E)$ where $x \in y \circ^{\Uparrow \Downarrow} z$ iff there exist $x', y', z'$ such that $x' \preccurlyeq x$, $y \preccurlyeq y'$, $z \preccurlyeq z'$ and $x' \in y' \circ z'$. By taking care over the respective associativity properties of each frame, this can easily be seen to be an upwards and downwards closed monoidal frame. Letting $\vDash'$ denote the satisfaction relation defined for UDMFs in Fig.~\ref{fig:sat-BBI}, we have the following result.

\begin{prop}
For all BI frames $\mathcal{X}$, persistent valuations $\mathcal{V}$ and formulas $\varphi$ of \logicfont{BI}, $\mathcal{X}, x \vDash_{\mathcal{V}} \varphi$ iff $\mathcal{X}^{\Uparrow \Downarrow}, x \vDash'_{\mathcal{V}} \varphi$. \qed%
\end{prop}

So, as far as the logic is concerned, these choices are academic: they all collapse to the same notion of validity. For models satisfying one or both of upwards and downwards closure this is implicit in the preservation results of Cao et al.~\cite{CCA17}. Our analysis of the remaining class of models that lack both conditions on $\circ$ and $\preccurlyeq$ completes this picture.

\begin{defi}
A \emph{BBI frame} $\framefont{X}$ is a triple $\framefont{X} = (X, \circ, E)$, such that  $(X, \circ)$ is an LGL frame, $E \subseteq X$ and the following conditions are satisfied:
\[	\begin{array}{llllr}
\text{(Commutativity)} & z \in x \circ y \rightarrow z \in y \circ x & \text{(Coherence)} &
 x \in y \circ e \land e \in E \rightarrow y = x. \\ 
\text{(Unit Existence)} & \exists e \in E(x \in x \circ e) &   \\
\multicolumn{4}{l}{\text{(Associativity) } \quad t \in x \circ y \land w \in t \circ z \rightarrow \exists s(s \in y \circ z \land w \in x \circ s)} & \qed
\end{array}  \]
\end{defi}

We note that when $\preccurlyeq$ is substituted for $=$, all sound choices of the (Associativity) axiom that were possible for BI frames collapse to the axiom given here, while every coherence condition on $\circ$ and $\preccurlyeq$ becomes trivial. Thus, in comparison to \logicfont{BI}, there are far fewer choices to be made about \logicfont{BBI} models and so this definition is more familiar, appearing in the literature in precisely the same form as \emph{BBI frames}~\cite{Brotherston2014hybrid} and \emph{non-deterministic monoids}~\cite{Larchey-Wendling2014}, and slightly modified as \emph{multi-unit separation algebras}~\cite{Dockins2009} and \emph{relational frames}~\cite{Galmiche2006}.

The key difference with the latter definitions is that multi-unit separation algebras are cancellative --- $z \in x \circ y$ and $z \in x \circ y'$ implies $y = y'$ --- and relational frames have a single unit. BBI frames do not enforce cancellativity and have multiple units. This difference is crucial for the present work as the duality theorems do not hold when we restrict to frames satisfying either of these properties. This is witnessed by the fact that \logicfont{BBI} is not expressive enough to distinguish between cancellative/non-cancellative models and single unit/multi-unit models~\cite{Brotherston2014hybrid}, all of which define the same notion of validity~\cite{Larchey-Wendling2014}. An interpretation of the duality theorem might thus be that BBI frames are the most general relational structures that soundly and completely interpret \logicfont{BBI}.

\begin{defi}[(B)BI Morphism]
Given (B)BI frames $\mathcal{X}$ and $\mathcal{X'}$, a (B)BI morphism is an ILGL (LGL) morphism $g: \mathcal{X} \rightarrow \mathcal{X}'$ that additionally satisfies $\text{(7) } e \in E \text{ iff } g(e) \in E'$. \qed%
\end{defi}
BI frames together with BI morphisms form a category $\mathrm{BIFr}$, itself a subcategory of $\mathrm{ILGLFr}$; likewise, BBI frames and BBI morphisms form the category $\mathrm{BBIFr}$, a subcategory of $\mathrm{LGLFr}$. Note that commutativity of $\circ$ collapses the final conditions in the definition of (I)LGL morphism when defined on BBI (BI) frames.  As in the case for the layered graph logics, surjective (B)BI morphisms preserve validity in models.

We now relate the two categories of \logicfont{(B)BI} semantic structures with functorial transformations and a representation theorem.

\begin{defi}[(B)BI Complex Algebra]\label{defn:bicomplex}
Given a (B)BI frame $\mathcal{X} $, the complex algebra of $\mathcal{X}$, $Com^\vtiny{(B)BI}(\mathcal{X})$ is given by extending $Com^\vtiny{ILGL}(\mathcal{X})$ ($Com^\vtiny{LGL}(\mathcal{X})$) with the unit set of $\mathcal{X}$, $E$. \qed%
\end{defi}

\begin{lem}\label{lem:bicomplexalgebra}
Given a (B)BI frame $\mathcal{X}$, $Com^\vtiny{(B)BI}(\mathcal{X})$ is a (B)BI algebra. \qed%
\end{lem}

As a special case of the analogous result for \logicfont{(I)LGL}, we obtain a correspondence between satisfiability on a frame and its complex algebra when the algebraic interpretation is generated by the valuation on the frame.

\begin{prop}
For any (B)BI frame $\mathcal{X}$ and valuation $\Valuation$, $x \vDash_{\Valuation} \varphi$ iff $x \in \Interp{\varphi}_{\mathcal{V}}$. \qed%
\end{prop}

In the other direction we transform BI algebras into BI frames.

\begin{defi}[Prime Filter (B)BI Frame]
Given a (B)BI algebra $\mathbb{A}$, the prime filter frame of $\mathbb{A}$, $Pr^\vtiny{(B)BI}(\mathbb{A})$, is given by extending $Pr^\vtiny{ILGL}(\mathbb{A})$ ($Pr^{\vtiny{LGL}}(\mathbb{A}))$ with $E_{\mathbb{A}} = \{ F \in Pr(A) \mid \mtop \in F \}$. \qed
\end{defi}

That the prime filter frame of a (B)BI algebra is a (B)BI frame follows an argument very similar to that of Galmiche \& Larchey-Wendling's~\cite{Galmiche2006} completeness theorem for the relational semantics of \logicfont{BBI}.

\begin{lem}\label{lem:biprimeframe}
Given a (B)BI algebra $\mathbb{A}$, the prime filter frame $Pr^\vtiny{(B)BI}(\mathbb{A})$ is a (B)BI frame. 
\end{lem}

\begin{proof}
Commutativity of $\circ_{\mathbb{A}}$ can be read off the definition, given that $*$ is commutative for (B)BI\@. We also have that $E$ satisfies Closure trivially. We are left to verify Associativity, Unit Existence and Coherence. We note that in the case for BBI, maximality of prime filters collapses all of the inclusions to equalities in what follows, so we just give the argument for BI\@.

First, Associativity. Assume $F_{t'} \supseteq F_{t} \in F_x \circ F_y$ and $F_w \in F_{t'} \circ F_z$. We show that \[ \PrimePredicate{F}{F \in F_y \circ_{\mathbb{A}} F_z \text{ and } F_w \in F_x \circ_{\mathbb{A}} F}\]
is a prime predicate. For a $\subseteq$-chain ${(F_\alpha)}_{\alpha < \lambda}$ such that $P(F_\alpha) = 1$ for all $\alpha$, we straightforwardly have $P(\bigcup_\alpha F_\alpha) = 1$. If $P(F \cap F') =1$, we have that $F, F' \in F_y \circ F_z$ immediately, so suppose for contradiction that there exists $a, a' \in F_x$, $b \in F$, $b' \in F'$ such that $a * b, a' * b' \not\in F_w$. We have that $a'' = a \land a' \in F_x$ and $b \lor b' \in F \cap F'$ so $a'' * (b \lor b') = (a'' * b) \lor (a'' * b') \in F_w$. Since $F_w$ is prime, either $a'' * b \in F_w$ or $a'' * b' \in F_w$. Thus, because $*$ is monotone, $a * b \in F_w$ or $a' * b' \in F_w$,  a contradiction.

Now consider the set $F = \{ a \in \mathbb{A} \mid \exists b \in F_y, c \in F_z: a \geq b * c\}$. We show this is a proper filter satisfying $P(F) = 1$. First, suppose for contradiction that $\bot \in F$. Then there exists $b \in F_y$ and $c \in F_z$ such that $b* c = \bot$. Letting $a \in F_x$ be arbitrary, we have that $a * b \in F_t \subseteq F_{t'}$, so $(a * b) * c = a * (b * c) = a * \bot = \bot \in F_w$, contradicting that $F_w$ is proper. $F$ is clearly upwards-closed; to see it is closed under meets, consider $a, a' \in F$. Then there exists $b, b' \in F_y$ and $c, c' \in F_z$ such that $a \geq b * c$ and $a' \geq b' * c'$. We have that $b \land b' \in F_y$ and $c \land c' \in F_z$, and by monotonicity of $*$, $(b \land b') * (c \land c') \leq a * b, a' * b'$. Hence $(b \land b') * (c \land c') \leq (a * b) \land  (a' * b') \leq c \land c'$ as required.

We now verify that $P(F) = 1$. Let $b \in F_y$ and $c \in F_z$. Clearly $b * c \in F$, so $F \in F_y \circ_{\mathbb{A}} F_z$. If $a \in F_x$ and $a' \geq b * c$ for $b \in F_y$ and $c \in F_z$, we have that $a * a' \geq a * (b * c) = (a * b) * c \in F_w$, since $a * b \in F_t \subseteq F_{t'}$ and $c \in F_z$. Thus $a * a' \in F_w$ and $F_w \in F_x \circ_{\mathbb{A}} F$ as required. We thus obtain a prime $F$ with $P(F) =1$ by the Prime Extension Lemma, which is precisely what is required to satisfy Associativity.

For Unit Existence, let $F$ be an arbitrary prime filter. We show that \[ \PrimePredicate{G}{F \in F \circ_\mathbb{A} G\text{ and } \mtop \in G} \] is a prime predicate. If $P(G_\alpha) =1$ for all $G_\alpha$ in a $\subseteq$-chain ${(G_\alpha)}_{\alpha < \lambda}$ then clearly $F \in F \circ_{\mathbb{A}} \bigcup_\alpha G_\alpha$. Next, let $P(G \cap G') =1$ and assume for contradiction that there exists $a, a' \in F$, $b \in G$ and $b' \in G'$ such that $a * b \not\in F$ and $a' * b' \not\in F$. $b \lor b' \in G \cap G'$ so for $a'' = a \land a' \in F$ we have $a'' * (b \lor b') = (a'' * b) \lor (a'' * b') \in F$. Since $F$ is prime, either $a'' * b \in F$ or $a'' * b' \in F$. Hence, by monotonicity of $*$, either $a * b \in F$ or $a' * b' \in F$, a contradiction. Now consider the filter $\filter{\mtop}$. We note that this can only fail to be proper when $\mtop = \bot$, but in that case it can be shown that for all $a \in \mathbb{A}$, $a = \bot$, and thus $\mathbb{A}$ is degenerate and not a BI algebra. Given any $a \in F$ and $b \geq \mtop$, we have $a * b \geq a * \mtop = a \in F$, so $a * b \in F$. Since $P(\filter{\mtop}) = 1$, there exists a prime filter $F$ with $P(F) =1$, and so Unit Existence is satisfied.

Finally, for Coherence, assume $F_x \in F_y \circ F_e$ where $\mtop \in F_e$. Then for all $a \in F_y$, $a * \mtop = a \in F_x$, so $F_y \subseteq F_x$ as required.
\end{proof}

The representation theorem for (B)BI algebras now follows immediately from the analogous result for (I)LGL algebras and the fact that $\theta_{\mathbb{A}}(\mtop) = E_{\mathbb{A}}$.

\begin{thm}[Representation Theorem for (B)BI Algebras]\label{thm:birepresentation}
Every (B)BI algebra is isomorphic to a subalgebra of a complex algebra. Specifically, given a (B)BI algebra $\mathbb{A}$, the map $\theta_{\mathbb{A}}: \mathbb{A} \rightarrow Com^\vtiny{(B)BI}(Pr^\vtiny{(B)BI}(\mathbb{A}))$ defined $\theta_{\mathbb{A}}(a) = \{ F \in Pr^\vtiny{(B)BI}(A) \mid a \in F \}$ is an embedding. \qed
\end{thm}

\begin{cor}
For all (B)BI algebras $\mathbb{A}$, given an interpretation $\Interp{-}$, the valuation $\Valuation_{\Interp{-}}(p) = \theta_{\mathbb{A}}(\Interp{p})$ on $Pr^\vtiny{(B)BI}(\mathbb{A})$ is such that $\Interp{\varphi} \in F$ iff $F \vDash_{\Valuation_{\Interp{-}}} \varphi$. \qed
\end{cor}

\begin{cor}[Relational Soundness and Completeness]
For all formulas $\varphi$, $\psi$ of $\logicfont{BI}$, $\varphi \vdash \psi$ is provable in $\mathrm{BI}_{\mathrm{H}}$ iff $\varphi \vDash \psi$ in the relational semantics.  \qed%
\end{cor}

Once again $Pr^\vtiny{(B)BI}$ and $Com^\vtiny{(B)BI}$ can be made into functors by setting $Pr^\vtiny{(B)BI}(f) = f^{-1}$ and $Com^\vtiny{(B)BI}(g) = g^{-1}$. 


To obtain a dual equivalence of categories we add topological structure to BI frames. This can be achieved by straightforwardly extending (I)LGL spaces with BBI (BI) frame structure and specifying a coherence condition for the unit set $E$.

\begin{defi}[BI Space]
A BI space is a structure $\mathcal{X} = (X, \mathcal{O}, \preccurlyeq, \circ, E)$ such that
\begin{enumerate}
\item $(X, \mathcal{O}, \preccurlyeq, \circ)$ is an ILGL space,
\item $(X, \preccurlyeq, \circ, E)$ is a BI frame, and
\item $E$ is clopen in $(X, \mathcal{O})$.
\end{enumerate}
A morphism of BI spaces is a continuous BI morphism, yielding a category $\mathrm{BISp}$. \qed%
\end{defi}

\begin{defi}[BBI Space]
A \emph{BBI space} is a structure $\mathcal{X} = (X, \mathcal{O}, \circ, E)$ such that
\begin{enumerate}
\item $(X, \mathcal{O}, \circ)$ is an LGL space,
\item $(X, \circ, E)$ is a BBI frame, and
\item $E$ is clopen in $(X, \mathcal{O})$.
\end{enumerate}
A morphism of BBI spaces is a continuous BBI morphism, yielding a category $\mathrm{BBISp}$. \qed%
\end{defi}

The duality theorems for \logicfont{BI} and \logicfont{BBI} follow essentially immediately from those for \logicfont{ILGL} and \logicfont{LGL}. The only additional structure that needs to be taken care of is the constant $\mtop$ and the unit set $E$. We define the functors and natural isomorphisms explicitly for their use in the Separation Logic duality. Hence for \logicfont{BI} we have $Pr^\vtiny{BI}: \mathrm{BIAlg} \rightarrow \mathrm{BISp}$ defined by $Pr^\vtiny{BI}(\mathbb{A}) = (Pr(A), \mathcal{O}_{\mathbb{A}}, \subseteq, \circ_{\mathbb{A}}, E_{\mathbb{A}})$ (where $\mathcal{O}_{\mathbb{A}}$ is as defined for \logicfont{ILGL}) and once again $Pr^\vtiny{BI}(f) = f^{-1}$; correspondingly, $Clop^\vtiny{BI}: \mathrm{BISp} \rightarrow \mathrm{BIAlg}$ is given by $Clop^\vtiny{BI}(\mathcal{X}) = (\mathcal{CL}_{\succcurlyeq}(\mathcal{X}), \cap, \cup, \Rightarrow, X, \emptyset, \bullet_{\mathcal{X}}, \Rres_{\mathcal{X}}, E)$ (where $\bullet_{\mathcal{X}}, \Rres_{\mathcal{X}}$ are the ILGL complex algebra operations) and, as in the case for \logicfont{ILGL}, $Clop^\vtiny{BI}(g) = g^{-1}$; $\theta$ and $\eta$ are given precisely as they are in \logicfont{ILGL} duality, relativized to $\mathrm{BIAlg}$ and $\mathrm{BISp}$. 

 Similarly, for \logicfont{BBI} we have $Pr^\vtiny{BBI}(\mathbb{A}) = (Pr(A), \mathcal{O}_{\mathbb{A}}, \circ_{\mathbb{A}}, E_{\mathbb{A}})$ (where $\mathcal{O}_{\mathbb{A}}$ is as defined for \logicfont{LGL}) and $Pr^\vtiny{BBI}(f) = f^{-1}$; $Clop^\vtiny{BBI}(\mathcal{X}) = (\mathcal{CL}(\mathcal{X}), \cap, \cup, \setminus, X, \emptyset, \bullet_{\mathcal{X}}, \Rres_{\mathcal{X}}, E)$ (where $\bullet_{\mathcal{X}}, \Rres_{\mathcal{X}}$ are the LGL complex algebra operations) and $Clop^\vtiny{BBI}(g) = g^{-1}$; $\theta$ and $\eta$ are given precisely as they are in \logicfont{LGL} duality, relativized to $\mathrm{BBIAlg}$ and $\mathrm{BBISp}$. 

That $E_{\mathbb{A}}$ is clopen in each instance can be seen by the fact that $E = \theta_{\mathbb{A}}(\mtop)$. By Esakia duality, every set of the form $\theta_{\mathbb{A}}(a)$ for some $a \in A$ is an upwards-closed clopen set of the prime filter space of $\mathbb{A}$. Similarly, the clopen sets of the prime filter space of a Boolean algebra $\mathbb{A}$  are the sets $\theta_{\mathbb{A}}(a)$ for $a \in A$ by Stone duality, so the analogous property for BBI spaces holds too. It is also easy to see that the components of $\eta$ are additionally isomorphic with respect to $E$. The duality theorems thus obtain.  We note that the duality theorem for \logicfont{BI} has been independently obtained by Jipsen \& Litak~\cite{JipsenLitak}.

\begin{thm}[Duality Theorem for \logicfont{(B)BI}]
$\theta$ and $\eta$ form a dual equivalence of categories between $\mathrm{(B)BIAlg}$ and $\mathrm{(B)BISp}$.\qed%
\end{thm}

\section{Separation Logic}\label{sec:separationlogic}
\subsection{Hyperdoctrines and Indexed Frames for Separation Logic}
We now extend the duality theorems for BI and BBI algebras to the algebraic and relational structures suitable for interpreting Separation Logic.
First, we must consider first-order (B)BI (\logicfont{FO(B)BI}). Hilbert-type proof systems $\mathrm{FO(B)BI}_{\mathrm{H}}$ are obtained by extending those
given for \logicfont{(B)BI} in Section~\ref{sec:preliminaries} with the usual rules for quantifiers (see, e.g.,~\cite{TroelSchwich}).
Second, to give the semantics for the additional structure of $\logicfont{FO(B)BI}$, we expand the definitions
from the propositional case with category-theoretic structure. As these semantic structures support it,
we consider a many-sorted first-order logic. We start on the algebraic side with BI hyperdoctrines.

\begin{defi}[(B)BI Hyperdoctrine (cf.~\cite{Biering2005})]\label{def:resourcehyperdoctrine}
A \emph{(B)BI hyperdoctrine} is a tuple
\[
(\algfont{P}: \catfont{C}^{op} \rightarrow \catfont{Poset}, {(=_X)}_{X \in Ob(\catfont{C})}, {(\exists X_{\Gamma}, \forall X_{\Gamma})}_{\Gamma, X \in Ob(\catfont{C})})
\]
such that:
	\begin{enumerate}
	\item $\catfont{C}$ is a category with finite products;
	\item $\algfont{P}: \catfont{C}^{op} \rightarrow \catfont{Poset}$ is a functor such that, for each object $X$ in $\catfont{C}$, $\algfont{P}(X)$ is a (B)BI algebra, and, for each morphism $f$ in $\catfont{C}$, $\algfont{P}(f)$ is a (B)BI algebra homomorphism;
	\item For each object $X$ in $\catfont{C}$ and each diagonal morphism
		$\Delta_X: X \rightarrow X \times X$ in $\catfont{C}$,  the element
		$=_{X} \in \algfont{P}(X \times X)$ is adjoint at $\top_{\mathbb{P}(X)}$. That is, for all $a \in \algfont{P}(X \times X)$,
		\[
			\top_{\mathbb{P}(X)} \leq \algfont{P}(\Delta_X)(a) \text{ iff } =_{X} \leq a \,;
		\]
	\item For each pair of objects $\Gamma, X$ in $\catfont{C}$ and each projection
		$\pi_{\Gamma, X}: \Gamma \times X \rightarrow \Gamma$ in
		$\catfont{C}$, $\exists X_{\Gamma}$ and $\forall X_{\Gamma}$ are left and right adjoint to $\mathbb{P}(\pi_{\Gamma, X})$. That is, they are monotone maps
		$\exists X_{\Gamma}: \algfont{P}(\Gamma \times X) \rightarrow \algfont{P}(\Gamma)$ and
		${\forall X_{\Gamma}: \algfont{P}(\Gamma \times X) \rightarrow \algfont{P}(\Gamma)}$
		such that, for all $a, b \in \algfont{P}(\Gamma)$,
	\[
	\begin{array}{lll}
	\exists X_{\Gamma}(a) \leq b & \text{iff} & a \leq \algfont{P}(\pi_{\Gamma, X})(b) \;\mbox{\rm and} \\
	\algfont{P}(\pi_{\Gamma, X})(b) \leq a & \text{iff} & b \leq \forall X_{\Gamma}(a).
	\end{array}
	\]
		This assignment of adjoints is additionally natural in $\Gamma$: given a morphism
	$s: \Gamma \rightarrow \Gamma'$, the following diagrams commute: \vspace{1mm}
\[
\begin{array}{lr}
	\begin{tikzcd}
	\algfont{P}(\Gamma' \times X) \arrow{r}{\algfont{P}(s \times id_X)} \arrow{d}[swap]{\exists X_{\Gamma'}} &
	\algfont{P}(\Gamma \times X) \arrow{d}{\exists X_{\Gamma}} \\
	\algfont{P}(\Gamma') \arrow{r}[swap]{\algfont{P}(s)} & \algfont{P}(\Gamma)
	\end{tikzcd} & \begin{tikzcd}
	\algfont{P}(\Gamma' \times X) \arrow{r}{\algfont{P}(s \times id_X)} \arrow{d}[swap]{\forall X_{\Gamma'}} &
	\algfont{P}(\Gamma \times X) \arrow{d}{\forall X_{\Gamma}} \\
	\algfont{P}(\Gamma') \arrow{r}[swap]{\algfont{P}(s)} & \algfont{P}(\Gamma)
	\end{tikzcd}
\end{array} \]
	\end{enumerate} \qed%
\end{defi} 

\noindent
(B)BI hyperdoctrines were first formulated by Biering et al.~\cite{Biering2005} to prove the existence of models of higher-order variants of Separation Logic. There it was shown that the standard model of Separation Logic could be seen as a BBI hyperdoctrine, and thus safely extended with additional structure in the domain $\mathrm{C}^{op}$ to directly define higher-order constructs like lists, trees, finite sets and relations inside the logic. The present work strengthens this result to a dual equivalence of categories. Other algebraic models of Separation Logic, like those based on Boolean quantales~\cite{DANG2011221} or formal power series~\cite{DongolGS15}, can be seen as particular instantiations of BBI hyperdoctrines.

To specify an interpretation  $\llbracket - \rrbracket$ of \logicfont{FO(B)BI} in a (B)BI hyperdoctrine $\algfont{P}$
we assign each type $X$ an object
$\llbracket X \rrbracket$ of $\catfont{C}$, and for each context of variables $\Gamma = \{ v_1: X_1, \ldots, v_n: X_n \}$
we have
$\llbracket \Gamma \rrbracket = \llbracket X_1 \rrbracket \times \cdots \times \llbracket X_n \rrbracket$. Each function
symbol $f: X_1 \times \cdots X_n \rightarrow X$ is assigned a morphism
$\llbracket f \rrbracket: \llbracket X_1 \rrbracket \times \cdots \llbracket X_n \rrbracket \rightarrow \llbracket X \rrbracket$.
This allows us to inductively assign to every term of type $X$ in context $\Gamma$ a morphism
$\llbracket t \rrbracket : \llbracket \Gamma \rrbracket \rightarrow \llbracket X \rrbracket$ in the standard way (see~\cite{Pitts}).
We additionally assign, for each $m$-ary predicate symbol $P$ of type $X_1, \ldots, X_m$,
$\llbracket P \rrbracket \in \algfont{P}(\llbracket X_1 \rrbracket \times \cdots \times \llbracket X_m \rrbracket)$.
Then the structure of the hyperdoctrine allows us to extend $\llbracket - \rrbracket$ to \logicfont{FO}(\logicfont{B})\logicfont{BI} formulae $\varphi$ in context $\Gamma$ as follows:
\[
\arraycolsep=0.5pt
\def\arraystretch{1.2}
\begin{array}{llllrlll}
\Interp{Pt_1\ldots t_m} &= \mathbb{P}(\langle \llbracket t_1 \rrbracket, \ldots, \llbracket t_m \rrbracket \rangle)(\llbracket P \rrbracket) & \llbracket \varphi \land \psi \rrbracket &= \llbracket \varphi \rrbracket \land_{\algfont{P}(\llbracket \Gamma \rrbracket)} \llbracket \psi \rrbracket   &  \Interp{\top} &= \top_{\algfont{P}(\Interp{\Gamma})}  \\
\Interp{t=_X t'} &= \algfont{P}(\langle \llbracket t \rrbracket, \llbracket t' \rrbracket \rangle)(=_{\llbracket X \rrbracket}) &\llbracket \varphi \lor \psi \rrbracket &= \llbracket \varphi \rrbracket \lor_{\algfont{P}(\llbracket \Gamma \rrbracket)} \llbracket \psi \rrbracket & \Interp{\bot} &= \bot_{\algfont{P}(\Interp{\Gamma})} \\
\llbracket \varphi \rightarrow \psi \rrbracket & = \llbracket \varphi \rrbracket \rightarrow_{\algfont{P}(\llbracket \Gamma \rrbracket)} \llbracket \psi \rrbracket & \Interp{\varphi * \psi} & = \Interp{\varphi} *_{\algfont{P}(\llbracket \Gamma \rrbracket)} \Interp{\psi} & \Interp{\mtop} &= \mtop_{\algfont{P}(\Interp{\Gamma})} \\
\Interp{\varphi \wand \psi} & = \Interp{\varphi} \wand_{\algfont{P}(\llbracket \Gamma \rrbracket)} \Interp{\psi}
& \Interp{\exists v:X. \varphi} &= \exists \Interp{X}_{\Interp{\Gamma}} (\Interp{\varphi}) & \Interp{\forall v:X. \varphi} &= \forall \Interp{X}_{\Interp{\Gamma}} (\Interp{\varphi}).
\end{array}
\]

Substitution of terms is given by $\llbracket \varphi(t/x) \rrbracket = \mathbb{P}(\llbracket t \rrbracket)(\llbracket \varphi \rrbracket)$. $\varphi$ is satisfied by an interpretation $\llbracket - \rrbracket$ if $\llbracket \varphi \rrbracket = \top_{\mathbb{P}(\Interp{\Gamma})}$. $\varphi$ is valid if it is satisfied by all interpretations. 
A standard Lindenbaum-Tarski style construction is sufficient to prove soundness and completeness in both cases.

\begin{thm}\cite{Pitts,Biering2005}
For all \logicfont{FO(B)BI} formulas $\varphi$, $\psi$ in context $\Gamma$, $\varphi \vdash^{\Gamma} \psi$ is provable in $\mathrm{FO(B)BI}_{\mathrm{H}}$  iff, for all (B)BI hyperdoctrines $\mathbb{P}$ and all interpretations $\Interp{-}$, $\Interp{\varphi} \leq_{\mathbb{P}(\Interp{\Gamma})} \Interp{\psi}$.  \qed%
\end{thm}

It is also worth stating a simple lemma that can be obtained as an immediate consequence of the adjointness properties of $\exists X_{\Gamma}$ and $\forall X_{\Gamma}$ as it will be invoked frequently in proofs.

\begin{lem}
Given a (B)BI hyperdoctrine $\mathbb{P}: \mathrm{C}^{op} \rightarrow \mathrm{Poset}$, for all $a, b \in \mathbb{P}(\Gamma)$ the following hold:
\begin{enumerate}
\item $a \leq \mathbb{P}(\pi_{\Gamma, X})(\exists X_{\Gamma}(a))$ and $\exists X_{\Gamma}(\mathbb{P}(\pi_{\Gamma, X})(b)) \leq b$;
\item $b \leq \forall X_{\Gamma}(\mathbb{P}(\pi_{\Gamma, X})(b))$ and $\mathbb{P}(\pi_{\Gamma, X})(\forall X_{\Gamma}(a)) \leq a$;
\item $\exists X_{\Gamma}(\bot) = \bot$ and $\forall X_{\Gamma}(\top) = \top$. \qed%
\end{enumerate}
\end{lem}

On the `relational' side we introduce new structures: \emph{indexed (B)BI frames}. This definition is adapted from the notion of indexed Stone space presented by Coumans~\cite{Coumans2010} as a topological dual for Boolean hyperdoctrines. In contrast to the duality presented there, we prove the duality for the more general intuitionistic case and additionally consider (typed) equality and universal quantification. They also appear to have some relation to a more general formulation of Shirasu's \emph{metaframes}~\cite{S98}, another type of indexed frame introduced to interpret predicate superintuitionistic and modal logics, but we defer an investigation of this connection to another occasion.

\begin{defi}\label{def:indresourceframe}
An \emph{indexed (B)BI frame} is a functor
$\functorfont{R}: \catfont{C} \rightarrow \catfont{BIFr}$ such that
	\begin{enumerate}
	\item $\catfont{C}$ is a category with finite products;
	\item For all objects $\Gamma, \Gamma'$ and $X$ in $\catfont{C}$, all morphisms
		$s: \Gamma \rightarrow \Gamma'$ and all product projections
		$\pi_{\Gamma, X}$, 
		 for the following commutative square
\[
		\begin{tikzcd}
	\functorfont{R}(\Gamma \times X) \arrow{r}{\functorfont{R}(\pi_{\Gamma, X})} \arrow{d}{\functorfont{R}(s \times id_X)} & \functorfont{R}(\Gamma) \arrow{d}[swap]{\functorfont{R}(s)} \\
	\functorfont{R}(\Gamma' \times X) \arrow{r}[swap]{\functorfont{R}(\pi_{\Gamma', X})} & \functorfont{R}(\Gamma')
	\end{tikzcd} \]
\begin{itemize}
\item[(a)] (for indexed BI frames) the (Pseudo Epi) property holds: $\mathcal{R}(\pi_{\Gamma', X})(y) \preccurlyeq \mathcal{R}(s)(x)$ implies there exists $z$ such that:
		$\mathcal{R}(\pi_{\Gamma, X})(z) \preccurlyeq x$ and
		$y \preccurlyeq \mathcal{R}(s \times id_X)(z)$;
\item[(b)] (for indexed BBI frames) the quasi-pullback property holds: the induced map
\[
    \functorfont{R}(\Gamma \times X) \rightarrow \functorfont{R}(\Gamma) \times_{\functorfont{R}(\Gamma')} \functorfont{R}(\Gamma' \times X)
\]
is an epimorphism.
\end{itemize}
	\end{enumerate}
Given an arbitrary indexed BI frame $\functorfont{R}: \catfont{C} \rightarrow \catfont{(B)BIFr}$
and an object $X$ we denote the BI frame at $X$ by $\functorfont{R}(X) = (\functorfont{R}(X),  \preccurlyeq_{\functorfont{R}(X)}, \circ_{\functorfont{R}(X)}, E_{\functorfont{R}(X)})$. Analogously, we denote the BBI frame at $X$ by $\functorfont{R}(X) = (\functorfont{R}(X), \circ_{\functorfont{R}(X)}, E_{\functorfont{R}(X)})$ in the case of an indexed BBI frame. \qed%
\end{defi}

Although it may not look like it yet, condition (2) ensures that an interpretation of quantifiers based on the projections coheres correctly with the appropriate changes in context. The relation between the definition for indexed BI and BBI frames may not seem entirely clear at first, but unpacking what it means for the square to be a quasi-pullback should clarify: if $\mathcal{R}(\pi_{\Gamma', X})(y) = \mathcal{R}(s)(x)$ then there exists $z$ such that:
		$\mathcal{R}(\pi_{\Gamma, X})(z) = x$ and
		$y = \mathcal{R}(s \times id_X)(z)$.

A Kripke-style semantics can be given for \logicfont{FO(B)BI} on indexed (B)BI frames. For \logicfont{FOBI}, an interpretation $\Interp{-}$ is given in precisely the same way as for BI hyperdoctrines, except for the key-difference that each $m$-ary predicate symbol $P$ of type $X_1, \ldots, X_m$ is assigned to an upwards closed subset $\Interp{P} \in \mathcal{P}_{\preccurlyeq}(\mathcal{R}(\llbracket X_1 \rrbracket \times \cdots \times \llbracket X_m \rrbracket))$. Similarly,  an interpretation $\Interp{-}$ for $\logicfont{FOBBI}$ is given in the same way as it is for BBI hyperdoctrines, except that, for every $m$-ary predicate symbol $P$ of type $X_1, \ldots, X_m$, $P$ is assigned to a subset $\Interp{P} \in \mathcal{P}(\mathcal{R}(\llbracket X_1 \rrbracket \times \cdots \times \llbracket X_m \rrbracket))$.  Then for formulas $\varphi$ of \logicfont{FO(B)BI} in context $\Gamma$ with $x \in \functorfont{R}(\Interp{\Gamma})$ the satisfaction relation $\vDash^{\Gamma}$ is inductively defined in Fig.~\ref{fig:sat-IRF}. There, $Ran(\mathcal{R}(\Delta_{\llbracket X \rrbracket})) =  \{ y \mid \exists z(\mathcal{R}(\Delta_{\llbracket X \rrbracket})(z) = y) \}$. We note that bound variables are renamed to be fresh throughout, in an order determined by quantifier depth.

The familiar persistence property of propositional intuitionistic logics also holds for satisfaction on indexed BI frames. For atomic predicate formulas this is by design, with the assignment of predicate symbols to upwards closed subsets akin to a persistent valuation. For formulas of the form $t =_X t'$ this follows from the fact that $\mathcal{R}(\Delta_X)$ is a BI morphism and hence order preserving. The rest of the clauses follow by an inductive argument, the most involved of which is for formulas of the form $\exists v_{n+1}:X\varphi$. Suppose $x, \Interp{-} \vDash^{\Gamma} \exists v_{n+1}:X \varphi$ and $y \succcurlyeq_{\mathcal{R}(\Interp{\Gamma})} x$. Then by definition there exists $x'$ such that $\mathcal{R}(\pi_{\llbracket \Gamma \rrbracket, \llbracket X \rrbracket})(x') = x \preccurlyeq_{\mathcal{R}(\Interp{\Gamma})} y$ and $x', \llbracket - \rrbracket  \vDash^{\Gamma \cup \{ v_{n+1}\!:\!X \}} \varphi$. Since $\mathcal{R}(\pi_{\llbracket \Gamma \rrbracket, \llbracket X \rrbracket})$ is a BI morphism, there exists $y'$ such that $y' \succcurlyeq_{\mathcal{R}(\Interp{\Gamma} \times \Interp{X})} x'$ and $\mathcal{R}(\pi_{\llbracket \Gamma \rrbracket, \llbracket X \rrbracket})(y') = y$. By the inductive hypothesis, $y', \llbracket - \rrbracket  \vDash^{\Gamma \cup \{ v_{n+1}\!:\!X \}} \varphi$ and so $y, \llbracket - \rrbracket  \vDash^{\Gamma} \exists v_{n+1}:X\varphi$.

\begin{figure}
\centering
\hrule
\vspace{1mm}
 \setlength\tabcolsep{1pt}
\setlength\extrarowheight{3pt}
\begin{tabular}{c c c l l r c c c c r r c}
$x, \Interp{-}$ & $\vDash^{\Gamma}$ & $Pt_1\ldots t_m$ & iff & \myalign{l}{$\functorfont{R}(\langle \llbracket t_1 \rrbracket, \ldots, \llbracket t_m \rrbracket \rangle)(x)
\in \llbracket P \rrbracket$} & &
$x, \Interp{-}$ & $\vDash^{\Gamma}$ & $\top$ \\
$x, \Interp{-}$ & $\vDash^{\Gamma}$ & $t=_X t'$ & iff & \myalign{l}{$\functorfont{R}(\langle \llbracket t \rrbracket, \llbracket t' \rrbracket \rangle)(x) \in
Ran(\functorfont{R}(\Delta_{\llbracket X \rrbracket}))$} & &
$x, \Interp{-}$ & $\not\vDash^{\Gamma}$ & $\bot$ \\
$x, \Interp{-}$ & $\vDash^{\Gamma}$ & $\varphi \land \psi$ & iff & \myalign{l}{$x, \Interp{-} \vDash^{\Gamma} \varphi$ and $x, \Interp{-} \vDash^{\Gamma} \psi$} \\
$x, \Interp{-}$ & $\vDash^{\Gamma}$ & $\varphi \lor \psi$ & iff & \myalign{l}{$x, \Interp{-} \vDash^{\Gamma} \varphi$ or $x, \Interp{-} \vDash^{\Gamma} \psi$}
\\
$x, \Interp{-}$ & $\vDash^{\Gamma}$ & $\varphi \rightarrow \psi$ & iff & \multicolumn{8}{l}{for all $x' \succcurlyeq_{\mathcal{R}(\Interp{\Gamma})} x$, $x', \Interp{-}$ $\vDash^{\Gamma}$ $\varphi$ implies $x', \Interp{-} \vDash^{\Gamma} \psi$} \\
$x, \Interp{-}$ & $\vDash^{\Gamma}$ & $\mtop$ & iff & \myalign{l}{$x \in E_{\mathcal{R}(\Interp{\Gamma})}$} \\
$x, \Interp{-}$ & $\vDash^{\Gamma}$ & $\varphi * \psi$ & iff & \multicolumn{8}{l}{there exists $x' \preccurlyeq x$ s.t. $x' \in y \circ_{\mathcal{R}(\Interp{\Gamma})} z$, $y, \Interp{-} \vDash^{\Gamma} \varphi$ and  $z, \Interp{-} \vDash^{\Gamma} \psi$} \\
$x, \Interp{-}$ & $\vDash^{\Gamma}$ & $\varphi \wand\psi$ & iff & \multicolumn{8}{l}{for all $x' \succcurlyeq x$ s.t. $z \in x' \circ_{\mathcal{R}(\Interp{\Gamma})} y$,  $y, \Interp{-} \vDash^{\Gamma} \varphi$ implies $z, \Interp{-} \vDash^{\Gamma} \psi$} \\
$x, \Interp{-}$ & $\vDash^{\Gamma}$ & $\exists v_{n+1}:X \varphi$ & iff & \multicolumn{8}{l}{there exists $x' \in \functorfont{R}(\llbracket \Gamma \rrbracket \times \llbracket X \rrbracket)$
s.t. $\functorfont{R}(\pi_{\llbracket \Gamma \rrbracket, \llbracket X \rrbracket})(x') = x$
and} \\
& & & & $x', \llbracket - \rrbracket  \vDash^{\Gamma \cup \{ v_{n+1}\!:\!X \}} \varphi$ \\
$x, \Interp{-}$ & $\vDash^{\Gamma}$ & $\forall v_{n+1}:X \varphi$ & iff & \multicolumn{8}{l}{for all $x' \in \functorfont{R}(\llbracket \Gamma \rrbracket \times \llbracket X \rrbracket)$,
$\functorfont{R}(\pi_{\llbracket \Gamma \rrbracket, \llbracket X \rrbracket})(x') \succcurlyeq_{\mathcal{R}(\Interp{\Gamma})} x$,
implies} \\
&&&&  $x', \llbracket - \rrbracket  \vDash^{\Gamma \cup \{ v_{n+1}:X\}} \varphi$
\end{tabular}
\caption{Satisfaction on indexed (B)BI frames for \logicfont{FO(B)BI}. \logicfont{FOBBI} replaces $\preccurlyeq$ with $=$.}
\vspace{1mm}
\hrule%
\label{fig:sat-IRF}
\end{figure}

\subsection{Pointer Logic as an Indexed Frame}\label{subsec:pointerframe}

Although at first sight it may not seem so, indexed frames and the semantics based
upon them are a generalization of the standard store--heap semantics of Separation Logic. 

Consider the BI frame $\mathrm{Heap}^\vtiny{BI} = (H, \uplus, \sqsubseteq, H)$, where
$H$ is the set of heaps, $\sqsubseteq$ is heap extension, and $\uplus$ is defined by $h_2 \in h_0 \uplus h_1$ iff
$h_0\# h_1$ and $h_0 \cdot h_1 = h_2$. This is the BI frame corresponding to the partial monoid of heaps. We define an indexed BI frame $\mathrm{Store}^\vtiny{BI}: \catfont{Set} \rightarrow \catfont{BIFr}$ on objects
by $\mathrm{Store}^\vtiny{BI}(X) = {(X \times H, \uplus_X,\sqsubseteq_X, X \times H )}$, where $(x_2, h_2) \in (x_0, h_0) \uplus_X (x_1, h_1)$ iff $x_0=x_1=x_2$ and $h_2 \in h_0 \uplus h_1$, and $(x_0, h_0) \sqsubseteq_X (x_1, h_1)$ iff $x_0 = x_1$ and $h_0 \sqsubseteq h_1$. On morphisms, set $\mathrm{Store}^\vtiny{BI}(f: X \rightarrow Y)(x, h) = (f(x), h)$.
It is straightforward to see this defines a functor: for arbitrary $X$, $\mathrm{Store}(X)$
inherits the BI frame properties from $\mathrm{Heap}^\vtiny{BI}$ and for arbitrary $f: X \rightarrow Y$, $\mathrm{Store}(f)$
is trivially a BI morphism as it is identity on the structure that determines the back and forth conditions.
The property (Pseudo Epi) is also trivially satisfied so this defines an indexed BI frame.

For BBI pointer logic, we instead start with the BBI frame $\mathrm{Heap}^\vtiny{BBI} = (H, \uplus, \{ [] \})$ where $[]$ is the empty heap. Then $\mathrm{Store}^\vtiny{BBI}$ is defined in essentially the same way, with $\mathrm{Store}^\vtiny{BBI}(X) = (X \times H, \uplus_X, X \times \{ [] \} )$ and $\mathrm{Store}^\vtiny{BBI}(f)(x, h) = (f(x), h)$. This defines an indexed BBI frame.

We now describe the interpretations $\llbracket - \rrbracket$ on $\mathrm{Store}^\vtiny{(B)BI}$ that yield the standard models of Separation Logic.
We have one type $\mathrm{Val}$ and we set $\llbracket \mathrm{Val} \rrbracket = \mathbb{Z}$, with the arithmetic operations
$\llbracket + \rrbracket, \llbracket - \rrbracket: \llbracket \mathrm{Val} \rrbracket^2 \rightarrow \llbracket \mathrm{Val} \rrbracket$
defined as one would expect. Term morphisms $\llbracket t \rrbracket: \llbracket \mathrm{Val} \rrbracket^n \rightarrow \llbracket \mathrm{Val} \rrbracket$
in context $\Gamma = \{v_1, \ldots v_n \}$ are then defined as usual,
with each constant $n$ assigned the morphism
\begin{tikzcd}
\llbracket n \rrbracket: \llbracket \Gamma \rrbracket \arrow{r} & \{ * \} \arrow{r}{n} & \llbracket \mathrm{Val} \rrbracket.
\end{tikzcd} As one would expect, the key difference between the two interpretations is in the interpretation of the points-to predicate. For Intuitionistic Separation Logic, the points-to predicate $\mapsto$ is assigned
\[ \llbracket \mapsto \rrbracket = \{ ((a, a'), h) \mid a \in dom(h) \text{ and } h(a) = a' \}
\in \mathcal{P}_{\sqsubseteq_{\Interp{\mathrm{Val}}^2}} (\mathrm{Store}^\vtiny{BI}(\llbracket \mathrm{Val} \rrbracket^2)).\]
This set is clearly upwards closed with respect to the order $\sqsubseteq_{\Interp{\mathrm{Val}}^2}$ so this is a well-defined interpretation. For Classical Separation Logic, $\mapsto$ is instead assigned
\[ \Interp{\mapsto} = \{ ((a, a'), h) \mid \{ a \} = dom(h) \text{ and } h(a) = a' \} \in \mathcal{P}(\mathrm{Store}^\vtiny{BBI}(\Interp{\mathrm{Val}}^2)). \]

In the indexed (B)BI frame $\mathrm{Store}^\vtiny{(B)BI}: \catfont{Set} \rightarrow \catfont{(B)BIFr}$ with the interpretations just defined, a store is represented as an $n$-place vector of values over $\llbracket \mathrm{Val} \rrbracket$.
That is, the store $s = \{(v_1, a_1), \ldots, (v_n, a_n) \}$ is given by the element
$(a_1, \ldots, a_n) \in \llbracket \mathrm{Val} \rrbracket^n$. By a simple inductive argument we have the following result:

\begin{thm}\label{theorem:slequiv}
For all formulas $\varphi$ of (B)BI pointer logic, all stores $s = \{ (v_1, a_1), \dots, \!(v_n, a_n) \}$ and all
heaps $h$, $s, h \vDash \varphi \text{ iff } ((a_1, \ldots, a_n), h), \llbracket - \rrbracket \vDash^{\Gamma} \varphi$. \qed%
\end{thm}

After verifying that terms are evaluated to the same elements as the standard model in both representations, the equivalence of the clauses for atomic formulas can be computed directly. That the clauses for the multiplicatives $*$ and $\wand$ are equivalent is a consequence of the upwards and downwards closure of heap composition with respect to heap extension, as discussed in Section~\ref{sec:bunchedimplications}. Finally, the equivalence of the quantifier clauses is down to the representation of stores as vectors and the action of the product projections under the functor $\mathrm{Store}$. The notions of indexed (B)BI frame and its associated semantics are therefore a natural generalization of the standard Separation Logic model.

\subsection{Duality for (B)BI Hyperdoctrines}\label{subsec:bihyperdoctrines}

We now extend the results given for (B)BI algebras to (B)BI hyperdoctrines. For such results to make sense, both (B)BI hyperdoctrines and indexed (B)BI frames need to be equipped with a notion of morphism to form categories. Our definition of hyperdoctrine morphism adapts that for \emph{coherent} hyperdoctrines~\cite{COUMANS20121940}.

\begin{defi}[(B)BI Hyperdoctrine Morphism]
Given a pair of (B)BI hyperdoctrines $\mathbb{P}: \mathrm{C}^{op} \rightarrow \mathrm{Poset}$ and
$\mathbb{P}': \mathrm{D}^{op} \rightarrow \mathrm{Poset}$, a \emph{(B)BI hyperdoctrine morphism}
$(\mathrm{K}, \tau): \mathbb{P} \rightarrow \mathbb{P}'$ is a pair $(\mathrm{K}, \tau)$ satisfying the following properties:
	\begin{enumerate}
	\item $\mathrm{K}: \mathrm{C} \rightarrow \mathrm{D}$ is a finite product preserving functor;
	\item $\tau: \mathbb{P} \rightarrow \mathbb{P}' \circ K$ is a natural transformation;
	\item For all objects $X$ in $\mathrm{C}$: $\tau_{X \times X}(=_X) ={\ } ='_{K(X)}$;
	\item For all objects $\Gamma$ and $X$ in $\mathrm{C}$, the following squares commute:
		\[ \begin{array}{ll}
		 \begin{tikzcd}
		\mathbb{P}(\Gamma \times X) \arrow{r}{\tau_{\Gamma \times X}}
		\arrow{d}[swap]{\exists X_{\Gamma}} & \mathbb{P'}(K(\Gamma) \times K(X))
		\arrow{d}{\exists' {K(X)}_{K(\Gamma)}} \\
		\mathbb{P}(\Gamma) \arrow{r}{\tau_{\Gamma}} & \mathbb{P'}(K(\Gamma))
		\end{tikzcd} &
		\begin{tikzcd}
		\mathbb{P}(\Gamma \times X) \arrow{r}{\tau_{\Gamma \times X}}
		\arrow{d}[swap]{\forall X_{\Gamma}} & \mathbb{P'}(K(\Gamma) \times K(X))
		\arrow{d}{\forall' {K(X)}_{K(\Gamma)}} \\
		\mathbb{P}(\Gamma) \arrow{r}{\tau_{\Gamma}} & \mathbb{P'}(K(\Gamma))
		\end{tikzcd}
		\end{array} \]
	\end{enumerate}
The composition of BI hyperdoctrine morphisms $(K, \tau): \mathbb{P} \rightarrow \mathbb{P}'$ and $(K', \tau'): \mathbb{P}' \rightarrow \mathbb{P}''$ is given by $(K' \circ K, \tau'_{K(-)} \circ \tau)$. This yields a category $\mathrm{BIHyp}$. \qed%
\end{defi}
For indexed (B)BI frames the definition of morphism splits into two because of the weakening of equality to a preorder on the intuitionistic side. It is straightforward to show that the notion of indexed BI frame morphism collapses to that for indexed BBI frames when the preorders $\preccurlyeq$ are substituted for $=$.

\begin{defi}[Indexed BI Frame Morphism]
Given indexed BI frames $\mathcal{R}: \mathrm{C} \rightarrow \mathrm{BIFr}$ and
$\mathcal{R}': \mathrm{D} \rightarrow \mathrm{BIFr}$, an \emph{indexed BI frame morphism}
$(L, \lambda): \mathcal{R} \rightarrow \mathcal{R}'$ is a pair $(L, \lambda)$ such that:
	\begin{enumerate}
	\item $L: D \rightarrow C$ is a finite product preserving functor;
	\item $\lambda: \mathcal{R} \circ L \rightarrow \mathcal{R}'$ is a natural transformation;
	\item (Lift Property) If there exists $x$ and $y$
		such that $\mathcal{R}'(\Delta_X)(y) \preccurlyeq \lambda_{X \times X}(x)$ then there exists
		$y'$ such that $\mathcal{R}(\Delta_{L(X)})(y')
		\preccurlyeq x$;
	\item (Morphism Pseudo Epi)  If there exists $x$ and
		$y$ with
		$\mathcal{R}'(\pi_{\Gamma, X})(x) \preccurlyeq \lambda_{\Gamma}(y) $ then there exists
		$z$ such that
		$x \preccurlyeq\lambda_{\Gamma \times X}(z)$ and
		$\mathcal{R}(\pi_{L(\Gamma), L(X)})(z) \preccurlyeq y$.
	\end{enumerate}
The composition of indexed BI frame morphisms $(L', \lambda'): \mathcal{R}' \rightarrow \mathcal{R}''$ and $(L, \lambda): \mathcal{R} \rightarrow \mathcal{R}'$ is given by $(L \circ L', \lambda' \circ \lambda_{L'(-)})$. This yields a category $\mathrm{IndBIFr}$. \qed%
\end{defi}

\begin{defi}[Indexed BBI Frame Morphism]
For indexed BBI frames $\mathcal{R}: \mathrm{C} \rightarrow \mathrm{BBIFr}$ and
$\mathcal{R}': \mathrm{D} \rightarrow \mathrm{BBIFr}$, an \emph{indexed BBI frame morphism} $(L, \lambda): \mathcal{R} \rightarrow \mathcal{R}'$ is a pair $(L, \lambda)$ satisfying (1) and (2) of the previous definition as well as
	\begin{enumerate}[align=left]
	\item[($3'$)] (Lift Property$'$) if there exist $x$ and
		$y$ such that
		$\lambda_{X \times X}(x) =  \functorfont{R'}(\Delta_X)(y)$, then there exists
		$y'$ such that $\functorfont{R}((\Delta_{L(X)}))(y') = x$, and
	\item[($4'$)] (Quasi-Pullback) for all objects $\Gamma$ and $X$ in $\catfont{C}$, the following square
		is a quasi-pullback:
		\[\begin{tikzcd}
		\functorfont{R}(L(\Gamma) \times L(X)) \arrow{r}{\lambda_{\Gamma \times X}}
		\arrow{d}{\functorfont{R}(\pi_{L(\Gamma), L(X)})} & \functorfont{R}(\Gamma \times X)
		\arrow{d}{\functorfont{R'}(\pi_{\Gamma, X})} \\
		\functorfont{R}(L(\Gamma)) \arrow{r}{\lambda_{\Gamma}} & \functorfont{R}(\Gamma)
		\end{tikzcd}\]
	\end{enumerate}
The composition of indexed BBI frame morphisms $(L', \lambda'): \mathcal{R}' \rightarrow \mathcal{R}''$ and $(L, \lambda): \mathcal{R} \rightarrow \mathcal{R}'$ is given by $(L \circ L', \lambda' \circ \lambda_{L'(-)})$. This yields a category $\mathrm{Ind(B)BIFr}$. \qed%
\end{defi}

We can now show that the `algebraic' and `relational' semantics of \logicfont{FO(B)BI} correspond to each other by defining functorial transformations analogous to the complex algebras and prime filter frames of Section~\ref{sec:propintduality}. We show these are related by a representation theorem which is then extended to a duality via the introduction of topology.

To obtain complex hyperdoctrines, we straightforwardly compose an indexed frame with the appropriate complex algebra functor from (B)BI\@.

\begin{defi}[Complex (B)BI Hyperdoctrine]\label{def:complexhyperdoctrine}
Given an indexed BI frame $\mathcal{R}: \mathrm{C} \rightarrow \mathrm{BIFr}$, the \emph{complex hyperdoctrine} of $\mathcal{R}$, $ComHyp^{\vtiny{BI}}(\mathcal{R})$, is given by
$Com^{\vtiny{BI}}(\mathcal{R}(-)): \mathrm{C}^{op} \rightarrow \mathrm{Poset}$, together with $Ran(\mathcal{R}(\Delta_X))$ as $=_X$, ${\mathcal{R}(\pi_{\Gamma, X})}^{*}$ as $\exists X_{\Gamma}$, and ${\mathcal{R}(\pi_{\Gamma, X})}_{*}$ as $\forall X_{\Gamma}$, where
	\[ \begin{array}{ll}
	{\mathcal{R}(\pi_{\Gamma, X})}^{*}(A) &= \{ x \mid \text{there exists } y \in A: \mathcal{R}(\pi_{\Gamma, X})(y) \preccurlyeq x \} \text{ and}\\
	{\mathcal{R}(\pi_{\Gamma, X})}_{*}(A) &= \{ x \mid \text{for all } y, \text{ if } x \preccurlyeq \mathcal{R}(\pi_{\Gamma, X})(y) \text{ then } y \in A \}.
	\end{array} \]
For indexed BBI frames, the definitions of ${\mathcal{R}(\pi_{\Gamma, X})}^{*}$ and ${\mathcal{R}(\pi_{\Gamma, X})}_{*}$ are as above, except with $\preccurlyeq$ replaced with $=$. \qed%
\end{defi}

Given that the complex algebra operations thus far have matched the corresponding semantic clauses on frames, one might have expected $\exists X_{\Gamma}$ to be given by the direct image $\mathcal{R}(\pi_{\Gamma, X})$. Using the fact that $\mathcal{R}(\pi_{\Gamma, X})$ is a BI morphism it can be shown that ${\mathcal{R}(\pi_{\Gamma, X})}^{*}$ is in fact identical to $\mathcal{R}(\pi_{\Gamma, X})$ so this is indeed the case. We use its presentation as ${\mathcal{R}(\pi_{\Gamma, X})}^{*}$ as it simplifies some proofs that follow.

\begin{lem}\label{lemma:comlexisirf}
Given an indexed (B)BI frame $\functorfont{R}: \catfont{C} \rightarrow \catfont{(B)BIFr}$, the
complex hyperdoctrine $ComHyp^{\vtiny{(B)BI}}(\mathcal{R})$ is a (B)BI hyperdoctrine.
\end{lem}

\begin{proof}
We concentrate on the verifications relating to ${\mathcal{R}(\pi_{\Gamma, X})}^{*}$ and ${\mathcal{R}(\pi_{\Gamma, X})}_{*}$ of a BI complex hyperdoctrine. It is straightforward to see these map upwards-closed sets to upwards-closed sets and are monotone with respect to the subset ordering $\subseteq$. The adjointness properties follow from the definitions so it just remains to prove naturality.

We give the case for $\exists X_{\Gamma}$. Given a morphism $s: \Gamma \rightarrow \Gamma'$ in $\mathrm{C}$ and an element $A \in Com^\vtiny{BI}(\mathcal{R}(\Gamma' \times X))$, we must show
${\mathcal{R}(\pi_{\Gamma, X})}^{*}({\mathcal{R}(s \times id_X)}^{-1}(A)) =
{\mathcal{R}(s)}^{-1}({\mathcal{R}(\pi_{\Gamma', X})}^{*}(A))$. Suppose
$x \in \mathcal{R}^{*}(\pi_{\Gamma, X})({\mathcal{R}(s \times id_X)}^{-1}(A))$: then there exists $y$ such that $\mathcal{R}(\pi_{\Gamma, X})(y) \preccurlyeq x$ and $\mathcal{R}(s \times id_X)(y) \in A$. We have $\mathcal{R}(\pi_{\Gamma', X})(\mathcal{R}(s \times id_X)(y)) =
\mathcal{R}(s)(\mathcal{R}(\pi_{\Gamma, X})(y)) \preccurlyeq \mathcal{R}(s)(x)$. Hence
$x \in {\mathcal{R}(s)}^{-1}({\mathcal{R}(\pi_{\Gamma', X})}^{*}(A))$, as required.

Conversely, assume $x \in {\mathcal{R}(s)}^{-1}({\mathcal{R}(\pi_{\Gamma', X})}^{*}(A))$.  Then there exists $y \in A$ such that $\mathcal{R}(\pi_{\Gamma', X})(y) \preccurlyeq \mathcal{R}(s)(x)$. Then by (Psuedo Epi), there exists $z$ such that $\mathcal{R}(\pi_{\Gamma, X})(z) \preccurlyeq x$ and $y \preccurlyeq \mathcal{R}(s \times id_X)(z)$. By upwards-closure of $A$, $\mathcal{R}(s \times id_X)(z) \in A$. Hence we have
$x \in {\mathcal{R}(\pi_{\Gamma, X})}^{*}({\mathcal{R}(s \times id_X)}^{-1}(A))$, as required.

The proof for BBI complex hyperdoctrines follows immediately by substituting every instance of $\preccurlyeq$ with $=$ in the above argument, where the quasi pullback property allows us to assume the existence of $z$ such that $\mathcal{R}(\pi_{\Gamma, X})(z) = x$ and $y = \mathcal{R}(s \times {id}_X)(z)$ from  $\mathcal{R}(\pi_{\Gamma', X})(y) = \mathcal{R}(s)(x)$.
\end{proof}

\begin{defi}[Indexed Prime Filter (B)BI Frame]\label{def:slultrafilter}
Given a (B)BI hyperdoctrine
$\algfont{P}$,
the \emph{indexed prime filter frame}, $IndPr^{\vtiny{(B)BI}}(\mathbb{P})$, is given by $Pr^\vtiny{(B)BI}(\algfont{P}(-))$. \qed
\end{defi}

\begin{lem}\label{lemma:ultrafilterisirf}
Given a (B)BI hyperdoctrine
$\algfont{P}: \catfont{C}^{op} \rightarrow \catfont{Poset}$,
the indexed prime filter frame $IndPr^{\vtiny{(B)BI}}(\mathbb{P})$ is an indexed (B)BI frame.
\end{lem}
\begin{proof}
We first show the Pseudo Epi property is satisfied in the \logicfont{FOBI} case. Assume we have objects $\Gamma, \Gamma'$ and $X$ in $\mathrm{C}$ and a morphism
$s: \Gamma \rightarrow \Gamma'$. Let prime filters $F_x$ and $F_y$ be such that
${\mathbb{P}(\pi_{\Gamma', X})}^{-1}(F_y) \subseteq {\mathbb{P}(s)}^{-1}(F_x)$. It is easy to see that
\[\PrimePredicate{F}{{\mathbb{P}(\pi_{\Gamma, X})}^{-1}(F) \subseteq F_x \text{ and } F_y \subseteq {\mathbb{P}(s \times id_X)}^{-1}(F)}\]
is a prime predicate. The only non-trivial verification is showing $P(F \cap F') = 1$ implies $P(F) =1$ or $P(F') =1$. Suppose $P(F \cap F') =1, P(F) = 0$ and $P(F') =0$. Necessarily there exists $a$ and $b$ such that $\mathbb{P}(\pi_{\Gamma, X})(a) \in F$, $\mathbb{P}(\pi_{\Gamma, X})(b) \in F'$ and $a, b \not\in F_x$. Then $\mathbb{P}(\pi_{\Gamma, X})(a) \lor \mathbb{P}(\pi_{\Gamma, X})(b) = \mathbb{P}(\pi_{\Gamma, X})(a \lor b) \in F \cap F'$ so $a \lor b \in F_x$. However $F_x$ is prime, so $a \in F_x$ or $b \in F_x$, a contradiction.

Consider the filter $F = \filter{\mathbb{P}(s \times id_X)(F_y)}$ and suppose for contradiction it is not proper. 
This entails there exists $a \in F_y$ such that $\mathbb{P}(s \times id_X)(a) = \bot$. By adjointness, $\exists X_{\Gamma}(\bot) = \bot$, so
$\mathbb{P}(s)(\exists X_{\Gamma'}(a)) = \exists X_{\Gamma}(\mathbb{P}(s \times id_X)(a)) = \bot$ by naturality.
This entails $\exists X_{\Gamma'}(a) \not\in {\mathbb{P}(s)}^{-1}(F_x)$ so
$\exists X_{\Gamma'}(a) \not\in {\mathbb{P}(\pi_{\Gamma', X})}^{-1}(F_y)$ by assumption. However, by adjointness and filterhood, $\mathbb{P}(\pi_{\Gamma', X})(\exists X_{\Gamma'}(a)) \in F_y$, a contradiction.

Clearly ${\mathbb{P}(s \times id_X)}^{-1}(F) \supseteq F_y$. To see that the other required inclusion holds, suppose
$a \in {\mathbb{P}(\pi_{\Gamma, X})}^{-1}(F)$. Then there exists $b \in F_y$ such that
$\mathbb{P}(s \times id_X)(b) \leq \mathbb{P}(\pi_{\Gamma, X})(a)$. By adjointness
$\exists X_{\Gamma}(\mathbb{P}(s \times id_X)(b)) \leq a$ and so by naturality
$\mathbb{P}(s)(\exists X_{\Gamma'}(b)) \leq a$. Since $\mathbb{P}(\pi_{\Gamma', X})(\exists X_{\Gamma'}(b)) \in F_y$, we have $\exists X_{\Gamma'}(b) \in {\mathbb{P}(\pi_{\Gamma', X})}^{-1}(F_y) \subseteq {\mathbb{P}(s)}^{-1}(F_x)$. Thus by filterhood, $a \in F_x$. Thus $P(F) =1$ and by the prime extension lemma we have a prime $F$ with $P(F) =1$, as required.

For \logicfont{FOBBI}, we instead start with the assumption of prime filters $F_x$ and $F_y$ such that ${\mathbb{P}(\pi_{\Gamma', X})}^{-1}(F_y) = {\mathbb{P}(s)}^{-1}(F_x)$. This is sufficient to once again prove the existence of a prime filter $F$ satisfying ${\mathbb{P}(s \times id_X)}^{-1}(F) \supseteq F_y$ and ${\mathbb{P}(\pi_{\Gamma, X})}^{-1}(F) \subseteq F_x$. However, maximality of prime filters on Boolean algebras collapses the inclusions to equalities --- ${\mathbb{P}(s \times id_X)}^{-1}(F) = F_y$ and ${\mathbb{P}(\pi_{\Gamma, X})}^{-1}(F) = F_x$ --- so the quasi pullback property holds.
\end{proof}

We now lift the representation theorem for (B)BI algebras to (B)BI hyperdoctrines, making essential use of the natural transformation $\theta$ used in the (B)BI duality theorems of the previous section.

\begin{thm}[Representation Theorem for (B)BI Hyperdoctrines]\label{thm:hyperdoctrinerep}
Every (B)BI hyperdoctrine $\mathbb{P}: \mathrm{C}^{op} \rightarrow \mathrm{Poset}$ can be embedded in a complex (B)BI hyperdoctrine. That is, $\Theta_{\mathbb{P}}: \mathbb{P} \rightarrow Com^{\vtiny{(B)BI}}Pr^{\vtiny{(B)BI}}(\mathbb{P}(-))$, defined $(Id_{\mathrm{C}}, \theta_{\mathbb{P}(-)})$, is a monomorphism. 
\end{thm}

\begin{proof}
Clearly, by the representation theorem for (B)BI algebras, each component of $\Theta_{\mathbb{P}}$ is mono, and hence $\Theta_{\mathbb{P}}$ is mono. It remains to show that $\Theta_{\mathbb{P}}$ is a (B)BI hyperdoctrine morphism. That ${Id}_{\mathrm{C}}$ preserves finite products is immediate and that $\theta_{\mathbb{P}(-)}: \mathbb{P} \rightarrow Com^{\vtiny{(B)BI}}Pr^{\vtiny{(B)BI}}(\mathbb{P}(-))$ is a natural transformation is given by (B)BI duality. 

First we show that property (3) of (B)BI hyperdoctrine morphisms holds. We must verify that $\theta_{\mathbb{P}(X \times X)}(=_{X}) = Ran({\mathbb{P}(\Delta_X)}^{-1})$ for any object $X$ of $\mathrm{C}$. First suppose $F = {\mathbb{P}(\Delta_X)}^{-1}(G)$ for some prime filter $G$. By adjointness of $=_X$ at $\top$ we have that $\mathbb{P}(\Delta_X)(=_{X}) = \top \in G$. Hence $=_{X} \in F$. Conversely, assume $=_{X} \in F$. Straightforwardly we have that
\[\PrimePredicate{G}{{\mathbb{P}(\Delta_X)}^{-1}(G) \subseteq F}\]
defines a prime predicate. By the adjointness property of $=_{X}$ we have that ${\mathbb{P}(\Delta_X)}^{-1}(\{ \top \}) \subseteq F$. Hence there exists prime $G$ with ${\mathbb{P}(\Delta_X)}^{-1}(G) \subseteq F$ by the prime extension lemma. For the case of \logicfont{FOBI}, since ${\mathbb{P}(\Delta_X)}^{-1}$ is a BI morphism there then exists $G' \supseteq G$ with ${\mathbb{P}(\Delta_X)}^{-1}(G') = F$; for the case of \logicfont{FOBBI}, maximality of prime filters means ${\mathbb{P}(\Delta_X)}^{-1}(G) = F$. In both cases, $F \in Ran({\mathbb{P}(\Delta_X)}^{-1})$ as required.

For property (4), we verify the naturality diagram for $\exists X_{\Gamma}$: the verification of $\forall X_{\Gamma}$ is similar. The verification reduces to showing that, given a prime filter $F$ of $\mathbb{P}(\Gamma)$ and $a \in \mathbb{P}(\Gamma \times X)$, $\exists X_{\Gamma}(a) \in F$ iff there exists $G$ such that $a \in G$ and ${\mathbb{P}(\pi_{\Gamma, X})}^{-1}(G) \subseteq F$: for $\logicfont{FOBI}$ this corresponds precisely to commutativity of the diagram, and for $\logicfont{FOBBI}$ we can conclude ${\mathbb{P}(\pi_{\Gamma, X})}^{-1}(G) = F$ by maximality of prime filters, yielding commutativity of the appropriate diagram for that case.

First assume $\exists X_{\Gamma}(a) \in F$. It is straightforward to see that
\[ \PrimePredicate{G}{{\mathbb{P}(\pi_{\Gamma, X})}^{-1}(G) \subseteq F}\] defines a prime predicate. Consider $G = \filter{a}$. This is proper, as otherwise $a = \bot$, which would entail $\exists X_{\Gamma}(a) = \bot \in F$, contradicting that $F$ is a prime (and thus proper) filter. Let $\mathbb{P}(\pi_{\Gamma, X})(b) \geq a$. Then by adjointness, $\exists X_{\Gamma}(a) \leq b \in F$. Hence $P(G) =1$ and so there exists a prime filter $G$ with $P(G) =1$, as required. Now assume $a \in G$ and ${\mathbb{P}(\pi_{\Gamma, X})}^{-1}(G) \subseteq F$. By adjointness $a \leq \mathbb{P}(\pi_{\Gamma, X})(\exists X_{\Gamma}(a)) \in G$, so $\exists X_{\Gamma}(a) \in {\mathbb{P}(\pi_{\Gamma, X})}^{-1}(G) \subseteq F$ as required.
\end{proof}

Just as in the propositional case, the representation theorem yields completeness for the indexed frame semantics. Given any interpretation on an indexed (B)BI frame $\Interp{-}$, we automatically have an interpretation for the complex hyperdoctrine as predicate symbols are interpreted as (upwards-closed) subsets; that is, elements of complex algebras of (B)BI frames.  A simple inductive argument shows that satisfaction coincides for these models.

\begin{prop}
Given an indexed (B)BI frame $\mathcal{R}$ and an interpretation $\Interp{-}$, for all \logicfont{FO(B)BI} formulae $\varphi$ in context $\Gamma$ and $x \in \mathcal{R}(\Interp{\Gamma})$, $x, \Interp{-} \vDash^{\Gamma} \varphi$ iff $x \in \Interp{\varphi}$. \qed%
\end{prop}

Similarly, given an interpretation $\Interp{-}$ on a (B)BI hyperdoctrine, we can define the interpretation $\widetilde{\Interp{-}}$ by setting $\widetilde{\Interp{P}} = \theta_{\mathbb{P}(\Interp{X_1} \times \cdots \times \Interp{X_m})}(\Interp{P})$ for each predicate symbol of type $X_1, \ldots, X_m$. As a corollary of the representation theorem we obtain the following proposition.

\begin{prop}
Given a (B)BI hyperdoctrine $\mathbb{P}$ and an interpretation $\Interp{-}$, for all \logicfont{FO(B)BI} formulae $\varphi$ in context $\Gamma$ and prime filters $F$ of $\mathbb{P}(\Interp{\Gamma})$, $\Interp{\varphi} \in F$ iff $F, \widetilde{\Interp{-}} \vDash^{\Gamma} \varphi$. \qed%
\end{prop}

\begin{thm}[Soundness and Completeness for Indexed (B)BI frames]
For all \logicfont{FO(B)BI} formulae $\varphi$ in context $\Gamma$, $\varphi \vdash^{\Gamma} \psi$ is provable iff $\varphi \vDash^{\Gamma} \psi$. \qed%
\end{thm}

From here it is straightforward to set an assignment of morphisms to make the assignment of complex hyperdoctrines and indexed prime filter frames functorial. Given a (B)BI hyperdoctrine morphism $(K, \tau): \mathbb{P} \rightarrow \mathbb{P}'$, $IndPr^{\vtiny{(B)BI}}(K, \tau) = (K, \tau^{-1})$. Similarly, given an indexed (B)BI frame morphism $(L, \lambda): \mathcal{R} \rightarrow \mathcal{R}'$, $ComHyp^{\vtiny{(B)BI}}(L, \lambda) = (L, \lambda^{-1})$.

\begin{lem}
The functors $ComHyp^{\vtiny{(B)BI}}$ are well defined.
\end{lem}

\begin{proof}
Let $(L, \lambda)$ be a indexed (B)BI frame morphism. First note that by definition $L$ is a finite product preserving functor. We also have that each component $\lambda_X: \mathcal{R}(LX) \rightarrow \mathcal{R}'(X)$ is a (B)BI morphism. Hence by functorality of $Com^{\vtiny(B)BI}$, each $\lambda^{-1}_X: Com^{\vtiny{(B)BI}}(\mathcal{R}'(X)) \rightarrow Com^{\vtiny{(B)BI}}(\mathcal{R}(LX))$ is a (B)BI algebra homomorphism, and naturality is inherited from $\lambda$.

Next we must verify that $\lambda^{-1}_{X \times X}(Ran(\mathcal{R}'(\Delta_X))) = Ran(\mathcal{R}(\Delta_{LX}))$. The right-to-left inclusion follows immediately from naturality of $\lambda$. For the left-to-right, suppose $\lambda_{X \times X}(x) \in Ran(\mathcal{R}'(\Delta_X))$. Then there exists $y$ such that $\lambda_{X \times X}(x) = \mathcal{R}'(\Delta_X)(y)$. In the case for \logicfont{FOBI}, by the lift property, there exists $y'$ such that $\mathcal{R}(\Delta_{LX})(y') \preccurlyeq x$. Since $\mathcal{R}(\Delta_{LX})$ is a BI morphism, there thus exists $x'$ such that $y' \preccurlyeq x'$ and $\mathcal{R}(\Delta_{LX})(x') = x$ as required. For \logicfont{FOBBI} we are given such an $x'$ immediately by the respective lift property.

Finally we verify the commutative diagram for $\exists X_{\Gamma}$, leaving the similar verification for $\forall X_{\Gamma}$ to the reader. We must show that ${\mathcal{R}(\pi_{L\Gamma, LX})}^{*}\lambda^{-1}_{\Gamma \times X}(A) = \lambda^{-1}_{\Gamma}{\mathcal{R}'(\pi_{\Gamma, X})}^{*}(A)$ for $A \in Com^{\vtiny{(B)BI}}(\mathcal{R}'(\Gamma \times X))$. First consider the case of \logicfont{FOBI}. Suppose $x \in {\mathcal{R}(\pi_{L\Gamma, LX})}^{*}\lambda^{-1}_{\Gamma \times X}(A)$. Then there exists $y$ with $\lambda_{\Gamma \times X}(y) \in A$ and $\mathcal{R}(\pi_{L\Gamma, LX})(y) \preccurlyeq x$. Since $\lambda$ is a natural transformation and its components are order-preserving we have $\mathcal{R}'(\pi_{\Gamma, X})(\lambda_{\Gamma \times X}(y)) = \lambda_{\Gamma}\mathcal{R}(\pi_{L\Gamma, LX})(y) \preccurlyeq \lambda_{\Gamma}(x)$, so $x \in  \lambda^{-1}_{\Gamma}{\mathcal{R}'(\pi_{\Gamma, X})}^{*}(A)$. Now, suppose $x \in \lambda^{-1}_{\Gamma}{\mathcal{R}'(\pi_{\Gamma, X})}^{*}(A)$. Then $\mathcal{R}'(\pi_{\Gamma, X})(y) \preccurlyeq \lambda_{\Gamma}(x)$ for $y \in A$. By the Morphism Pseudo Epi property, there exists $z$ such that $y \preccurlyeq \lambda_{\Gamma \times X}(z)$ and $\mathcal{R}(\pi_{L\Gamma, LX})(z) \preccurlyeq x$. $A$ is an upwards-closed set so $\lambda_{\Gamma \times X}(z) \in A$, hence $x \in  {\mathcal{R}(\pi_{L\Gamma, LX})}^{*}\lambda^{-1}_{\Gamma \times X}(A)$ as required. For the case of \logicfont{FOBBI} the same argument applies, where $\preccurlyeq$ is substituted for $=$ and the other Morphism Pseudo Epi property is applied to find a sufficient $z$ in the right-to-left direction.
\end{proof}

\begin{lem}
The functors $IndPr^{\vtiny{(B)BI}}$ are well defined.
\end{lem}

\begin{proof}
Let $(K, \tau)$ be a (B)BI hyperdoctrine morphism. As in the previous lemma, we automatically obtain properties (1) and (2) for $(K, \tau^{-1})$ from the definition and the complex algebra functor $Com^{\vtiny{(B)BI}}$. For properties (3) and (4) we verify the case for \logicfont{FOBI}, obtaining the case for \logicfont{FOBBI} as a special case.

First we consider the Lift Property. Suppose ${\mathbb{P}(\Delta_X)}^{-1}(G) \subseteq \tau^{-1}_{X \times X}(F)$. It is simple to see that
\[ \PrimePredicate{G'}{{\mathbb{P}'(\Delta_{KX})}^{-1}(G') \subseteq F} \] is a prime predicate. By Theorem~\ref{thm:hyperdoctrinerep} we have that $=_X \in {\mathbb{P}(\Delta_X)}^{-1}(G) \subseteq \tau^{-1}_{X \times X}(F)$ so $\tau_{X \times X}(=_X) ={\ } ='_{KX} \in F$. By adjointness it then follows that ${\mathbb{P}'(\Delta_{KX})}^{-1}(\{ \top \}) \subseteq F$, as $\mathbb{P}(\Delta_{KX})(a) = \top$ entails $='_{KX} \leq a$. By the prime extension lemma there thus exists prime $G$ with $P(G) =1$, as required.  For the case of \logicfont{FOBBI}, maximality entails ${\mathbb{P}'(\Delta_{KX})}^{-1}(G') = F$.

Next, the Morphism Pseudo Epi property. Suppose ${\mathbb{P}(\pi_{\Gamma, X})}^{-1}(F) \subseteq \tau^{-1}_{\Gamma}(G)$. We can once again define a prime predicate
\[ \PrimePredicate{G'}{F \subseteq \tau^{-1}_{\Gamma \times X}(G') \text{ and } {\mathbb{P}'(\pi_{K\Gamma, KX})}^{-1}(G') \subseteq G} \] that we can use to prove the existence of the appropriate prime filter. Consider the filter $G' = \filter{\tau_{\Gamma \times X}(F)}$. This is proper, otherwise there exists $a\in F$ such that $\tau_{\Gamma \times X}(a) = \bot$. By property (4) of BI hyperdoctrine morphism, this would entail $\tau_{\Gamma}(\exists X_{\Gamma}(a)) = \exists' KX_{K\Gamma}\tau_{\Gamma \times X}(a) = \exists' KX_{K\Gamma}(\bot) = \bot$. Since $a \in F$, by adjointness $\mathbb{P}(\pi_{\Gamma, X})(\exists X_{\Gamma}(a)) \in F$. Hence $\tau_{\Gamma}(\exists X_{\Gamma}(a)) = \bot \in G$ by assumption, contradicting that $G$ is a prime filter.

Clearly $F \subseteq \tau^{-1}_{\Gamma \times X}(G')$. Further, let $b \in {\mathbb{P}'(\pi_{K\Gamma, KX})}^{-1}(G')$. Then there exists $a \in F$ such that $\mathbb{P}'(\pi_{K\Gamma, KX})(b) \geq \tau_{\Gamma \times X}(a)$. By adjointness it follows that $\exists KX_{K\Gamma}( \tau_{\Gamma \times X}(a)) \leq b$, and by property (4) of BI hyperdoctrine morphisms we have $\tau_{\Gamma} (\exists X_{\Gamma}(a)) \leq b$. Since $a \in F$ we have that $\mathbb{P}(\pi_{\Gamma, X})(\exists X_{\Gamma}(a)) \in F$ by adjointness and upwards closure, hence $\exists X_{\Gamma}(a) \in {\mathbb{P}(\pi_{\Gamma, X})}^{-1}(F) \subseteq \tau^{-1}_{\Gamma}(G)$. It follows that $\tau_{\Gamma}(\exists X_{\Gamma}(a)) \in G$, and so $b \in G$. Thus $P(G) =1$ and so there exists a prime filter $G$ with $P(G) = 1$ by the prime extension lemma. In the case for \logicfont{FOBBI}, by maximality of prime filters these inclusions become equalities, and this yields a witness for the Quasi-Pullback property.
\end{proof}

At this stage topology must be introduced to yield a dual equivalence.

\begin{defi}[Indexed (B)BI Space]
An \emph{indexed  (B)BI space} is a functor $\mathcal{R}: \mathrm{C} \rightarrow \mathrm{(B)BISp}$ such that
\begin{enumerate}
\item $U \circ \mathcal{R}: \mathrm{C} \rightarrow \mathrm{(B)BIFr}$ is an indexed (B)BI frame, where $U: \mathrm{(B)BISp} \rightarrow \mathrm{(B)BIFr}$ is the functor that forgets topological structure.
\item For each object $X$ in $\mathrm{C}$, $Ran(\mathcal{R}(\Delta_X))$ is clopen;
\item For each pair of objects $\Gamma$ and $X$ in $\mathrm{C}$,
		${\mathcal{R}(\pi_{\Gamma, X})}^{*}$ and
		${\mathcal{R}(\pi_{\Gamma, X})}_{*}$ map (upwards-closed) clopen sets to (upwards-closed) clopen sets. \qed%
\end{enumerate}
\end{defi}

\noindent
In the case for \logicfont{FOBBI} it is possible to weaken condition (3) to ${\mathcal{R}(\pi_{\Gamma, X})}^{*}$ being an open map and ${\mathcal{R}(\pi_{\Gamma, X})}_{*}$ a closed map. This is because $\mathcal{R}(\pi_{\Gamma, X})$ is a continuous map between a compact and a Hausdorff space, and so the direct image $\mathcal{R}(\pi_{\Gamma, X}) = {\mathcal{R}(\pi_{\Gamma, X})}^{*}$ is a closed map automatically. We also have that ${\mathcal{R}(\pi_{\Gamma, X})}_{*}$ is an open map by definition. In the intuitionistic case the same reasoning applies for ${\mathcal{R}(\pi_{\Gamma, X})}^{*}$ (using its equivalence with the direct image) but it is not clear how to make the analogous case for ${\mathcal{R}(\pi_{\Gamma, X})}_{*}$. Nonetheless, this definition of indexed (B)BI space gives us what we need.

\begin{lem}
Given a (B)BI hyperdoctrine $\mathbb{P}$, the indexed prime filter space $IndPr^{\vtiny{(B)BI}}(\mathbb{P})$ is an indexed (B)BI space.
\end{lem}

\begin{proof}
Given Lemma~\ref{lemma:ultrafilterisirf} the only verifications left are of properties (2) and (3). We immediately obtain (2) by noting once again that $Ran({\mathbb{P}(\Delta_X)}^{-1}) = \theta_{\mathbb{P}(X \times X)}(=_X)$, a clopen set by (B)BI duality. Utilising (B)BI duality once more, we have that every (upwards-closed) clopen set of $Pr^{\vtiny{(B)BI}}(\mathbb{P}(Y))$ is of the form $\theta_{\mathbb{P}(Y)}(a)$ for some $a \in \mathbb{P}(Y)$. We thus demonstrate that ${({\mathbb{P}(\pi_{\Gamma, X})}^{-1})}^{*}(\theta_{\mathbb{P}(\Gamma \times X)}(a)) = \theta_{\mathbb{P}(\Gamma)}(\exists X_{\Gamma}(a))$ and ${({\mathbb{P}(\pi_{\Gamma, X})}^{-1})}_{*}(\theta_{\mathbb{P}(\Gamma \times X)}(a)) = \theta_{\mathbb{P}(\Gamma)}(\forall X_{\Gamma}(a))$. 

First assume $F \in {({\mathbb{P}(\pi_{\Gamma, X})}^{-1})}^{*}(\theta_{\mathbb{P}(\Gamma \times X)}(a))$. Then there exists $F'$ such that $a \in F'$ and ${\mathbb{P}(\pi_{\Gamma, X})}^{-1}(F') \subseteq F$. By adjointness $a \leq \mathbb{P}(\pi_{\Gamma, X})(\exists X_{\Gamma}(a))$ so $\exists X_{\Gamma}(a) \in
{\mathbb{P}(\pi_{\Gamma, X})}^{-1}(F') \subseteq F$ so $F \in \theta_{\mathbb{P}(\Gamma)}(\exists X_{\Gamma}(a))$ as required. Conversely, suppose $\exists X_{\Gamma}(a) \in F$. It is easy to see that
\[ \PrimePredicate{G}{{\mathbb{P}(\pi_{\Gamma, X})}^{-1}(G) \subseteq F \text{ and } a \in G }\]
is a prime predicate. Consider the filter $G = \filter{a}$. $F$ is proper as $a \not= \bot$ (otherwise $\exists X_{\Gamma}(\bot) = \bot \in F$), and by adjointness, if $a \leq \mathbb{P}(\pi_{\Gamma, X})(b)$, it follows that $\exists X_{\Gamma}(a) \leq b \in F$. Hence $P(G) =1$ and by the prime extension lemma there exists a prime $G$ with $P(G) =1$ as required. In the case for \logicfont{FOBI} we're done; in the case for \logicfont{FOBBI}, maximality of prime filters makes the inclusion an equality.

For the other equality, first assume we have $F$ with $\forall X_{\Gamma}(a) \in F$ and let $F \subseteq {\mathbb{P}(\pi_{\Gamma, X})}^{-1}(G)$. Then $\mathbb{P}(\pi_{\Gamma, X})(\forall X_{\Gamma}(a)) \in G$, and by adjointness and upwards closure of $G$ we have $a \in G$. In the other direction, assume $\forall X_{\Gamma}(a) \not\in F$. We show there exists a prime filter $G$ such that $F \subseteq {\mathbb{P}(\pi_{\Gamma, X})}^{-1}(G)$ and $a \not\in G$. First note that for proper ideals $I$,
\[ \PrimePredicate{I}{F \subseteq {\mathbb{P}(\pi_{\Gamma, X})}^{-1}(\overline{I}) \text{ and } a \in I} \] is a prime predicate. Consider $I = \ideal{a}$. This is proper as $a \not= \top$, as otherwise $\forall X_{\Gamma}(a) = \top \not\in F$, contradicting that $F$ is a filter. Suppose $b \in F$. Then $\mathbb{P}(\pi_{\Gamma, X}(b)) \not\leq a$ as otherwise by adjointness $b \leq \forall X_{\Gamma}(a) \in F$. Thus $P(I) = 1$ and so there exists a prime ideal $I$ such that $P(I) =1$. The prime filter $G = \overline{I}$ gives the required witness to the inclusion. Once again, in the case of \logicfont{FOBBI} maximality ensures the inclusion of prime filters is equality.
\end{proof}

In the other direction, composing an indexed (B)BI space $\mathcal{R}$ with the clopen algebra functor $Clop^{\vtiny{(B)BI}}$ yields the clopen hyperdoctrine $ClopHyp^{\vtiny{(B)BI}}(\mathcal{R})$. Conditions (2) and (3) of indexed (B)BI space ensure that the assignment of $Ran(\mathcal{R}(\Delta_X))$ as $=_X$, ${\mathcal{R}(\pi_{\Gamma, X})}^{*}$ as $\exists X_{\Gamma}$, and ${\mathcal{R}(\pi_{\Gamma, X})}_{*}$ as $\forall X_{\Gamma}$ is well defined, and Lemma~\ref{lemma:comlexisirf} suffices to show that they satisfy the required properties. The definition of indexed (B)BI space morphism is given by taking that for indexed (B)BI frames.  Then the assignment of morphisms given by the indexed prime filter frame and complex hyperdoctrine functors works the same way as before.

It remains to specify the natural isomorphisms that form the dual equivalence of categories. We already have $\Theta: Id_{\mathrm{(B)BIHyp}} \rightarrow ClopHyp^{\vtiny{(B)BI}}IndPr^{\vtiny{(B)BI}}$ from the representation theorem. We additionally define $\mathrm{H}: Id_{\mathrm{Ind(B)BISp}} \rightarrow IndPr^{\vtiny{(B)BI}}ClopHyp^{\vtiny{(B)BI}}$ by $\mathrm{H}_{\mathcal{R}} = (Id_C, \eta_{\mathcal{R}(-)})$, where $\eta: Id_{\mathrm{(B)BISp}} \rightarrow Pr^{\vtiny{(B)BI}}Clop^{\vtiny{(B)BI}}_{\succcurlyeq}$ is the natural isomorphism given by (B)BI duality. It is straightforward (using the underlying dualities) to show that the components of $\mathrm{H}$ are indeed indexed (B)BI space morphisms, and that $\Theta$ and $\eta$ are natural isomorphisms. 

\begin{thm}[Duality Theorem for \logicfont{FO}(\logicfont{B})\logicfont{BI}]
$\Theta$ and $\mathrm{H}$ form a dual equivalence of categories between $\mathrm{(B)BIHyp}$ and $\mathrm{Ind(B)BISp}$. \qed%
\end{thm}

Given that the structures we have defined comprise the most general classes of bunched logic model, to what extent can we restrict the dualities to a subclass of models that correspond specifically to the standard model of Separation Logic? There are a number of properties that the (classical) memory model satisfies implicitly, given as follows by Brotherston \& Villard~\cite{Brotherston2014hybrid}:

\medskip

\begin{tabular}{ll}
Partial deterministic: & $w, w' \in w_1 \circ w_2$ implies $w = w'$ \\
Cancellativity: & $w \circ w_1 \cap w \circ w_2$ implies $w_1 = w_2$ \\
Indivisible Units: & $(w \circ w') \cap E \neq \emptyset$ \\
Disjointness: & $w \circ w \neq \emptyset$ implies $w \in E$ \\
Divisibility: & for every $w \not\in E$, there are $w_1, w_2 \not\in E$ such that $w \in w_1 \circ w_2$ \\
Cross Split: & whenever $(t \circ u) \cap (v \circ w) \neq \emptyset$, there exist $tv, tw, uv, uw$ such \\
                    & that $t \in tv \circ tw, u \in uv \circ uw, v \in tv \circ uv$ and $w \in tw \circ uw$. \\
\end{tabular}

\medskip

\noindent Here we do not consider their property Single Unit as it is only satisfied by the propositional heap model, and not by the predicate store-heap models. Brotherston \& Villard show that while Divisibility and Indivisible Units are definable in \logicfont{BBI}, Partial Deterministic, Cancellativity and Disjointness are not, with Cross Split conjectured to be similarly undefinable.

It is straightforward to restrict the dual adjunction and duality theorems to the algebras satisfying axioms corresponding to Indivisible Units ($\mtop \land (a * b) \leq a$) and Divisibility ($\neg \mtop \leq \neg \mtop * \neg \mtop$). However, in the other cases the undefinability results preclude this possibility: no algebraic axiom can possibly pick out these classes of model.

The extent to which the remaining properties define different notions of validity has been partially investigated by Larchey-Wendling \& Galmiche~\cite{Larchey-Wendling2014}. In particular, they show that the formula $\mathcal{I} * \mathcal{I} \rightarrow \mathcal{I}$, where $\mathcal{I} = \neg (\top \wand \neg \mtop)$, distinguishes partial deterministic models from non-deterministic models. In summary, this situation isn't totally benign.

Brotherston \& Villard's solution is to give a conservative extension of \logicfont{BBI} in the spirit of hybrid logic~\cite{AC2006} called \logicfont{HyBBI}. The additional expressivity of nominals and satisfaction operators allows the logic to pick out specific states of the model, making the axiomatization of the remaining properties possible. Thus, in order to precisely capture the concrete model of Separation Logic we would have to extend the techniques of the present work to hybrid extensions of bunched logics. This would also enable us to connect our work to the extensive tableaux proof theory given for bunched logics~\cite{GMP05,Larchey-Wendling16,LSM,epistemic} via the close connection between hybrid extensions and labelled proof systems. In a first step in this direction, in other work we have given a modular labelled tableaux proof theory that systematically captures validity in all classes of memory model defined by combinations of the \emph{separation properties} given in the literature~\cite{Docherty2018d}. Although duality theorems for axiomatic extensions of hybrid logic with one unary modality have been given~\cite{CR2017}, the generalization required to achieve such a result for hybrid bunched logics is beyond the scope of this paper. We defer such an investigation to another occasion.

\section{Modal and Multiplicative Extensions}\label{sec:extension}

To conclude, we adapt the results of Section~\ref{sec:bunchedimplications} to the range of logics extending \logicfont{BI} and \logicfont{BBI}.  Applications of these logics include reasoning about deny-guarantee permissions, concurrency, timed petri nets, and --- via the interpretation of heap intersection operations --- complex resource sharing.
In particular we detail: the family of \emph{separating} modal logics that extend \logicfont{BBI} with resource-offset modalities; the \emph{De Morgan} bunched logics that extend \logicfont{(B)BI} with a De Morgan negation, facilitating the definition of multiplicatives corresponding to disjunction and falsum~\cite{CBI}; the family of \emph{sub-classical} bunched logics that sit intermediate between \logicfont{(B)BI} and the classical bunched logics;
and finally \logicfont{CKBI}, a new logic suggested by algebraic interpretations of (a basic version of) Concurrent Separation Logic~\cite{exchange}. In doing so we give an exhaustive semantic analysis of modalities and multiplicative connectives in the bunched logic setting, while indicating how the framework can potentially be applied to Concurrent Separation Logic.


\subsection{Separating Modal Logics}
First we consider separating modal logics. These logics extend BBI with resource modalities $\Diamond_r$  and include Courtault et al.'s~\cite{LSM} logic of separating modalities and Galmiche et al.'s~\cite{epistemic} epistemic resource logic. In that work, the logics are introduced semantically and given a tableaux proof theory with countermodel extraction. In models of these logics, for each $\Diamond_r$ in the signature, $r$ is assigned to a resource $\realization{r}$. $\Diamond_r \varphi$ is then interpreted as stating that there exists a resource $x$ that can be composed with the \emph{local resource} $\realization{r}$ to access a state satisfying $\varphi$. We generalise this to a schema for defining separating modal logics.

Let $\mathrm{Prop}$ be a set of atomic propositions, ranged over by $\mathrm{p}$. The set of all formulae of separating modal logic $\mathrm{Form}_{\mathrm{SML}}$ is generated by the grammar
\[ \varphi ::= \mathrm{p} \mid \top \mid \bot \mid \mtop \mid \varphi \land \varphi \mid \varphi \lor \varphi \mid \varphi \rightarrow \varphi \mid \varphi * \varphi \mid \varphi \wand \varphi \mid \Diamond \varphi \] where additive negation is defined by $\neg \varphi := \varphi \rightarrow \bot$ and the necessity modality is defined by $\Box \varphi := \neg \Diamond \neg \varphi$. For each formula $\varphi$ of \logicfont{SML}, a \emph{separating modality} $\Diamond_{\varphi}$ is defined by $\Diamond_{\varphi} \psi := \neg(\varphi \wand \neg \Diamond \psi)$. The additional rules to be added to the Hilbert system are those governing that $\Diamond$ is a normal modality, together with any axioms that may define the character of the modality: for example, in the logic of separating modalities $\Diamond$ is an S4 modality, thus satisfying the axioms $\varphi \vdash \Diamond \varphi$ and $\Diamond \Diamond \varphi \vdash \Diamond \varphi$, whereas in epistemic resource logic $\Diamond$ is an S5 modality, thus satisfying the S4 axioms plus $\Diamond \Box \varphi \vdash \Box \varphi$.

\logicfont{SML} is interpeted on structures that extend BBI frames with an accessibility relation that we call SML frames. Examples of these frames given in the literature include models of the producer-consumer problem, timed petri nets and a range of security scenarios.

\begin{defi}[SML Frame]
An \emph{SML frame} is a structure $\mathcal{X} = (X, \circ, E, R)$ such that $(X, \circ, E)$ is a BBI frame and $R$ a binary relation on $X$. \qed%
\end{defi}

If $\Diamond$ is axiomatised by modal axioms with frame correspondents, the SML frame must also satisfy those frame correspondents. For example, for an S4 modality, $R$ must be reflexive and transitive; for an S5 modality, $R$ must additionally be symmetric. A SML frame together with a valuation $\mathcal{V}$ gives a SML model $\mathcal{M}$, and for such a model the satisfaction relation $\vDash_{\mathcal{V}} \subseteq X \times \mathrm{Form}_{\mathrm{SML}}$ is inductively generated by the clauses for BBI, together with the clause for $\Diamond$ given in Figure~\ref{fig:sat-sml}. This figure also includes the satisfaction clause for $\Diamond_{\varphi}$, obtained directly from the definition $\Diamond_{\varphi} \psi := \neg (\varphi \wand \neg \Diamond \psi)$. Intuitively, this clause states that $\Diamond_\varphi \psi$ is true at a resource $x$ iff $x$ can be composed with a resource satisfying $\varphi$, with that composition having access to a state at which $\psi$ is true. If $\mathrm{Prop}$ contains atoms $r$ that are assigned to a single state $\realization{r} \in X$ by $\mathcal{V}$, the clause for $\Diamond_r$ is precisely that given in the primitive satisfaction clauses for the logic of separating modalities and epistemic resource logic.

\begin{figure}
\centering
\hrule
\vspace{1mm}
 \setlength\tabcolsep{3pt}
\setlength\extrarowheight{4pt}
\begin{tabular}{c c c c l r c}
$x$ & $\vDash_{\Valuation}$ & $\Diamond \varphi$ & iff & there exists $y$ such that $Rxy$ and $y \vDash_{\Valuation} \varphi$ \\
$x$ & $\vDash_{\Valuation}$ & $\Diamond_\varphi \psi$ & iff & there exists $w, y, z$ such that $z \in x \circ y$, $y \vDash_{\Valuation} \varphi$, $Rzw$ and $w \vDash_{\Valuation} \psi$.
\end{tabular}
\caption{Satisfaction for \logicfont{SML}}
\vspace{1mm}
\hrule%
\label{fig:sat-sml}
\end{figure}

The separating modalities $\Diamond_\varphi$ inherit the property of being normal from $\Diamond$, and are thus well behaved. First note that $x \vDash \Diamond_\varphi \bot$ never holds, as there is no state at which $w \vDash \bot$. It is thus equivalent to $\bot$. For the distribution of $\Diamond_\varphi$ over $\lor$, we note that (\emph{cf.} Proposition~\ref{prop-alg-prop}), $\varphi \wand (\psi_1 \land \psi_2)$ is equivalent to $(\varphi \wand \psi_1) \land (\varphi \wand \psi_2)$. By applying De Morgan laws and the distribution of $\Diamond$ over $\lor$, it is easily seen that $\Diamond_\varphi(\psi_1 \lor \psi_2) := \neg (\varphi \wand \neg \Diamond (\psi_1 \lor \psi_2))$ is logically equivalent to $\neg (\varphi \wand \neg \Diamond \psi_1) \lor \neg (\varphi \wand \neg \Diamond \psi_2)$, or, $\Diamond_\varphi \psi_1 \lor \Diamond_\varphi \psi_2$. They do not, however, necessarily inherit any additional axioms from $\Diamond$: for example, $\Diamond$ being an S4 modality does not entail that $\Diamond_\varphi$ is an S4 modality.

It follows that extending the duality theorems to the family of separating modal logics can be achieved by extending BBI duality with the structure governing normal operators and modal correspondence in the standard modal logic duality~\cite{goldblatt,correspondence}. This also applies to non-separating modal bunched logics which simply add a diamond modality to \logicfont{(B)BI}~\cite{modalbi,modaldmbi} (using \emph{intuitionistic} modal logic duality~\cite{Wolter1998} where appropriate).

\subsection{De Morgan Bunched Logics}\label{subsec:cbi}

Brotherston \& Calcagno~\cite{CBI} introduce the logic Classical BI (\logicfont{CBI}) which extends $\logicfont{BBI}$ with a De Morgan negation $\mneg $. By substituting $*$ and $\mtop$ into the De Morgan laws relating $\land$ to $\lor$ and $\top$ to $\bot$, this yields multiplicative connectives corresponding to disjunction and falsum. Although it is shown that the Separation Logic heap model is not a model of \logicfont{CBI}, a number of interesting applications are suggested ranging from deny-guarantee permissions to regular languages. Brotherston~\cite{displayed} also gives a display calculus for De Morgan BI (\logicfont{DMBI}\footnote{This should not be confused with the modal bunched logic \logicfont{DMBI} of Courtault \& Galmiche~\cite{modaldmbi}.}), the evident intuitionistic variant of \logicfont{CBI} that instead extends \logicfont{BI}, but no Kripke semantics or completeness proof can be found in the literature.


Let $\mathrm{Prop}$ be a set of atomic propositions, ranged over by $\mathrm{p}$. The set of all formulae of the De Morgan bunched logics $\mathrm{Form}_{\mathrm{DMBI}}$ is generated by the grammar \[ \varphi ::= \mathrm{p} \mid \top \mid \bot \mid \mtop \mid \varphi \land \varphi \mid \varphi \lor \varphi \mid \varphi \rightarrow \varphi \mid \mneg \varphi \mid \varphi * \varphi \mid \varphi \wand \varphi, \] where additive negation is given by $\neg \varphi := \varphi \rightarrow \bot$, multiplicative falsum is given by $\mbot := \mneg \mtop$ and multiplicative disjunction is given by $\varphi \mor \psi := \mneg (\mneg \varphi * \mneg \psi)$.

Figure~\ref{fig:hilbert-dmbi} gives Hilbert rules that need to be added to the proof systems of the bunched implication logics to obtain systems for the De Morgan bunched logics: to get \logicfont{DMBI}, the Hilbert system for \logicfont{BI} is extended with 19.\ and 20.; to get \logicfont{CBI}, 19.\ and 20.\ are instead added to the Hilbert system for \logicfont{BBI}. These logics can be interpreted on the following frame structures by extending the semantic clauses for \logicfont{(B)BI} with that from Fig.~\ref{fig:sat-dmbi}. In the case of \logicfont{DMBI} this interpretation is persistent. 

\begin{figure}
\hrule
    \vspace{1mm}
\setlength\tabcolsep{7pt}
\setlength\extrarowheight{15pt}
\centering
\begin{tabular}{lclc}
19. & $\cfrac{}{\mneg \mneg \varphi \dashv\vdash \varphi}$ &
20. & $\cfrac{}{\mneg \varphi \dashv\vdash \varphi \wand \mbot}$
\end{tabular}

\caption{Hilbert rules for De Morgan bunched logics.}
\vspace{1mm}
\hrule%
\label{fig:hilbert-dmbi}
\end{figure}

\begin{defi}[DMBI/CBI Frame]\label{def:dmbiframe}
A DMBI frame is a tuple $\mathcal{X} = (X, \succcurlyeq, \circ, E, -)$ where $(X, \succcurlyeq, \circ, E)$ is a BI frame and $-:X \rightarrow X$ is an operation satisfying the following conditions (with outermost universal quantification omitted for readability):
\[
\begin{array}{rlrr}
\text{(Dual)} & x \succcurlyeq y \rightarrow -y \succcurlyeq -x & \text{(Involutive)} & --x = x       \\
\text{(Compatability)} & z \in x \circ y \rightarrow -x \in -z \circ y.
\end{array}
\]
A CBI frame is a DMBI frame for which the order $\succcurlyeq$ is equality $=$. \qed%
\end{defi}

The definition of CBI frame here looks different to the notion given by Brotherston \& Calcagno~\cite{CBI} but is equivalent. There, a (multi-unit) CBI model is a tuple $(X, \circ, E, -, \infty)$ such that $(X, \circ, E)$ is a BBI frame, with $-: X \rightarrow X$ and $\infty \subseteq X$ satisfying, for all $x \in X$, $-x$ is the unique element such that $\infty \cap (-x \circ x) \neq \emptyset$. (Involutive) and (Compatability) are then proved as consequences of this definition in their Proposition 2.3 (1) and (3). As they discuss, the choice of $\infty$ is fixed by the choice of $-$, and it can easily be seen that defining $\infty = \{-e \mid e  \in E\}$ on our CBI frames yields their CBI models. We choose our presentation as it simplifies proofs.

\begin{figure}
\centering
\hrule
\vspace{1mm}
\setlength\tabcolsep{3pt}
\setlength\extrarowheight{2pt}
\begin{tabular}{c c c c l r c c c c r r c}
$x$ & $\vDash_{\Valuation}$ & $\mneg \varphi$ & iff & \myalign{l}{$-x \not\vDash_{\Valuation} \varphi$} & &
\end{tabular}
\caption{Satisfaction for De Morgan bunched logics.}
\vspace{1mm}
\hrule%
\label{fig:sat-dmbi}
\end{figure}

\begin{defi}[DMBI/CBI Algebra]\hfill
\begin{enumerate}
\item A \emph{DMBI algebra} is an algebra $\mathbb{A} = (A, \land, \lor, \rightarrow, \top, \bot, *, \wand, \mtop, \mbot)$
such that $(A, \land, \lor, \rightarrow, \top, \bot, *, \wand, \mtop)$ is a BI algebra and, defining $\mneg a := a \wand \mbot$, $\mneg \mneg a = a$ and $\mneg \mtop = \mbot$.
\item A \emph{CBI algebra} is a DMBI algebra $\mathbb{A} = (A, \land, \lor, \rightarrow, \top, \bot, *, \wand, \mtop, \mbot)$ in which $(A, \land, \lor, \rightarrow, \top, \bot, *, \wand, \mtop)$ is a BBI algebra. \qed%
\end{enumerate}
\end{defi}

We collect a number of useful properties of these algebras in the following proposition.

\begin{prop}\label{prop:mnegprop}
Let $\mathbb{A}$ be a DMBI or CBI algebra with $a, b, c \in \mathbb{A}$ and $X \subseteq X$. Then the following hold.
\begin{enumerate}
\item If $\bigvee X$ exists, then $\bigwedge \mneg X$ exists and $\mneg \bigvee X = \bigwedge \mneg X$;
\item If $a \leq b$ then $\mneg b \leq \mneg a$;
\item If $\bigwedge X$ and $\bigvee \mneg X$ exist then $\mneg \bigwedge X = \bigvee \mneg X$;
\item $a * b \leq c$ iff $b * \mneg c \leq \mneg a$. \qed%
\end{enumerate}
\end{prop}

\noindent
As a result of this proposition, we have that for any DMBI or CBI algebra $\mathbb{A}$, the fragment $(A, \land, \lor, \mneg, \top, \bot)$ is a De Morgan algebra~\cite{Moisil1935}. Thus $\mneg$ is a dual automorphism on the underlying bounded distributive lattice of $\mathbb{A}$. Now an algebraic interpretation of DMBI or CBI on a DMBI or CBI algebra extends one on the underlying BI or BBI algebra by additionally setting $\Interp{\mneg a} = \mneg \Interp{a}$. That this is sound and complete follows straightforwardly from the additional De Morgan bunched logic Hilbert rules matching the defining properties of $\mneg$.

\begin{defi}[DMBI/CBI Morphism]
A DMBI (CBI) morphism $g: \mathcal{X} \rightarrow \mathcal{X}'$ is a BI (BBI) morphism satisfying the additional property (8) $g(-x) = -g(x)$. This forms a category $\mathrm{DMBIFr}$ ($\mathrm{CBIFr}$). \qed%
\end{defi}

We now extend (B)BI duality systematically to obtain it for DMBI and CBI\@.

\begin{defi}[Complex DMBI/CBI Algebra]
Given a DMBI (CBI) frame $\mathcal{X}$, the complex algebra of $\mathcal{X}$, $Com^\vtiny{DMBI}(\mathcal{X})$ ($Com^\vtiny{CBI}(\mathcal{X})$) is given by extending $Com^\vtiny{BI}(\mathcal{X})$ ($Com^{\vtiny{BBI}}(\mathcal{X})$) with the set $U = \{ x \in \mathcal{X} \mid -x \not\in E \}$. \qed%
\end{defi}

\begin{lem}\label{lem:dmbicomplexalgebra}
Given a DMBI (CBI) frame $\mathcal{X}$, $Com^\vtiny{DMBI}(\mathcal{X})$ ($Com^\vtiny{CBI}(\mathcal{X})$) is a DMBI (CBI) algebra.
\end{lem}

\begin{proof}
First we note that $U$ is an upwards-closed set. Suppose $u \in U$ and $u' \succcurlyeq u$. Since $-u \not\in E$, and $E$ is upwards-closed, we must have that $-u' \not\in E$ as $-u \succcurlyeq -u'$.
On the complex algebra we define the multiplicative negation by $\sim_{\mathcal{X}} A := A \Rres_{\mathcal{X}} U$, as guided by the definition of DMBI/CBI algebra. We must show that $\sim_{\mathcal{X}} \sim_{\mathcal{X}} A = A$ and $\sim_{\mathcal{X}} E = U$, and this follows immediately if $\sim_{\mathcal{X}} A = \{ a \mid -a \not\in A\}$; we verify this identity.

First assume $a \in \sim_{\mathcal{X}} A$. Let $e \in E$ be such that $a \in a \circ e$ by the frame axiom Unit Existence. Then by Compatibility, $-e \in a \circ -a$ and if $- a \in A$, we would have $- e \in U$, a contradiction as $--e = e \in E$. Now assume $a$ is such that $-a \not\in A$. Let $a' \succcurlyeq a$ with $b \in A$ and $c \in a' \circ b$. We assume for contradiction that $c \not\in U$. Then $-c \in E$ and by Compatibility we have $-a' \in b \circ -c$. By the frame axiom Coherence $-a' \succcurlyeq b$, and by upwards-closure of $A$, $-a \in A$; a contradiction. Hence $c \in U$.
\end{proof}

%

\begin{defi}[Prime Filter DMBI/CBI Frame]
Given a DMBI (CBI) algebra $\mathbb{A}$, the prime filter frame of $\mathbb{A}$, $Pr^\vtiny{DMBI}(\mathbb{A})$ $(Pr^{\vtiny{CBI}}(\mathbb{A}))$ is given by extending $Pr^\vtiny{BI}(\mathbb{A})$ $(Pr^{\vtiny{BBI}}(\mathbb{A}))$ with the operation $-_{\mathbb{A}}F := \overline{\mneg F}$. \qed
\end{defi}

That this is well defined follows from the fact $\mneg $ is a dual automorphism on the underlying bounded distributive lattice: this entails that, given a prime filter $F$, $\mneg F$ is a prime ideal, and thus $\overline{\mneg F}$ is a prime filter.

\begin{lem}\label{lem:dmbiprimefilter}
Given a DMBI (CBI) algebra $\mathbb{A}$, the prime filter frame $Pr^\vtiny{DMBI}(\mathbb{A})$ ($Pr^\vtiny{CBI}(\mathbb{A})$) is a DMBI (CBI) frame. 
\end{lem}

\begin{proof}
We must check that the three DMBI frame axioms hold. If $F' \supseteq F$ then $\mneg F' \supseteq \mneg F$ and so $\overline{\mneg F} \supseteq \overline{\mneg F'}$, as required for Dual. $\mneg \mneg a = a$ straightforwardly entails that Involutive is satisfied. Finally we verify Compatibility. Assume $F_z \in F_x \circ_{\mathbb{A}} F_y$ and let $c \in -_{\mathbb{A}} F_z$ and $d \in F_y$. For contradiction, suppose $c * d \not\in -_{\mathbb{A}} F_x$. Then there necessarily exists $a \in F_x$ such that $c * d \leq \mneg a$. By Proposition~\ref{prop:mnegprop} this entails $a * d \leq \mneg c$. Since $a \in F_x$ and $d \in F_y$ we have $a * d \in F_z$, and thus $\mneg c \in F_z$. However, $c \in -_{\mathbb{A}} F_z$ entails $c \not\in \mneg F_z$, so $\mneg c \not\in F_z$, a contradiction. Thus $c * d \in -_{\mathbb{A}} F_x$ as required.
\end{proof}

The representation theorem now follows from the easy verification that $\theta_{\mathbb{A}}(\mbot) = \{ F \mid \mtop \not\in -_{\mathbb{A}} F \}$, from which soundness and completeness of the frame semantics is an immediate corollary. In the case of \logicfont{DMBI}, this is the first existing completeness result.

\begin{thm}[Representation Theorem for DMBI/CBI Algebras]\label{thm:dmbirep}
Every DMBI (CBI) algebra is isomorphic to a subalgebra of a complex algebra. Specifically, given a DMBI (CBI) algebra $\mathbb{A}$, the map $\theta_{\mathbb{A}}: \mathbb{A} \rightarrow Com^\vtiny{DMBI}(Pr^\vtiny{DMBI}(\mathbb{A}))$ ($\theta_{\mathbb{A}}: \mathbb{A} \rightarrow Com^\vtiny{CBI}(Pr^\vtiny{CBI}(\mathbb{A}))$) defined $\theta_{\mathbb{A}}(a) = \{ F \in Pr(\mathbb{A}) \mid a \in F \}$ is an embedding.  \qed
\end{thm}

\begin{cor}[Relational Soundness and Completeness]
For all formulas $\varphi$, $\psi$ of \logicfont{DMBI} (\logicfont{CBI}), $\varphi \vdash \psi$ is provable in $\mathrm{DMBI}_{\mathrm{H}}$ ($\mathrm{CBI}_{\mathrm{H}}$) iff $\varphi \vDash \psi$ in the relational semantics.  \qed%
\end{cor}

These assignments are once again made functorial by sending morphisms to their inverse image. To obtain a dual equivalence of categories we introduce topology.

\begin{defi}[DMBI/CBI Space]
A \emph{DMBI space} is a structure $\mathcal{X} = (X, \mathcal{O}, \succcurlyeq, \circ, E, -)$ such that
\begin{enumerate}
\item $(X, \mathcal{O}, \succcurlyeq, \circ, E)$ is a BI space,
\item $(X, \succcurlyeq, \circ, E, -)$ is a DMBI frame, and
\item $-$ is a continuous map.
\end{enumerate}
A \emph{CBI} space is a DMBI space for which $\succcurlyeq$ is equality. \qed%
\end{defi}

\noindent
As usual, morphisms for these spaces are given by continuous DMBI (CBI) morphisms. Now we have the functor $Pr^\vtiny{DMBI}: \mathrm{DMBIAlg} \rightarrow \mathrm{DMBISp}$ defined by $Pr^\vtiny{DMBI}(\mathbb{A}) = (Pr(\mathbb{A}), \mathcal{O}_{\mathbb{A}}, \supseteq, \circ_{\mathbb{A}}, E_{\mathbb{A}}, -_{\mathbb{A}})$ and $Pr^\vtiny{DMBI}(f) = f^{-1}$. 
Continuity of $-_{\mathbb{A}}$ can be verified on the subbase elements of $\mathcal{O}_{\mathcal{A}}$, and this holds because ${(-_{\mathbb{A}})}^{-1}[\theta_{\mathbb{A}}(a)] = \overline{\theta_{\mathbb{A}}(\mneg a)}$ and ${(-_{\mathbb{A}})}^{-1}[\overline{\theta_{\mathbb{A}}(a)}] = \theta_{\mathbb{A}}(\mneg a)$. In the other direction, we have the functor $Clop^\vtiny{DMBI}_{\succcurlyeq}: \mathrm{DMBISp} \rightarrow \mathrm{DMBIAlg}$ defined $Clop^\vtiny{DMBI}(\mathcal{X}) = (\mathcal{CL}_{\succcurlyeq}(\mathcal{X}), \cap, \cup, \Rightarrow_{\mathcal{X}}, X, \emptyset, \bullet_{\mathcal{X}}, \Rres_{\mathcal{X}}, E, U)$ and $Clop^{\vtiny{DMBI}}_{\succcurlyeq}(g) = g^{-1}$. That $U \in \mathcal{CL}_{\succcurlyeq}(\mathcal{X})$ follows from the fact that $U = \overline{-E}$. $E$ is clopen, so $-E$ is clopen by continuity and so too is $\overline{-E}$. Further, $E$ is upwards-closed, so $-E$ is downwards-closed, meaning $\overline{-E}$ is upwards-closed. The analogous definitions then give the required structure for \logicfont{CBI}. We once again consider the collection of maps $\eta_{\mathcal{X}}(x) = \{ C \in \mathcal{CL}_{\succcurlyeq}(\mathcal{X}) \mid x \in C \}$ to complete the duality.

\begin{thm}[Duality Theorem for \logicfont{DMBI}/\logicfont{CBI}]
$\theta$ and $\eta$ form a dual equivalence of categories between $\mathrm{DMBIAlg}$ ($\mathrm{CBIAlg}$ and $\mathrm{DMBISp}$ ($\mathrm{CBISp}$).
\end{thm}

\begin{proof}
The last remaining steps are to show that the components $\eta_{\mathcal{X}}$ are isomorphisms in $\mathrm{DMBISp}$. The key step is to verify that $-_{Clop^\vtiny{DMBI}_{\succcurlyeq}(\mathcal{X})} \eta_{\mathcal{X}}(x) = \eta_{\mathcal{X}}(-x)$, as the rest obtains from BI duality. Unpacking the definition, we must check $\overline{\{C' \Rres_{\mathcal{X}} U \mid C' \in \eta_{\mathcal{X}}(x) \}} = \eta_{\mathcal{X}}(-x)$. For the right-to-left inclusion, suppose $-x \in C$ and for contradiction $C = C' \Rres_{\mathcal{X}} U$ for some upwards-closed clopen $C'$ such that $x \in C'$. Then by Unit Existence there exists $e \in E$ such that $-x \in -x \circ e$, and by Compatibility $-e \in -x \circ x$. By assumption this entails $-e \in U$, but $--e = e \in E$, a contradiction. Hence $C \in \overline{\{C' \Rres_{\mathcal{X}} U \mid C' \in \eta_{\mathcal{X}}(x) \}}$.

For the left-to-right inclusion, note that \[\eta_{\mathcal{X}}(x) = \{-C \mid C \emph{ downwards-closed clopen and } x \in C\}\] holds; that this is the case is a consequence of $-$ being continuous and the frame axiom Dual. Now suppose we have $-C$ such that $C$ downwards-closed and clopen and $x \not\in C$. Then $x \in \overline{C}$ and we claim that $-C = \overline{C} \Rres_{\mathcal{X}} U$. First assume $-y \in -C$. Suppose $y' \succcurlyeq -y$ and $z \in \overline{C}$ such that $w \in y' \circ z$ and assume for contradiction that $w \not\in U$. Then $-w \in E$. By Compatibility, $-y' \in z \circ -w$, and by Coherence $-y' \succcurlyeq z$. By Dual and our assumption, $-z \succcurlyeq y' \succcurlyeq -y$, and by Dual again $y \succcurlyeq z$. Thus by upwards-closure of $\overline{C}$ we have $y \in\overline{C}$, but $y \in C$ by assumption; a contradiction. Hence $w \in U$ and $-y \in \overline{C} \Rres_{\mathcal{X}} U$. Now suppose $-y \not\in -C$. Then $y \in \overline{C}$. By Unit Existence there is $e \in E$ such that $y \in y \circ e$, and by Compatibility $-e \in -y \circ y$. We have $y \in \overline{C}$ with $-e \not\in U$, so $-y \not\in \overline{C} \Rres_{\mathcal{X}} U$. \logicfont{CBI} duality is the particular case given when $\succcurlyeq$ is equality.
\end{proof}


%

\subsection{Sub-Classical Bunched Logics}\label{subsec:bibi}

Brotherston \& Villard~\cite{subclassical} introduce a family of logics extending \logicfont{BBI} that they call \emph{sub-classical bunched logics}, as they lie intermediate between \logicfont{BBI} and \logicfont{CBI}. As heaps equipped with intersection operations are models of the logics, they are of clear interest to the Separation Logic community, with verification of algorithms involving complex resource sharing suggested as a natural application. Basic Bi-intuitionistic Boolean Bunched logic (\logicfont{BiBBI}) is defined to be the multiplicative extension of \logicfont{BBI} that drops all De Morgan laws between multiplicative conjunction, disjunction, truth and falsum. Thus in \logicfont{BiBBI} these connectives can no longer be defined in terms of each other (as they were for \logicfont{DMBI} and \logicfont{CBI}) and must be given as primitives. A number of these correspondences can then be added as axioms without the logic collapsing into \logicfont{CBI}. We show that our framework captures \logicfont{BiBBI} and its axiomatic extensions, as well as the evident intuitionistic variant $\logicfont{BiBI}$ and the intermediate logics weaker than $\logicfont{DMBI}$.

Let $\mathrm{Prop}$ be a set of atomic propositions, ranged over by $\mathrm{p}$. The set of all formulae of the subclassical bunched logics $\mathrm{Form}_{BiBI}$ is generated by the grammar
\[ \varphi ::= \mathrm{p} \mid \top \mid \bot \mid \mtop \mid \mbot \mid \varphi \land \varphi \mid \varphi \lor \varphi \mid \varphi \rightarrow \varphi \mid \varphi * \varphi \mid \varphi \mor \varphi \mid \varphi \wand \varphi \mid \varphi \rslash \varphi, \] where additive negation is defined by $\neg \varphi := \varphi \rightarrow \bot$.

The simplest subclassical bunched logics are called \emph{basic bi-intuitionistic (B)BI}, or basic \logicfont{Bi(B)BI}. Fig.~\ref{fig:hilbert-bibi} gives Hilbert rules for basic \logicfont{Bi(B)BI} to be added to the system for \logicfont{(B)BI}. In this basic case, very little is enforced for the new connectives. 
A number of the De Morgan correspondences between these connectives can be added back as axioms \emph{without} collapsing the logic to a De Morgan bunched logic, and in Figure~\ref{fig:hilbert-subclassical} these axioms are given as Hilbert-style rules which can be added to the system for basic \logicfont{Bi(B)BI}. The logic can be interpreted on the following frame structures by extending the semantic clauses for \logicfont{(B)BI} with those from Fig.~\ref{fig:sat-bibi}. In the case of \logicfont{BiBI} this interpretation is persistent. Each of the optional subclassical axioms can be witnessed by a corresponding frame property, given in Fig.~\ref{fig:algebrasubclassical}.

\begin{defi}[Basic Bi(B)BI Frame]
A \emph{basic Bi(B)BI frame} is a structure $\mathcal{X} = (X, \succcurlyeq, \circ, E, \triangledown, U)$ such that $(X, \succcurlyeq, \circ, E)$ is a (B)BI frame, $\triangledown: X^2 \rightarrow \mathcal{P}(X)$ and $U \subseteq X$, satisfying (with outermost universal quantification omitted for readability):
\begin{equation*}
\text{(Commutativity)} \quad z \in x \mathbin{\triangledown} y \rightarrow z \in y \mathbin{\triangledown} x;
\quad\quad
\text{(U-Closure)} \quad u \in U \land u \succcurlyeq u' \rightarrow u' \in U.
\tag*{\qed}
\end{equation*}
\end{defi}

\begin{figure}
\hrule
    \vspace{1mm}
\setlength\tabcolsep{7pt}
\setlength\extrarowheight{15pt}
\centering
\begin{tabular}{lclclclc}
21. & $\cfrac{\eta \vdash \varphi \mor \psi}{\eta \rslash \varphi \vdash \psi}$ &
22. & $\cfrac{\eta \rslash \varphi \vdash \psi}{\eta \vdash \varphi \mor \psi}$ &
23. & $\cfrac{\xi \vdash \varphi \quad \eta \vdash \psi}{\xi \mor \eta \vdash \varphi \mor \psi}$ &
24. & $\cfrac{}{\varphi \mor \psi \vdash \psi \mor \varphi}$
\end{tabular}

\caption{Hilbert rules for basic \logicfont{Bi(B)BI}.}
\vspace{1mm}
\hrule%
\label{fig:hilbert-bibi}
\end{figure}

\begin{figure}
\hrule
    \vspace{1mm}
\setlength\tabcolsep{7pt}
\centering
\begin{tabular}{lclc}
Associativity & $\cfrac{}{\varphi \mor (\psi \mor \chi) \vdash (\varphi \mor \psi) \mor \chi}$ &
$\mbot$ Weakening & $\cfrac{}{\varphi \vdash \varphi \mor \mbot}$ \\
$\mbot$ Contraction & $\cfrac{}{\varphi \mor \mbot \vdash \varphi}$ &
$\mor$ Contraction & $\cfrac{}{\varphi \mor \varphi \vdash \varphi}$ \\
Weak Distributivity & $\cfrac{}{\varphi * (\psi \mor \chi) \vdash (\varphi * \psi) \mor \chi}$
\end{tabular}

\caption{Hilbert rules for subclassical bunched logics.}
\vspace{1mm}
\hrule%
\label{fig:hilbert-subclassical}
\end{figure}

\begin{figure}
\centering
\hrule
\vspace{1mm}
 \setlength\tabcolsep{3pt}
\setlength\extrarowheight{4pt}
\begin{tabular}{c c c c l r c}
$x$ & $\vDash_{\Valuation}$ & $\mbot$ & iff & $x \not\in U$ \\
$x$ & $\vDash_{\Valuation}$ & $\varphi \mor \psi$ & iff & for all $s, t, u$, $x \preccurlyeq s \in t \mathbin{\triangledown} u$ implies $t \vDash_{\Valuation} \varphi$ or  $u \vDash_{\Valuation} \psi$ \\
$x$ & $\vDash_{\Valuation}$ & $\varphi \rslash \psi$ & iff & there exist $s, t, u$ such that $x \succcurlyeq s$, $u \in t \mathbin{\triangledown} s$, $u \vDash_{\Valuation} \varphi$ and $t \not\vDash_{\Valuation} \psi$ \\
\end{tabular}
\caption{Satisfaction for \logicfont{Bi(B)BI}. \logicfont{BiBBI} is the case where $\succcurlyeq$ is =.}
\vspace{1mm}
\hrule%
\label{fig:sat-bibi}
\end{figure}

\begin{figure}
\centering
\hrule
\vspace{1mm}
 \setlength\tabcolsep{2pt}
\setlength\extrarowheight{4pt}
\begin{tabular}{lll}
 Property & Axiom & Frame Correspondent \\\toprule
Associativity & $a \mor (b \mor c) \leq (a \mor b) \mor c$ & $t' \preccurlyeq t \in x \mathbin{\triangledown} y \land w \in t' \mathbin{\triangledown} z \rightarrow \exists s, s', w'$ \\
& & $(s' \preccurlyeq s  \in y \mathbin{\triangledown} z \land w \preccurlyeq w' \in x \mathbin{\triangledown} s')$  \\
$\mbot$ Weakening & $a \leq a \mor \mbot$ & $u \in U \land x \in y \mathbin{\triangledown} u \rightarrow x \preccurlyeq y$ \\
$\mbot$ Contraction & $a \mor \mbot \leq a$ & $\exists u \in U(w \in w \mathbin{\triangledown} u)$ \\
$\mor$ Contraction & $a \mor a \leq a$ & $x \in x \mathbin{\triangledown} x$\\
Weak Distributivity & $a * (b \mor c) \leq (a * b) \mor c$ & $t' \succcurlyeq t \in x_1 \circ x_2 \land t' \preccurlyeq t'' \in y_1 \mathbin{\triangledown} y_2 \rightarrow$\\
&& $\exists w( y_1 \in x_1 \circ w \land x_2 \in w \mathbin{\triangledown} y_2)$
\end{tabular}
\caption{Bi(B)BI correspondence (cf.~\cite{subclassical}). The BiBBI variants replace $\succcurlyeq$ with =. }
\vspace{1mm}
\hrule%
\label{fig:algebrasubclassical}
\end{figure}


\begin{defi}[Basic Bi(B)BI Algebra]
A \emph{basic Bi(B)BI algebra} is an algebra $\mathbb{A} = (A, \land, \lor, \rightarrow, \top, \bot, *, \mor, \wand, \rslash, \mtop, \mbot)$ such that $(A, \land, \lor, \rightarrow, \top, \bot, *, \wand, \mtop)$ is a (B)BI algebra, $\mor$ a commutative binary operation, $\rslash$ a binary operation, $\mbot$ a constant, such that, for all $a, b, c \in \mathbb{A}$, $a \leq b \mor c \text{ iff } a \rslash b \leq c$. \qed%
\end{defi}

The residuation property of $\mor$ and $\rslash$ ensures $\mor$ is monotone, as well as a number of useful properties dual to those of Proposition~\ref{prop-alg-prop}.

\begin{prop}\label{prop:subclassicalprop}
Let $\mathbb{A}$ be a basic Bi(B)BI algebra. Then, for all $a, b, a', b' \in A$ and $X, Y \subseteq A$, we have the following:
\begin{enumerate}
\item If $a \leq a'$ and $b \leq b'$ then $a \mor b \leq a' \mor b'$;
\item If $\bigwedge X$ and $\bigwedge Y$ exist then $\bigwedge_{x \in X, y \in Y} x \mor y$ exists and
	$(\bigwedge X) \mor (\bigwedge Y) =  \bigwedge_{x \in X, y \in Y} x \mor y$;
\item If $a = \top$ or $b = \top$ then $a \mor b = \top$;
\item If $\bigwedge X$ exists then for any $z \in A$: $\bigvee_{x \in X} (x \rslash z)$ exists with $\bigvee_{x \in X} (x \rslash z) = (\bigwedge X) \rslash z$;
\item If $\bigvee X$ exists then for any $z \in A$: $\bigvee_{x \in X} (z \rslash x)$ exists with $\bigvee_{x \in X} (z \rslash x) = z \rslash (\bigvee X)$; and
\item $a \rslash \top =  \bot \rslash a = \bot$. \qed%
\end{enumerate}
\end{prop}

\noindent
Fig.~\ref{fig:algebrasubclassical} gives algebraic axioms directly corresponding to the Hilbert axioms of subclassical bunched logics. For any collection of subclassical axioms $\Sigma$, we denote by $\mathrm{Bi(B)BIAlg}_{\Sigma}$ the category of Bi(B)BI algebras satisfying $\Sigma$. By an argument analogous to those that preceeded, sound and complete algebraic interpretations can be defined on these structures. We also denote by $\mathrm{Bi(B)BIFr}_{\Sigma}$ the category of Bi(B)BI frames satisfying the frame correspondents of $\Sigma$, where Bi(B)BI frame morphisms are given by the following definition.

\begin{defi}[Bi(B)BI Morphism]
A \emph{Bi(B)BI morphism} is a map $f: \mathcal{X} \rightarrow \mathcal{X}'$ such that $f$ is a (B)BI morphism satisfying the following additional properties:
\begin{enumerate}[align=left]
\item[(8)] $x \in y \mathbin{\triangledown} z$ implies $g(x) \in g(y) \mathbin{\triangledown} g(z)$;
\item[(9)] $g(x) \preccurlyeq' s' \in t' \mathbin{\triangledown'} u'$ implies there exists $s,t, u$ such that $x \preccurlyeq s \in t \mathbin{\triangledown} u$, $t' \succcurlyeq' g(t)$ and $u' \succcurlyeq g(u)$;
\item[(10)] $g(x) \succcurlyeq' s', u' \in t' \mathbin{\triangledown}' s'$ implies there exists $s, t, u$ such that $x \succcurlyeq s, u \in t \mathbin{\triangledown} s$, $g(u) \succcurlyeq' u'$ and $t' \succcurlyeq' g(t)$. \qed%
\end{enumerate}
\end{defi}


\noindent
We now set up the basic duality theory for these structures.

\begin{defi}[Bi(B)BI Complex Algebra]
Given a Bi(B)BI frame $\mathcal{X}$ the complex algebra of $\mathcal{X}$, $Com^\vtiny{Bi(B)BI}(\mathcal{X})$, is given by extending $Com^\vtiny{(B)BI}(\mathcal{X})$ with $\overline{U}$, together with $\blacktriangledown_{\mathcal{X}}$ and $\Cslash_{\mathcal{X}}$ defined 
\begin{align*}
A \blacktriangledown_{\mathcal{X}} B  &=  \{ x \mid \text{ for all } s, t, u, x \preccurlyeq s \in t \mathbin{\triangledown} u \text{ implies } t \in A \text{ or } u \in B \} \\
A \Cslash_{\mathcal{X}} B  &=  \{ x \mid \text{there exists } s, t, u \text{ s.t. } x \succcurlyeq s, u \in t \mathbin{\triangledown} s, u \in A \text{ and } t \not\in B \} \tag*{\qed}
\end{align*}
\end{defi}

\begin{lem}\label{lem:bibbicomplex}\hfill
\begin{enumerate}
\item Given a basic Bi(B)BI frame $\mathcal{X}$, $Com^\vtiny{Bi(B)BI}(\mathcal{X})$ is a basic Bi(B)BI algebra.
\item If $\mathcal{X}$ satisfies any Fig.~\ref{fig:algebrasubclassical} property, $Com^\vtiny{Bi(B)BI}(\mathcal{X})$ satisifies the corresponding axiom. \qed%
\end{enumerate}
\end{lem}

\begin{defi}[Prime Filter Bi(B)BI Frame]
Given a Bi(B)BI algebra $\mathbb{A}$, the prime filter frame of $\mathbb{A}$, $Pr^\vtiny{Bi(B)BI}(\mathbb{A})$ is given by extending $Pr^\vtiny{(B)BI}(\mathbb{A})$ with the operation $\mathbin{\triangledown}_{\mathbb{A}}$, defined 
\[ F \mathbin{\triangledown}_{\mathbb{A}} F' = \{ F'' \mid \forall a, b \in A: a \mor b \in F'' \text{ implies } a \in F \text{ or } b \in F' \} \]
and the set $U_{\mathbb{A}} = \{ F \mid \mbot \not\in F \}$.\qed%
\end{defi}

\begin{lem}\label{lem:bibbiprimefilter}\hfill
\begin{enumerate}
\item Given a basic Bi(B)BI algebra $\mathbb{A}$, $Pr^{Bi(B)BI}(\mathbb{A})$ is a basic Bi(B)BI frame. 
\item If $\mathbb{A}$ satisfies any axiom of Figure~\ref{fig:algebrasubclassical}, $Pr^{Bi(B)BI}(\mathbb{A})$ satisfies the corresponding frame property. 
\end{enumerate}
\end{lem}
\begin{proof}
We restrict ourselves to the non-trivial 2. We focus on the Weak Distributivity property for BiBI\@. Suppose $F_{t'} \supseteq F_t \in F_{x_1} \circ_{\mathbb{A}} F_{x_2}$ and $F_{t'} \subseteq F_{t''} \in F_{y_1} \mathbin{\triangledown}_{\mathbb{A}} F_{y_2}$. We show that
\[ \PrimePredicate{F}{F_{y_1} \in F_{x_1} \circ_{\mathbb{A}} F \text{ and } F_{x_2} \in F \mathbin{\triangledown}_{\mathbb{A}} F_{y_2}} \]
is a prime predicate. First suppose $P(F_\alpha) =1$ for all $\alpha$ in a $\subseteq$-chain ${(F_\alpha)}_{\alpha < \lambda}$. Then clearly $F_{y_1} \in F_{x_1} \circ_{\mathbb{A}} \bigcup_\alpha F_\alpha$. Suppose $a \mor b \in F_{x_2}$ and $b \not\in F_{y_2}$. Then necessarily $a \in F_\alpha$ for all $\alpha$, so $F_{x_2} \in \bigcup_\alpha F_\alpha \mathbin{\triangledown}_{\mathbb{A}} F_{y_2}$. Now let $P(F \cap F') = 1$. If $F_{x_2} \in (F \cap F') \mathbin{\triangledown}_{\mathbb{A}} F_{y_2}$ it follows that $F_{x_2} \in F \mathbin{\triangledown}_{\mathbb{A}} F_{y_2}$ and $F_{x_2} \in F' \mathbin{\triangledown}_{\mathbb{A}} F_{y_2}$, so assume $F_{y_1} \not\in F_{x_1} \circ_{\mathbb{A}} F, F_{x_1} \circ_{\mathbb{A}} F'$. Then there exists $a, a' \in F_{x_1}$, $b \in F$ and $b' \in F'$ such that $a * b, a' * b' \not\in F_{y_1}$. We have that $a'' = a \land a' \in F_{x_1}$ and $b \lor b' \in F \cap F'$ so $a'' * (b \lor b') = (a'' * b) \lor (a'' * b') \in F_{x_1}$. $F_{x_1}$ is prime so $a'' * b \in F_{x_1}$ or $a'' * b' \in F_{x_1}$. By monotonocity of $*$ and upwards-closure of $F_{x_1}$, $a * b \in F_{x_1}$ or $a' * b' \in F_{x_1}$, a contradiction. Hence either $P(F) =1$ or $P(F') =1$.

Now consider the set $F = \{ b \mid \exists y \not\in F_{y_2}(b \mor d \in F_{x_2})\}$. We prove $F$ is a proper filter. It is upwards-closed because $\mor$ is monotonic: if $b \in F$ and $b' \geq b$ we have $d \not\in F_{y_2}$ such that $b \mor d \in F_{x_2}$ and $b \mor d \leq b' \mor d \in F_{x_2}$. To see it is closed under meets, suppose $b, b' \in F$. Then there exist $d, d' \not\in F_{y_2}$ such that $b \mor d$, $b' \mor d' \in F_{x_2}$. $F_{y_2}$ is prime so $d \lor d' \not\in F_{y_2}$ and by montonocity of $\lor$, $b \mor (d \lor d'), b' \mor (d \lor d') \in F_{x_2}$. Let $d'' := d \lor d'$. By Proposition~\ref{prop:subclassicalprop}, $(b \land b') \mor d'' = (b \mor d'') \land (b' \mor d'') \in F_{x_2}$. Finally, to see that $F$ is proper, suppose $\bot \in F$. Then there exists $d \not\in F_{x_2}$ such that $\bot \mor d \in F_{x_2}$. Letting $a \in F_{x_1}$ be arbitrary, by Weak Distributivity and our assumption we have $a * (\bot \mor d) \leq (a * \bot) \mor d = \bot \mor d \in F_t \subseteq F_{t'}$. Thus $\bot \mor d \in F_{t'} \subseteq F_{t''}$ but $\bot \not\in F_{y_1}$ and $d \not\in F_{y_2}$, contradicting that $F_{t''} \in F_{y_1} \mathbin{\triangledown}_{\mathbb{A}} F_{y_2}$.

We finish the proof by showing that $P(F) =1$, yielding the existence of a prime $F_w$ satisfying the requirements of the frame property by the prime extension lemma. First let $a \in F_{x_1}$ and $b \in F$.    Then there exists $d \not\in F_{y_2}$ such that $b \mor d \in F_{x_2}$. By Weak Distributivity $a * (b \mor d) \leq (a * b) \mor d \in F_t \subseteq F_{t'} \subseteq F_{t''}$, and since $d \not\in F_{y_2}$ we necessarily have that $a * b \in F_{y_1}$. Now let $b \mor c \in F_{x_2}$ and suppose $c \not\in F_{y_2}$.  Then $b \in F$ by definition.
\end{proof}

%

\begin{thm}[Representation Theorem for Bi(B)BI + $\Sigma$ Algebras]\label{thm:bibbirep}
Every Bi(B)BI + $\Sigma$ algebra is isomorphic to a subalgebra of a complex algebra. Specifically, given an Bi(B)BI algebra $\mathbb{A}$, the map $\theta_{\mathbb{A}}: \mathbb{A} \rightarrow Com^\vtiny{Bi(B)BI}(Pr^\vtiny{(B)BI}(\mathbb{A}))$ defined $\theta_{\mathbb{A}}(a) = \{ F \in Pr(\mathbb{A}) \mid a \in F \}$ is an embedding. 
\end{thm}

\begin{proof}
The remaining verifications are that $\theta_{\mathbb{A}}$ respects $\mor, \rslash$ and $\mbot$. $\mbot$ follows straightforwardly from $\theta_{\mathbb{A}}(\mbot) = \overline{U_{\mathbb{A}}}$, and we verify $\mor$ leaving the similar $\rslash$ to the reader. We must show $\theta_{\mathbb{A}}(a \mor b) = \theta_{\mathbb{A}}(a) \blacktriangledown_{Pr^{\vtiny{Bi(B)BI}}(\mathbb{A})} \theta_{\mathbb{A}}(a)$. First suppose $a \mor b \in F$. Then $F \subseteq F_s \in F_t \triangledown_{\mathbb{A}} F_u$ means $a \mor b \in F_s$ and so either $a \in F_t$ or $b \in F_u$ as required. 

In the other direction, suppose $a \mor b \not\in F$. We show that \[ \PrimePredicate{I, I'}{F \in \overline{I} \mathbin{\triangledown}_{\mathbb{A}} \overline{I'}, a \in I \text{ and } b \in I'}\]
is a prime predicate for proper ideals $I, I'$. First suppose we have a $\subseteq$-chain ${(I_{\alpha}, I'_{\alpha})}_{\alpha}$ such that $P(I_{\alpha}, I'_{\alpha})=1$ for all $\alpha$. Clearly $a \in \bigcup_{\alpha} I_{\alpha}$ and $b \in \bigcup_{\alpha} I'_{\alpha}$. Suppose $c \mor d \in F$ with $c \not\in \overline{\bigcup_{\alpha} I_{\alpha}}$ and $d \not\in \overline{\bigcup_{\alpha} I'_{\alpha}}$. Then there exists $\beta, \beta'$ such that $c \in I_{\beta}$ and $d \in I'_{\beta'}$. By assumption we must have $c \in \overline{I_{\beta'}}$ and $d \in \overline{I'_{\beta}}$, and \emph{wolog} we may assume $\beta \leq \beta'$. Then, because $I_{\beta} \subseteq I_{\beta'}$ we have $c \in \overline{I_{\beta'}} \subseteq \overline{I_{\beta}}$, a contradiction.

Now suppose $P(I_0 \cap I_1, I')=1$. We have that $a \in I_0, I_1$ and $b \in I'$ so suppose both $P(I_0, I') = 0$ and $P(I_1, I') = 0$. Then there exists $c \mor d, c' \mor d' \in F$ such that $c \not\in \overline{I_0}$, $d \not\in \overline{I'}$, $c' \not\in \overline{I_1}$ and $d' \not\in \overline{I'}$. It folows that $d'' := d \lor d' \in I'$ and $c \land c' \in I_0 \cap I_1$. By upwards-closure and monotonicity of $\mor$, $c \mor d''$, $c' \mor d'' \in F$. Hence by Proposition~\ref{prop:subclassicalprop} $(c \land c') \mor d'' = (c \mor d'') \land (c' \mor d'') \in F$. However $c \land c' \not\in \overline{I_0 \cap I_1}$ and $d'' \not\in \overline{I'}$, contradicting that $F \in \overline{I_0 \cap I_1} \mathbin{\triangledown}_{\mathbb{A}} \overline{I'}$. Thus $P$ is a prime predicate.

Now consider the ideals $\ideal{a}$ and $\ideal{b}$. These must be proper as if $\top =a$ or $b$ then $a \mor b = \top \in F$, contradicting our assumption. We also have that for any $c \mor d \in F$, if $c \leq a$ and $d \leq b$ we have $c \mor d \leq a \mor b \in F$, a contradiction. Hence $F \in \overline{\ideal{a}} \mathbin{\triangledown}_{\mathbb{A}} \overline{\ideal{b}}$ and $P(\ideal{a}, \ideal{b}) =1$, yielding the necessary prime filters by taking the complements of the prime ideals guaranteed to exist by the prime extension lemma.
\end{proof}

\begin{cor}[Relational Soundness and Completeness]
For all formulas $\varphi$, $\psi$ of \logicfont{Bi(B)BI}, $\varphi \vdash \psi$ is provable in $\mathrm{Bi(B)BI +} \Sigma_{\mathrm{H}}$ iff $\varphi \vDash \psi$ in the relational semantics.  \qed%
\end{cor}

We note that the completeness result for the logics $\logicfont{BiBI} + \Sigma$ is new. Indeed, it would not be possible to adapt the argument Brotherston \& Villard give for \logicfont{BiBBI}~\cite{subclassical} directly as it relies on a translation into a Sahqlvist-axiomatized modal logic that uses Boolean negation in an essential way. There, the weak distributivity property is particularly difficult to deal with, requiring translation into a significantly more complicated frame property that is equivalent in the auxillary modal logic. In contrast, our proof is direct, and --- with the groundwork of \logicfont{(B)BI} duality done --- efficient.

These assignment lift to a functors in a way that is now standard, and it is straightforward to extend this to a dual equivalence given the appropriate definitions.

\begin{defi}[$BiBI _{\Sigma}$ Space]
Let $\Sigma$ be a set of subclassical bunched logic axioms. A $BiBI_{\Sigma}$ \emph{space} is a structure $\mathcal{X} = (X, \mathcal{O}, \succcurlyeq, \circ, E, \triangledown, U)$ such that
\begin{enumerate}
\item $(X, \mathcal{O}, \succcurlyeq, \circ, E)$ is a BI space,
\item $(X, \succcurlyeq, \circ, E, \triangledown, U)$ is a basic BiBI frame satisfying the frame correspondents of $\Sigma$;
\item The upwards-closed clopen sets of $(X, \mathcal{O}, \succcurlyeq)$ are closed under $\blacktriangledown_{\mathcal{X}}$ and $\Cslash_{\mathcal{X}}$,
\item $U$ is clopen; and
\item If $x \not\in y \triangledown z$ then there exists upwards-closed clopen sets $C_1, C_2$ such that $y \not\in C_1$, $z \not\in C_2$ and $x \in C_1 \blacktriangledown_{\mathcal{X}} C_2$.
\end{enumerate}
A $BiBBI_{\Sigma}$ \emph{space} is a $BiBI_{\Sigma}$ space for which $\succcurlyeq$ is equality. \qed%
\end{defi}

\begin{thm}[Duality Theorem for \logicfont{Bi(B)BI}]
For any set of subclassical axioms $\Sigma$, $\theta$ and $\eta$ form a dual equivalence of categories between $\mathrm{BiBIAlg}_{\Sigma}$ and $\mathrm{BiBISp}_{\Sigma}$. \qed%
\end{thm}

\subsection{Concurrent Kleene Algebra and Concurrent Separation Logic}\label{subsec:ckbi}

We conclude with a tentative application of our framework to Concurrent Separation Logic (\logicfont{CSL})~\cite{concur2, concur1}. Without an account of the semantics of the programming language it is not immediately obvious how our duality theoretic approach can be extended to \logicfont{CSL}, which requires strictly more structure than just the heap model of \logicfont{FOBBI}. In any case, algebraic models of a basic version of \logicfont{CSL}, $\logicfont{ASL}^{--}$, have been given that connect the logic to concurrent Kleene algebra~\cite{exchange}. The proof rules for $\logicfont{ASL}^{--}$ are given in Figure~\ref{fig:ASL}.

\begin{figure}
\hrule
    \vspace{1mm}
\centering
\setlength\tabcolsep{6pt}
\setlength\extrarowheight{15pt}
\begin{tabular}{lclc}
Frame:
&
\AxiomC{$\{p\}c\{q\}$}
\UnaryInfC{$\{p * r \} c \{q * r \}$}
\DisplayProof%
&
Concurrency:
&
\AxiomC{$\{p_1\}c_1 \{q_1\}$}
\AxiomC{$\{p_2\}c_2 \{q_2\}$}
\BinaryInfC{$\{p_1 * p_2 \} c_1 \parallel c_2 \{q_1 * q_2 \}$}
\DisplayProof%
\\
Skip:
&
\AxiomC{}
\UnaryInfC{$\{p\}\mathrm{skip}\{ p \}$}
\DisplayProof%
&
Seq:
&
\AxiomC{$\{p\}c_1\{q\}$}
\AxiomC{$\{q\}c_2 \{r\}$}
\BinaryInfC{$\{p\} c_1; c_2 \{r\}$}
\DisplayProof%
\\
NonDet:
&
\AxiomC{$\{p\}c_1\{q\}$}
\AxiomC{$\{p\}c_2\{q\}$}
\BinaryInfC{$\{ p \} c_1 + c_2 \{q\}$}
\DisplayProof%
&
Iterate:
&
\AxiomC{$\{p\}c\{ p \}$}
\UnaryInfC{$\{p\}\mathrm{Iterate}(c)\{p \}$}
\DisplayProof%
\\
Disjunction:
&
\AxiomC{$\{p_i\}c\{q\}, \text{ all } i \in I$}
\UnaryInfC{$\{\bigvee_{i \in I} p\} c \{ q \} $}
\DisplayProof%
&
Consequence:
&
\AxiomC{$p \leq p'$}
\AxiomC{$\{p\} c \{q\}$}
\AxiomC{$q \leq q'$}
\TrinaryInfC{$\{p'\}c\{q'\}$}
\DisplayProof%

\end{tabular}
\caption{Rules for $\logicfont{ASL}^{--}$.}
\vspace{1mm}
\hrule%
\label{fig:ASL}
\end{figure}

\begin{defi}[Concurrent Kleene Algebra (cf.~\cite{exchange})]\hfill
\begin{enumerate}
\item A \emph{concurrent monoid} $(M, \leq, \parallel, ; , \mathrm{skip})$ is a partial order $(M, \leq)$, together with two monoids $(M, \parallel, \mathrm{skip})$ (with $\parallel$ commutative) and $(M, ; , \mathrm{skip})$ satisfying the \emph{exchange law}	\[ (p \parallel r ) ; (q \parallel s) \leq (p ; q ) \parallel (r ; s ). \]
It is \emph{complete} if $(M, \leq)$ is a complete lattice.
\item A \emph{concurrent Kleene algebra} (CKA) is a complete concurrent monoid where $\parallel$ and $;$ preserve joins in both arguments. 
\item A \emph{weak CKA} is a complete concurrent monoid together with a subset $A \subseteq M$ (the assertions of the algebra) such that
    \begin{inparaenum}[i)] 
        \item $\mathrm{skip} \in A$;
        \item $A$ is closed under $\parallel$ and all joins;
        \item $\parallel$ restricted to $A$ preserves all joins in both arguments;
        \item for each $a \in A$, $a;(-): M \rightarrow M$ preserves all joins; and
        \item for each $m \in M$, $(-);m: A \rightarrow M$ preserves all joins.
    \end{inparaenum}
\item A CKA or weak CKA is \emph{Boolean} if the underlying lattice is a Boolean algebra and \emph{intuitionistic} if the underlying lattice is a Heyting algebra. \qed%
\end{enumerate}
\end{defi}

\noindent
In concurrent Kleene algebra, $p \parallel q$ is interpreted as giving the parallel execution of programs $p$ and $q$ while $p\mathbin{;}q$ is interpreted as giving the sequential execution $p$, then $q$.  One of the key aspects of this definition is the exchange law, which enforces a liberal interpretation of interleaving that states that a program that runs $p$ and $r$ in parallel, followed by $q$ and $s$ in parallel can be implemented as a program that runs $p$ then $q$ in parallel to $r$ then $s$. O'Hearn et al.\ show that $\logicfont{ASL}^{--}$ is sound and complete for weak CKAs when Hoare triples $\{p\}c \{q\}$ are interpreted as inequalities $p;c \leq q$ (where $p, q \in A$ and $c \in M$) and $*$ is interpreted as $\parallel$ restricted to $A$. This is achieved via the construction of a predicate transformer model over $\logicfont{ASL}^{--}$ propositions. They also show that a trace model of $\logicfont{ASL}^{--}$ generates a Boolean CKA\@.

Elsewhere, O'Hearn~\cite{ohearnslides} suggests that the structures involved could be used as inspiration for a bunched logic extending \logicfont{BBI}. We define such a logic and call it \emph{Concurrent Kleene BI} or \logicfont{CKBI}. We leave the apparent intuitionistic variant extending \logicfont{BI} to another occasion. 

Let $\mathrm{Prop}$ be a set of atomic propositions, ranged over by $\mathrm{p}$. The set of all formulae of the concurrent Kleene bunched logic $\mathrm{Form}_{CKBI}$ is generated by the grammar
\[ \varphi ::= \mathrm{p} \mid \top \mid \bot \mid \mtop \mid \varphi \land \varphi \mid \varphi \lor \varphi \mid \varphi \rightarrow \varphi \mid \varphi * \varphi \mid \varphi \mathbin{;} \varphi \mid \varphi \wand \varphi \mid \varphi \rseq \varphi \mid \varphi \lseq \varphi \] where additive negation is defined by $\neg \varphi := \varphi \rightarrow \bot$.

Figure~\ref{fig:hilbert-ckbi} gives rules that can be added to the system for \logicfont{BBI} to obtain a system for \logicfont{CKBI}. These rules essentially dictate that the $(*, \wand)$-free fragment of \logicfont{CKBI} is non-commutative \logicfont{BBI}. \logicfont{CKBI} is interpreted on structures extending BBI frames called CKBI frames by extending \logicfont{BBI}'s semantics by the clauses in Fig.~\ref{fig:sat-ckbi}.

\begin{defi}[CKBI Frame]
A \emph{CKBI frame} is a structure $\mathcal{X} = (X, \circ, E, \triangleright)$ such that $(X, \circ, E)$ is a BBI frame and $\triangleright: X^2 \rightarrow \mathcal{P}(X)$ a binary operation satisfying (with outermost quantification omitted for readability):

\vspace{0.5em}
\adjustbox{raise=1.4cm}{
\begin{tabular}{llr}
(Unit Existence\textsubscript{L}) & $\exists e \in E(x \in e \triangleright x)$; \\
(Unit Existence\textsubscript{R}) & $\exists e \in E(x \in x \triangleright e)$; \\
(Coherence\textsubscript{L}) & $e \in E \land x \in e \triangleright y \rightarrow x = y$; \\
(Coherence\textsubscript{R}) & $e \in E \land x \in y \triangleright e \rightarrow x = y$; \\
(Associativity) & $\exists t(t \in x \triangleright y \land w \in t \triangleright z) \leftrightarrow \exists t'(t' \in y \triangleright z \land w \in x \triangleright t')$ \\
\text{(Exchange)} & $t \in w \circ y \land s \in x \circ z \land u \in t \triangleright s \rightarrow$ \\
& $\exists r, v(r \in w \triangleright x \land v \in y \triangleright z \land u \in r \circ u)$
\end{tabular}
}
\qed%
\end{defi}

The traces model of $\logicfont{ASL}^{--}$ can be seen as (the complex algebra of) a CKBI frame, where $\circ$ is interleaving, $\triangleright$ is concatenation and $E$ is the singleton set containing the empty trace. Another example is given by pomsets~\cite{pomset}, with $\circ$ given by the parallel pomset composition, $\triangleright$ the series pomset composition, and $E$ the singleton set containg the empty pomset.

\begin{figure}
\centering
\hrule
\setlength\tabcolsep{4pt}
\setlength\extrarowheight{15pt}
\scalebox{0.87}{
\begin{tabular}{lclclcl}
25. & $\cfrac{\xi\vdash\varphi\quad\eta\vdash\psi}{\xi\mathbin{;}\eta\vdash\varphi\mathbin{;}\psi}$  &
26. & $\cfrac{\eta\mathbin{;}\varphi\vdash\psi}{\eta\vdash\varphi\rseq\psi}$ & 27. & $\cfrac{\xi\vdash\varphi\rseq\psi\quad\eta\vdash\varphi}{\xi\mathbin{;}\eta\vdash\psi}$ \\
28. & $\cfrac{\eta\mathbin{;}\varphi\vdash\psi}{\varphi\vdash\eta\lseq\psi}$ & 29. & $\cfrac{\xi\vdash\varphi\lseq\psi\quad\eta\vdash\varphi}{\eta\mathbin{;}\xi\vdash\psi}$ & 30. & $\cfrac{}{\mtop \mathbin{;} \varphi \dashv \vdash \varphi}$  \\
31. & $\cfrac{}{\varphi \mathbin{;} \mtop \dashv \vdash \varphi}$ & 34. & $\cfrac{}{\varphi \mathbin{;} (\psi \mathbin{;} \chi) \dashv \vdash (\varphi \mathbin{;} \psi) \mathbin{;} \chi}$ & 35. & $\cfrac{}{(\varphi * \psi) \mathbin{;} (\chi * \xi) \vdash (\varphi \mathbin{;} \chi) * (\psi \mathbin{;} \xi)}$
\end{tabular}}
\caption{Hilbert rules for concurrent Kleene bunched logic.}
\vspace{1mm}
\hrule%
  \label{fig:hilbert-ckbi}
\end{figure}

\begin{figure}
\centering
\hrule
\vspace{1mm}
 \setlength\tabcolsep{3pt}
\setlength\extrarowheight{4pt}
\begin{tabular}{c c c c l r c c c c r r c}
$x$ & $\vDash_{\Valuation}$ & $\varphi \mathbin{;} \psi$ & iff & \multicolumn{8}{l}{there exists $y, z$ s.t. $x \in y \triangleright z$, $y \vDash_{\Valuation} \varphi$ and  $z \vDash_{\Valuation} \psi$} \\
$x$ & $\vDash_{\Valuation}$ & $\varphi \rseq \psi$ & iff & \multicolumn{8}{l}{for all $y, z$ s.t. $z \in x \triangleright y$: $y \vDash_{\Valuation} \varphi$ implies $z \vDash_{\Valuation} \psi$} \\
$x$ & $\vDash_{\Valuation}$ & $\varphi \lseq \psi$ & iff & \multicolumn{8}{l}{for all $y, z$ s.t. $z \in y \circ x$:  $y \vDash_{\Valuation} \varphi$ implies $z \vDash_{\Valuation} \psi$}
\end{tabular}
\caption{Satisfaction for concurrent Kleene bunched logic.}
\vspace{1mm}
\hrule%
\label{fig:sat-ckbi}
\end{figure}

\begin{defi}[CKBI Algebra]
A \emph{CKBI algebra} is an algebra \[\mathbb{A} = (A, \land, \lor, \rightarrow, \top, \bot, *, \wand, \mtop, \mseq, \rseq, \lseq)\] such that $(A, \land, \lor, \rightarrow, \top, \bot, *, \wand, \mtop)$ is a BBI algebra and $(A, \mseq, \mtop)$ a monoid, satisfying, for all $a, b, c, d \in \mathbb{A}$,
\begin{enumerate}
\item $a \mseq b \leq c \text{ iff } a \leq b \rseq c \text{ iff } b \leq a \lseq c$, and
\item Exchange: $(a * b)\mseq(c * d) \leq (a \mseq c) * (b \mseq d)$. \qed%
\end{enumerate}
\end{defi}

\noindent
In effect, a CKBI algebra is a BBI algebra in which there are \emph{two} coexisting monoidal residuated structures sharing a unit: one commutative (corresponding to concurrent execution) and one non-commutative (corresponding to sequential execution). As such the corresponding properties of Proposition~\ref{prop-alg-prop} hold for $\mseq, \rseq$ and $\lseq$. In the terminology of O'Hearn et al.~\cite{exchange}, a CKBI algebra is a Boolean CKA extended with the residuals corresponding to $*$ and $\mseq$. The evident algebraic soundness and completeness theorem for these algebras is easily proved.

It is straightforward to adapt the results of Section~\ref{sec:bunchedimplications} to these structures. Let $\mathbb{A}$ be a CKBI algebra. Then the prime filter frame of $\mathbb{A}$, $Pr^\vtiny{CKBI}(\mathbb{A})$, is given by extending the prime filter frame of the underlying BBI algebra with the operation $\triangleright_{\mathbb{A}}$, defined 
\[
    F \triangleright_{\mathbb{A}} F' = \{ F'' \mid \forall a \in F, \forall b \in F': a \mseq b \in F'' \}.
\]
In the other direction, the complex algebra of a CKBI frame $\mathcal{X}$, $Com^\vtiny{CKBI}(\mathcal{X})$, is given by extending the complex algebra of the underlying BBI frame with the operation
\[
    A \mseq_{\mathcal{X}}B = \{ z \mid \exists x\in A, y \in B(z \in x \triangleright y)\}
\]
and its associated adjoints. The respective results for CKBI algebras follow straightforwardly from the case for BBI\@: the key remaining step is the correspondence between the algebraic Exchange axiom and the frame property Exchange. 

\begin{lem}\hfill
\begin{enumerate}
\item Given a CKBI algebra $\mathbb{A}$, the prime filter frame $Pr^\vtiny{CKBI}(\mathbb{A})$ is a CKBI frame. 
\item Given a CKBI frame $\mathcal{X}$, the complex algebra $Com^\vtiny{CKBI}(\mathcal{X})$ is a CKBI algebra.
\end{enumerate}
\end{lem}
\begin{proof}
We focus on the correspondence between the Exchange properties of the respective structures.
\begin{enumerate}
\item Suppose we have prime filters of $\mathbb{A}$ satisfying $F_{wy} \in F_w \circ_{\mathbb{A}} F_y, F_{xz} \in F_x \circ_{\mathbb{A}} F_z$ and $F_t \in F_{wy} \triangleright_{\mathbb{A}} F_{xz}$. Using similar arguments to those given in previous results, it can be seen that \[ \PrimePredicate{F, G}{F \in F_w \triangleright_{\mathbb{A}} F_x, G \in F_y \triangleright_{\mathbb{A}} F_z \text{ and } F_t \in F \circ_{\mathbb{A}} G} \]
is a prime predicate on proper filters $F$ and $G$. Consider the sets $F = \{ c \mid \exists a \in F_w, b \in F_x(a \mseq b \leq c) \}$ and $G = \{ c \mid \exists a \in F_y, b \in F_z(a \mseq b \leq c) \}$. Both sets are obviously upwards-closed and closed under meets as monotonicity of $\mseq$ gives that $a\mseq b \leq c$ and $a' \mseq b' \leq c'$ implies $(a \land a') \mseq (b \land b') \leq c \land c'$. Hence $F$ and $G$ are filters. They are also proper: suppose for contradiction that $\bot \in F$. Then there exists $a \in F_w$ and $b \in F_x$ such that $a \mseq b =\bot$. Let $c \in F_y$ and $d \in F_z$ be arbitrary. By assumption we have that $a * c \in F_{wy}$, $b * d \in F_{xz}$ and so $(a * c)\mseq (b * d) \in F_t$. By Exchange and upwards-closure of filters, $(a \mseq b) * (c \mseq d) = \bot * (c \mseq d) = \bot \in F_t$, a contradiction. The same argument suffices to show $G$ is proper.

Clearly $F \in F_w \triangleright_{\mathbb{A}} F_x$ and $G \in F_y \triangleright_{\mathbb{A}} F_z$. Further, $F_t \in F \circ_{\mathbb{A}} G$: let $c \geq a \mseq b$ and $c' \geq a' \mseq b'$ for $a \in F_w, b \in F_x, a' \in F_y$ and $b' \in F_z$. By monotonicity of $*$ and Exchange, $(a * a') \mseq (b * b') \leq (a \mseq b) * (a' \mseq b') \leq c * c'$. It then follows that $c * c' \in F_t$, since by assumption $a * a' \in F_{wy}$ and $b * b' \in F_{xz}$, so $(a * a')\mathbin{;} (b * b') \in F_t$. Hence by the prime extension lemma there exist prime $F$ and $G$ satisfying these properties, and so the frame property Exchange is satisfied on  $Pr^\vtiny{CKBI}(\mathbb{A})$. 

\item Suppose $t \in (A \bullet_{\mathcal{X}} C) \mseq_{\mathcal{X}} (B \bullet_{\mathcal{X}} D)$. Then there exist $w, x, y, z, wy, xz$ such that $wy \in w \circ y, xz \in x \circ z$ and $t \in wy \triangleright xz$. The frame property Exchange then ensures there are witnesses to the fact that $t \in (A \mseq_{\mathcal{X}} B) \bullet_{\mathcal{X}} (C \mseq_{\mathcal{X}} D)$. \qedhere
\end{enumerate} 
\end{proof}

\noindent
We immediately obtain the following representation theorem from the representation theorem for BBI algebras.

\begin{thm}[Representation Theorem for CKBI Algebras]
Every CKBI algebra is isomorphic to a subalgebra of a complex algebra. Specifically, given a CKBI algebra $\mathbb{A}$, the map $\theta_{\mathbb{A}}: \mathbb{A} \rightarrow Com^\vtiny{CKBI}(Pr^\vtiny{CKBI}(\mathbb{A}))$ defined $\theta_{\mathbb{A}}(a) = \{ F \in Pr^\vtiny{CKBI}(\mathbb{A}) \mid a \in F \}$ is an embedding.  \qed
\end{thm}


\begin{cor}[Relational Soundness and Completeness]
For all formulas $\varphi$, $\psi$ of \logicfont{CKBI}, $\varphi \vdash \psi$ is provable in $\mathrm{CKBI}_{\mathrm{H}}$ iff $\varphi \vDash \psi$ in the relational semantics. \qed%
\end{cor}

The extension to a dual equivalence of categories is now a simple task given the preceeding material: we leave it to the reader as an exercise.We remain agnostic about the extent to which \logicfont{CKBI} can be used as a simplified version of \logicfont{CSL}. The fact that it essentially supplies a Kripke semantics formulation of CKAs suggests that it may have uses as a logic for reasoning about concurrency more generally. We defer a thorough investigation of these ideas to another occasion.

\section{Conclusions and Further Work}\label{sec:conclusions}

We have given a systematic treatment of Stone-type duality for
the structures that interpret bunched logics, starting with the weakest systems,
recovering the familiar \logicfont{BI} and \logicfont{BBI}, and concluding with both the classical and intuitionistic variants of Separation Logic. Our results
encompass all the known existing algebraic approaches to Separation
Logic and prove them sound with respect to the standard store-heap semantics.
As corollaries, we uniformly recover soundness and completeness
theorems for the systems we consider. These results are extended to the bunched logics with additional modalities and multiplicatives corresponding to negation, disjunction and falsum --- \logicfont{DMBI}, \logicfont{CBI}  and the full range of sub-classical bunched logics --- as well as
\logicfont{CKBI}, a new logic inspired by the algebraic structures that interpret (a basic version of) Concurrent Separation Logic.  
We believe this treatment will simplify completeness arguments for future bunched logics by providing a modular framework within which existing results can be extended. This is demonstrated with our results on \logicfont{SML}, \logicfont{DMBI}, \logicfont{BiBI} and \logicfont{CKBI} which are all new to the literature. More generally, the notion of indexed frame and its associated completeness argument can easily be adapted for a wide range of non-classical predicate logics.

These results are exhaustive of the bunched logics that can be found in the literature, with the exception of Brotherston \& Villard's hybrid extension of \logicfont{BBI}, \logicfont{HyBBI}. Capturing this logic would require new and more sophisticated techniques to handle nominals and satisfaction operators in the setting of substructural connectives and we defer this investigation to another occassion. In particular, in the present work it is shown how complex algebra operations, morphisms and topological coherence conditions for the Kripke frame structure that corresponds to multiplicative variants of each of the propositional connectives should be handled, as well as the correspondence between the defining axioms of those existing bunched logics and the frame properties of the Kripke models that interpret them. While we do not give a characterisation of precisely which axiomatic extensions to bunched logics automatically yield a sound and complete Kripke semantics (in the sense of Sahlqvist's theorem~\cite{Sahlqvist1975} for modal logic), these results are a first step towards that goal: for example, through an adaptation of Sambin \& Vaccaro's~\cite{Sambin1989} topological proof of Sahlqvist's theorem, which utilises modal duality theory in an essential way. Although there is great modularity and uniformity in the present work --- each logic's results are able to be verified on just the structure new to that logic and the construction of the correct prime predicate uniformly facilitates each argument --- this can clearly be pushed further. This work is a first step towards (and a necessary foundation for) the investigation of this issue.  

We identify four areas of interest for further work. The first is applications of the duality theory to the metatheory of bunched logics. Preliminary work in this direction appears in the PhD thesis of the first author~\cite{DochertyThesis}, and includes a Goldblatt-Thomason style theorem that characterises the classes of bunched logic model that are definable by bunched logic formulae as well as the negative resolution of the open problem of Craig interpolation for bunched logics. An obvious further result --- which is in some sense prefigured by the results in this paper --- is the identification of a Sahlqvist-like fragment of bunched logic axioms which are guaranteed by their syntactic shape to produce sound and complete bunched logics, in the sense described above. 

Second, extending our approach to account for the operational semantics of program execution given by Hoare triples.  As a consequence, we aim
to interpret computational approaches to the Frame Rule such
as bi-abduction~\cite{bi-abduction} within our semantics and investigate if algebraic or topological methods can be brought to bear on these important aspects for implementations of Separation Logic. We believe the evident extension of our framework with duality-theoretic approaches to Hoare logic such as those of Abramsky~\cite{abramsky} or Brink \& Rewitzky's~\cite{Brink2001} could facilitate this. 
We also believe connections can be made with the monadic approach to predicate transformers and Hoare logic described by Jacobs~\cite{hoaremonad}. That work is suggestive of the existence of a Separation Logic monad operating on categories of semantic structures for (\logicfont{FO})(\logicfont{B})\logicfont{BI}.
A related approach that combines these two suggestions would be to use the dualities of the present work to generate state-and effect triangles for Separation Logic, which would enable us to investigate healthiness for predicate transformers associated with pointer manipulating programs, in a similar fashion to the work of Hino et al.~\cite{HKHJ16}.

A third area of investigation concerns the interfacing of resource semantics with coalgebraic models of stateful systems. Previous work has given a sound and complete coalgebraic semantics for \logicfont{(B)BI}~\cite{DP15} in line with coalgebraic generalizations of modal logic, but we are interested in another direction: the use of coalgebra as a mathematical foundation for transition systems. Separation Logic has shown the power of resource semantics for modelling real world phenomena when extended with suitable dynamics. Another approach utilizing process algebras generated by resource semantics has shown one way in which this idea can be applied to a more general class of distributed systems~\cite{AP2016, CP09}. We believe an analysis of coalgebras definable over the category of (B)BI frames would provide a general mathematical foundation for both these instances and more. A well-developed line of work in coalgebraic logic has produced general machinery for generating sound and complete coalgebraic logics from dual adjunctions under favourable conditions~\cite{JS10,KP2011}. We thus believe such an analysis, combined with the dual adjunctions given in the present work, would yield a powerful framework for modelling, specifying and reasoning about resource-sensitive transition systems, with obvious wide applicability.

A fourth line of research would be the expansion of the results here to give a general treatment of categorical structures for non-classical predicate logics. The results given on (B)BI hyperdoctrines straightforwardly apply to any hyperdoctrine with target algebras that have a duality theorem and include the Heyting or Boolean connectives. To what extent can this treatment generalize existing semantic approaches to non-classical predicate logics that fit this criteria? For example, the various categorical semantics given for predicate modal logic on sheaves~\cite{AK12}, metaframes~\cite{S98} and modal hyperdoctrines~\cite{BG2006}. Our results could also be generalized in two ways. First, for hyperdoctrines with weaker-than-Heyting target algebras and their corresponding dual indexed frames, allowing us to encompass predicate substructural logics, predicate relevant logics and predicate positive logics. Second, to more exotic notions of quantification. As one example, further extending our framework to encompass the breadth of the bunched logic literature would require an account of multiplicative quantification, an area which has only partially been explored algebraically. Collinson et al.~\cite{polymorphism} give hyperdoctrines defined on monoidal categories to give semantics to a bunched polymorphic lambda calculus.  There, the right adjoint of $\mathbb{P}(\pi_{\Gamma, X})$ must satisfy a compatibility condition with the monoidal structure in order to interpret multiplicative universal quantification in the calculus, and --- with some subtle technical tweaks and additional structure --- it is possible to give left adjoints that can interpret multiplicative existential quantification. Is there a unification of this approach with the present work that adapts these ideas to the more general logical setting? We believe our framework provides the mathematical foundation to explore these ideas.

\bibliographystyle{plain}

\end{document}